\documentclass[11pt,oneside,letterpaper]{article}
\usepackage{amsmath}
\usepackage{amsfonts}
\usepackage{amssymb}
\usepackage{amsthm}
\usepackage{graphicx}
\usepackage{caption}
\usepackage{subcaption}
\usepackage{rotating}
\usepackage{multirow}
\usepackage{color}
\usepackage{rotfloat}
\usepackage{pgfplots}
\usepackage[T1]{fontenc}
\usepackage{amssymb,amsmath}
\usepackage{pdflscape}
\usepackage[left=1in, bottom=1in, right=1in, top=1in]{geometry}
\usepackage{times}
\usepackage{bbm, dsfont}
\usepackage{natbib}
\usepackage{cases}
\usepackage{url}
\usepackage{xr}
\usepackage{etoolbox}
\apptocmd{\UrlBreaks}{\do\-}{}{}
\bibpunct{(}{)}{;}{a}{,}{,}
\usepackage{float}
\usepackage{pdfsync}
\usepackage{setspace}
\usepackage{enumerate}
\numberwithin{equation}{section}
\usepackage[shortlabels]{enumitem}
\theoremstyle{definition}\newtheorem{definition}{Definition}
\newtheorem{lemma}{Lemma}[section]
\newtheorem{proposition}{Proposition}[section]
\newtheorem{theorem}{Theorem}
\newtheorem{assumption}{Assumption}

\theoremstyle{definition}\newtheorem{example}{Example}
\theoremstyle{definition}
\theoremstyle{definition}\newtheorem{remark}{Remark}[section]
\theoremstyle{definition}

\DeclareMathOperator*{\pconv}{\overset{P}{\longrightarrow}}
\DeclareMathOperator*{\dconv}{\overset{d}{\longrightarrow}}
\DeclareMathOperator*{\aequal}{\overset{a}{=}}
\title{Testing for Restricted Stochastic Dominance under Survey Nonresponse with Panel Data: Theory and an Evaluation of Poverty in Australia}
\author{Matthew J. Elias\thanks{e61 Institute, Sydney, Australia.}\and Rami V. Tabri\thanks{\ Corresponding author. Department of Econometrics and Business Statistics,
Monash University, Melbourne, Australia; Email: rami.tabri@monash.edu}}
\begin{document}
\maketitle
\begin{abstract}
This paper lays the groundwork for a unifying approach to stochastic dominance testing under survey nonresponse that integrates the partial identification approach to incomplete data and design-based inference for complex survey data. We propose a novel inference procedure for restricted $s$th-order stochastic dominance, tailored to accommodate a broad spectrum of nonresponse assumptions. The method uses pseudo-empirical likelihood to formulate the test statistic and compares it to a critical value from the chi-squared distribution with one degree of freedom. We detail the procedure's asymptotic properties under both null and alternative hypotheses, establishing its uniform validity under the null and consistency against various alternatives. Using the Household, Income and Labour Dynamics in Australia survey, we demonstrate the procedure's utility in a sensitivity analysis of temporal poverty comparisons among Australian households.
\end{abstract}

\section{Introduction}\label{Section - Intro}
\par This paper connects three bodies of literature:stochastic dominance testing, partial identification for incomplete data, and design-based inference for complex survey data (e.g.~involving stages, clustering, and stratification). Although there is extensive literature on stochastic dominance testing (e.g., \citealp{mcfadden1989,Abadie2002,barrettdonald2003,Linton-Maasoumi-Whang,Linton2010,davidsonduclos2013,Donald-Hsu,LOK-Tabri2021}, among others), most methods fall short in practical applications because they assume complete datasets and rely on either random sampling or time series data processes. In empirical applications, data often come from complex socioeconomic surveys like the Current Population Survey, the Household, Income and Labour Dynamics in Australia (HILDA) Survey, and the Survey of Labour Income Dynamics, where nonresponse is common and data are not simple random samples. While the partial identification literature offers bounding approaches to address missing data issues (e.g., \citealp{blundell2007,klinesantos2013,manski2016}, among others), these often overlook the intricate designs of surveys, potentially leading to incorrect standard error estimations and subsequent distortions in tests' size and power. 

\par Conversely, the design-based inference literature accommodates survey design complexities but generally depends on designers' handling of nonresponse who either posit point-identifying assumptions on nonresponse or specify a model of nonrandom missing data with reweighting/imputation (e.g., \citealp{berger2020empirical,seker2015poverty,Chen-Duclos,Qin-Zhang-Leung2009}, among others). Such assumptions may not hold when nonrespondents possess key characteristics influencing survey outcomes, such as income in socioeconomic surveys where the rich are often underrepresented (\citealp{Bourguignon-2018}), rendering these inferences potentially misleading and not credible. It is also noteworthy that econometricians have considered design-based perspectives to inference in other contexts. See, for example,~\cite{BHATTACHARYA2005145}, who developed tests for Lorenz dominance, and~\cite{Abadie-Athey-Imbens-Wooldridge}, who focus on causal estimands in a regression framework. At the intersection of the three literatures is the testing procedure of~\cite{Tabri2021}. It is, however, limited in scope because it applies only to ordinal data generated from independent cross-sections and uses the worst-case bounds to account for nonresponse.

\par This paper bridges those bodies of literature. It provides the first asymptotic framework and inference procedure 
that integrates design-based inference with partial identification for handling incomplete data in econometrics, specifically for restricted stochastic dominance. Our approach effectively handles the survey’s complex design and the missing outcome data due to nonresponse. Importantly, it accommodates a broad spectrum of assumptions on nonresponse, enabling researchers to transparently perform sensitivity analyses of test conclusions by clearly linking them to various nonresponse scenarios.

\par To illustrate our approach, consider a data example on analyzing equivalized household net income (EHNI) data from the HILDA survey for 2001 (wave 1) and 2002 (wave 2). The objective is to detect a decline in poverty using the \emph{headcount ratio} across a given range of poverty lines, $[\underline{t}, \overline{t}]$, by seeking evidence for restricted stochastic dominance: $F_2(y) < F_1(y)$ for all  $y$ in this interval, where $F_j$ represents the cumulative distribution function (CDF) for EHNI in wave  $j$. Strong evidence for this decline is usually sought through a test of $H_0:$ Not $H_1$ against $H_1: F_2(y) < F_1(y)\,\forall y\in [\underline{t}, \overline{t}]$ and rejecting $H_0$ in favor of $H_1$ (\citealp{davidsonduclos2013}). A significant challenge with implementing this test in practice is that the achieved sample is incomplete:
\begin{align*}
    Y_i = \begin{cases}
    (Y^{1}_i, Y^{2}_i) & \text{if } Z^1_i = Z^2_i = 1, \\
    (Y^1_i, *) & \text{if } Z^1_i = 1, Z^2_i = 0, \\
    (*, *) & \text{if } Z^1_i = Z^2_i = 0,
    \end{cases}
    \quad \text{for}\quad i=1,\ldots,n,
\end{align*}
where "*" indicates missing data, \( Y^{j}_i \) and \( Z^{j}_i \) are the EHNI and response indicators for wave \( j \) respectively, expressed in 2001 Australian dollars. Nonrespondents in wave 1 are not followed up, equating nonresponse in wave 1 to unit nonresponse and in wave 2 to wave nonresponse. The unit nonresponse rate is approximately 33\%, while for wave 2, it is about 8\%; item nonresponse is disregarded for simplicity.

\par Consequently, without any assumptions on nonresponse the EHNI CDFs are only partially identified. The first panel in Figure~\ref{Figure: Intro} reports the survey-weighted estimates of the worst-case identified sets of these CDFs over a given range of poverty lines, assuming that the theoretically largest and lowest values of EHNI are 1000000 and -150000. Observe that while these sets are wide, they are maximally credible in that they reflect the inherent uncertainties that cannot be eliminated through less plausible assumptions. However, since one set is a proper subset of the other, they are uninformative for comparing $F_1$ and $F_2$ in terms of restricted stochastic dominance.
\begin{figure}[pt]
\centering
\includegraphics[width=15cm, height=6cm]{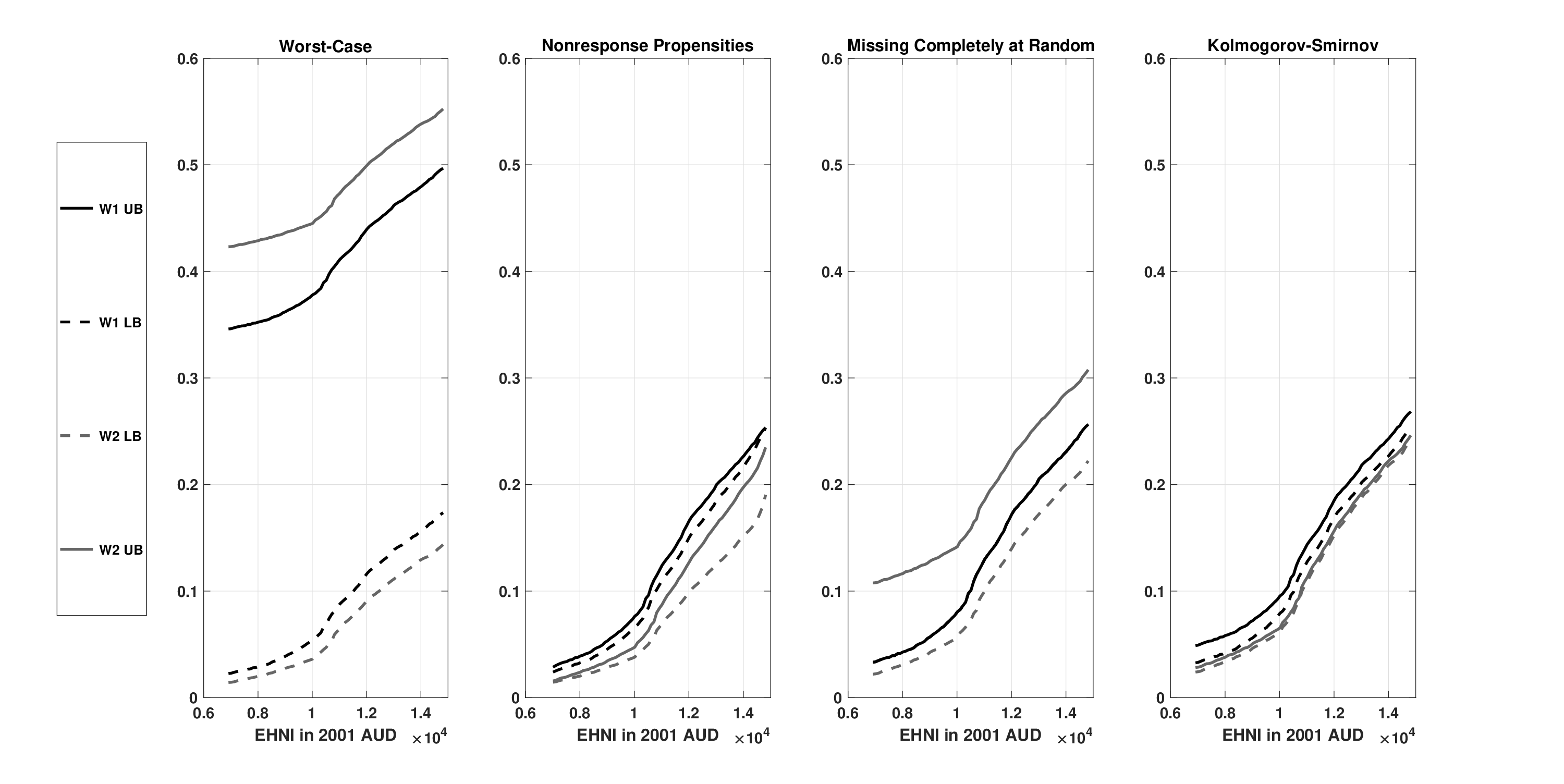}
\caption{Survey-weighted estimates: identified sets of $F_1$ and $F_2$. }\label{Figure: Intro}
\end{figure}

\par An alternative approach involves considering a range of assumptions on nonresponse, in the vast middle ground between the worst-case scenario (no assumptions) and the standard practice of assuming random nonresponse to point-identify $F_1$ and $F_2$. This method allows consumers of economic research to draw their own conclusions based on their preferred assumptions. The subsequent panels in Figure~\ref{Figure: Intro} illustrate this by reporting estimates under various assumptions: U-shape bounds on nonresponse propensities conditional on EHNI (panel 2), the missing completely at random (MCAR) assumption for unit nonresponse combined with a worst-case scenario for wave nonresponse (panel 3), and a Kolmogorov-Smirnov neighborhood of the MCAR assumption for both nonresponse types (panel 4). While the sets in panel 3 are uninformative for comparing $F_1$ and $F_2$ using restricted stochastic dominance for the same reason as the worst-case scenario, panels 2 and 4 reveal that the upper boundaries of wave 2 identified sets are strictly below the lower boundaries of wave 1, presenting some evidence of a decline in poverty (i.e., $H_1$). Strong evidence for this decline under those nonresponse assumptions can be obtained by testing
\begin{align}\label{eq - into test problem}
H^1_0: \text{Not } H^1_1 \quad\text{vs.} \quad H_1: \overline{F}_2(y) < \underline{F}_1(y)\;\forall y\in [\underline{t}, \overline{t}]
\end{align}
and rejecting $H^1_0$ in favor of $H^1_1$. Here, $\underline{F}_1$ and $\overline{F}_2$ represent the lower and upper boundaries of identified sets for $F_1$ and $F_2$, respectively, under the maintained nonresponse assumptions. The point being that rejection of $H^1_0$ in favor of $H^1_1$ implies rejection of $H_0$ in favor of $H_1$ under the maintained nonresponse assumptions.

\par We propose a survey-weighted statistical testing procedure for test problems of the form~(\ref{eq - into test problem}). Our approach employs estimating functions with nuisance functionals (e.g.,~\citealp{GODAMBE-Thompson-book-chapter}, and~\citealp{zhaowu2020}) in a general framework that covers a broad spectrum of nonresponse assumptions through specifications of the nuisance functionals -- see Section~\ref{Subsection Examples} for some examples. This approach is helpful for performing a sensitivity analysis of inferential conclusions of the test under different assumptions on nonresponse within our framework. We illustrate this point in Section~\ref{Section - Illustration} with an empirical example of poverty orderings (\citealp{fostershorrocks1988b}) using data from the HILDA Survey.

\par The testing procedure employs the pseudo-empirical likelihood method (\citealp{chensitter1999,wurao2006}) to formulate the test statistic, applicable across any predesignated order of stochastic dominance and range. It compares the statistic against a chi-squared distribution's quantile with one degree of freedom. To derive the test's asymptotic properties, we have developed a design-based framework for survey sampling from large finite populations. We establish the test's asymptotic validity with uniformity, and its asymptotic power against both local and non-local alternatives, achieving consistency against distant alternatives. However, the test exhibits asymptotic bias against alternatives that rapidly converge to the null hypothesis boundary, which is not surprising, as the test is not asymptotically similar on the boundary.

\par In developing our testing method, we primarily addressed unit and wave nonresponse for simplicity, though the method can be adapted to include item nonresponse. We have also limited our analysis to temporal pairwise comparisons between wave one and any subsequent wave, which is particularly useful when assessing restricted stochastic dominance of an outcome variable from the survey’s start to a later period. Extending beyond these comparisons introduces significant methodological challenges due to the changing composition of the target population and the impact of nonresponse. Section~\ref{Subsection - Discussion - Longitudinal weights} discusses these issues in depth and introduces a pseudo-empirical likelihood testing method based on the \emph{restricted tests} (\citealp{Aitchison-1962}) approach as well as technical details in the Appendix.

\par This paper is organised as follows. Section~\ref{Section - Setup} introduces the testing problem of interest and presents examples of bounds on the dominance functions. Section~\ref{Section - Testing Procedure} describes the pseudo-empirical likelihood ratio testing procedure and the decision rule. Section~\ref{Section - Asymptotic Framework and Results} introduces the large finite population asymptotic framework, and presents our results. Section~\ref{Section - Discussion} discusses the scope of our main results, their implications for empirical practice, and some directions of future research. Section~\ref{Section - Illustration} illustrates our testing procedure using HILDA survey data in testing for poverty orderings among Australian households, and Section~\ref{Section - Conclusion} presents our conclusions.

\section{Setup}\label{Section - Setup}

\par Panel survey samples on the variable of interest are drawn from units in the first wave of a reference finite population who are also followed over time. The inference procedure of this paper uses data arising from panel surveys between wave 1 and a given subsequent wave. Borne in mind are two finite populations corresponding to the first wave and the subsequent wave of interest. Let $\mathcal{P}_{N_K} = \left\{ \{Y^K_i, Z^K_i \} : i = 1,2,\hdots, N_K \right\}$ denote a finite population characterized by the variables $(Y^K, Z^K)$ for each $K\in\{A,B\}$, where $N_K$ is the population total, $Y^K_i\in\mathcal{X}_K\subset\mathbb{R}$ is the outcome variable of interest, and $Z^K_i$ 0/1 binary variable indicating the unit's response to the survey in the reference population, for each $i$. In the design-based approach (\citealp{neyman1934twoaspects}), the elements of $\mathcal{P}_{N_K}$ are treated as constant for each $K\in\{A,B\}$. This approach is different from mainstream statistics and econometrics, and is the core of survey data estimation and inference.

\par Without loss of generality, let $A$ correspond to the reference population in wave 1 and $B$ its subsequent counterpart, so that period $A$ precedes period $B$. In this setup,
\begin{align}\label{eq - Sample pairs}
    \text{the practitioner observes}\quad Y_i= \ \begin{cases}
    (Y^A_i,Y^{B}_i) & \text{if} \ Z^A_i=Z^B_i = 1 \\
    (Y^A_i,*) & \text{if} \ Z^A_i=1,Z^B_i = 0,\\
    (*, *) & \text{if} \ Z^A_i=Z^B_i = 0,
    \end{cases}
\end{align}
where ``$*$'' denotes the missing value code, with sampling unit $i$ from population $A$ (i.e., in wave 1) and followed through time to population $B$  (i.e., the subsequent wave). Unit nonresponse corresponds to the event $\left\{Z^{A}=Z^{B}=0\right\}$ in our setup since units that do not respond in wave 1 (i.e., population $A$) are not followed into later waves (i.e., population $B$). Wave nonresponse corresponds to the event $\left\{Z^{A}=1,Z^{B}=0\right\}$, since the unit responds in wave 1 but not in the subsequent wave. The event $\{Z^A_i=0,Z^B_i = 1\}$ does not occur in our setup since units that have not responded to the survey in the first wave are not followed in subsequent waves by design, so that $\{Z^A_i=0,Z^B_i = 1\}=\emptyset$, holds. We do not treat item nonresponse for simplicity and to avoid notational clutter. Its treatment is similar to that of the other kinds of nonresponse and requires the introduction of two more 0/1 binary variables to capture this behavior.

\par The objective of this paper is to compare the finite population distributions of the outcome variable arising from $\mathcal{P}_{N_A}$ and $\mathcal{P}_{N_B}$ using stochastic dominance under survey nonresponse. For each $K\in\{A,B\}$, the CDFs of $Y_K$ under $\mathcal{P}_{N_K}$ is defined as $F_{N_K}(x)=N_K^{-1}\sum_{i=1}^{N_K}\mathbbm{1}[Y_i^K \leq x]$ for each $x\in\mathcal{X}_K$. Following~\cite{davidsonSD2008}, for $K = A,B$, $x \in \mathbb{R}$, define the
dominance functions by the recursion: $D_{N_K}^1 (x)= F_{N_K}(x)$ and  $D^{s}_{N_K}(x) =\int_{-\infty}^{x}D^{s-1}_{N_K}(u)\, \text{d}u$ for $s=2,3,4,\ldots$. Tedious calculation of these integrals yields  $D^{s}_{N_K}(x)=N_K^{-1}\sum_{i=1}^{N_K} \frac{(x-Y_i^K)^{s-1}}{(s-1)!} \mathbbm{1}[Y_i^K \leq x]$ for $K\in\{A,B\}$ and $s=1,2,\ldots$. We say that $F_{N_A}$ stochastically dominates $F_{N_B}$ at order $s \in \mathbb{N}$, if $D_{N_A}^s (x) \leq D_{N_B}^s (x)$ $\forall x \in \mathcal{X}_A \cup \mathcal{X}_B$. We have strict dominance when the inequality is strict over all points in the joint support $\mathcal{X}_A \cup \mathcal{X}_B$. We make the following assumption on $\mathcal{X}_A$ and $\mathcal{X}_B$.
\begin{assumption}\label{Assumption - Compactness of X_K}
$\mathcal{X}_K\subset\mathbb{R}$ is compact, for $K=A,B$.
\end{assumption}
\noindent This assumption is natural in many applications, for example in the context of income and wealth distributions.

\subsection{Hypotheses and Estimating Functions}
\par The test problem of interest in practice is
\begin{align}\label{eq - test problem}
H_0:\max_{x\in[\underline{t},\overline{t}]}\left(D_{N_A}^{s}(x)-D_{N_B}^{s}(x)\right)\geq0\;\;\text{vs.}\;\; H_1:\max_{x\in[\underline{t},\overline{t}]}\left(D_{N_A}^{s}(x)-D_{N_B}^{s}(x)\right)<0,
\end{align}
where $[\underline{t},\overline{t}] \subseteq \text{interior}\left(\mathcal{X}_A \cup \mathcal{X}_B \right)$ and the order $s$ are pre-designated by the researcher and depends on the context of the application; for example, in poverty analysis using the headcount ratio (i.e., $s=1$), this interval could be the set of feasible poverty lines. The null hypothesis states that $F_{N_A}$ does not stochastically dominate $F_{N_B}$ at the $s$th order over the interval $[\underline{t},\overline{t}]$, and the alternative hypothesis $H_1$ is its negation; that is, strict restricted stochastic dominance of $F_{N_A}$ by $F_{N_B}$ at the $s$th order.

\par The natural estimand in developing a statistical procedure for the test problem~(\ref{eq - test problem}) is
\begin{align}\label{eq - contrasts}
\left\{D_{N_A}^s (x) - D_{N_B}^s (x), x\in[\underline{t},\overline{t}]\right\}.
\end{align}
Using the notion of estimating functions (\citealp{GODAMBE-Thompson-book-chapter}), we can define the contrasts~(\ref{eq - contrasts}) as the unique solution of \emph{census estimating equations}, yielding their point-identification. However, nonresponse calls their point-identification into question. In an attempt to circumvent these identification challenges posed by missing data, survey designers have implemented assumptions that nonresponse is ignorable, such as MCAR and Missing at Random (MAR), which point-identify the dominance contrasts~(\ref{eq - contrasts}). While such assumptions enable the development of statistical procedures for tackling the testing problem~(\ref{eq - test problem}), they are also typically implausible in practice as nonresponse is not necessarily ignorable.

\par It is productive to firstly consider the identification of the contrasts~(\ref{eq - contrasts}) under nonresponse. Given $s$, by the Law of Total Probability, for each $x\in\mathbb{R}$
\begin{align}\label{eq - LTP Expansion}
D_{N_K}^s (x) & =\sum_{z^A,z^B\in\{0,1\}}\mathbbm{E}_{F_K\left(\cdot| Z^A=z^A,Z^B=z^B\right)}\left[\frac{(x-Y^K)^{s-1}}{(s-1)!}\,\mathbbm{1}\left[Y^K\leq x\right]\right]\,\delta_{z^A z^{B}},
\end{align}
with $F_K\left(x|Z^A=z^A,Z^B=z^B\right)=\delta_{z^A z^B}^{-1}\,N^{-1}_{A}N^{-1}_{B}\sum_{i,j}\mathbbm{1}\left[Y_i^K\leq x, Z_i^A=z^A,Z^B_j=z^B\right]$, and $\delta_{z^A z^{B}} = N_A^{-1} N_B^{-1}\sum_{i,j}\mathbbm{1}\left[Z^A_i=z^A,Z^B_j=z^B\right]$ whose sums run through the elements of the two finite populations, $\mathcal{P}_{N_A}$ and $\mathcal{P}_{N_B}$. Since $\{Z^A=0,Z^B=1\}=\emptyset$, holds, it implies that $\delta_{01}=0$, simplifying the representation of $D_{N_K}^s (\cdot)$ in~(\ref{eq - LTP Expansion}) by dropping the conditional expectation $\mathbbm{E}_{F_K\left(\cdot| Z^A=0,Z^B=1\right)}\left[\frac{(x-Y^K)^{s-1}}{(s-1)!}\,\mathbbm{1}\left[Y^K\leq x\right]\right]$. For population $K=A$, this representation reveals that $D_{N_A}^s(\cdot)$ is not point-identified, since all terms in this representation are identifiable by the data, except for $\mathbbm{E}_{F_A\left(\cdot| Z^A=Z^B=0\right)}\left[\frac{(x-Y^K)^{s-1}}{(s-1)!}\,\mathbbm{1}\left[Y^K\leq x\right]\right]$. The same outcome holds for population $B$, but now the non-identifiable parts are $$\mathbbm{E}_{F_B\left(\cdot| Z^A=1,Z^B=0\right)}\left[\frac{(x-Y^B)^{s-1}}{(s-1)!}\,\mathbbm{1}\left[Y^B\leq x\right]\right]\;\text{and}\;\mathbbm{E}_{F_B\left(\cdot| Z^A=Z^B=0\right)}\left[\frac{(x-Y^B)^{s-1}}{(s-1)!}\,\mathbbm{1}\left[Y^B\leq x\right]\right].$$

\par Consequently, the contrasts~(\ref{eq - contrasts}) are not point-identified unless one is willing to make strong assumptions on nonresponse in practice. The most appealing way to mitigate this identification problem created by nonrepsonse is for survey designers to improve the response rates of their surveys and to obtain validation data that delineates the nature of nonresponse. However, in the absence of this, the only way to conduct inference in the test problem~(\ref{eq - test problem}) is by making assumptions that either directly or indirectly constrain the distribution of missing data. Of course, such assumptions are generally nonrefutable as the available data imply no constraints on the missing data. But this approach enables empirical researchers to draw their own conclusions using assumptions they deem credible enough to maintain.

\par Our approach is to develop a testing procedure based on bounds of the dominance contrasts that depend on observables and encode assumptions on nonresponse. Suppose that the sharp identified set of $D_{N_K}^s (\cdot)$ is given by $\underline{D}_{N_K}^s(x)\leq  D_{N_K}^s (x) \leq \overline{D}_{N_K}^s(x)$ $\forall x\in[\underline{t},\overline{t}]$ and $K=A,B$, where $\underline{D}_{N_K}^s(\cdot)$ and $\overline{D}_{N_K}^s(\cdot)$ are the respective lower and upper bounding functions of this set. Then the boundaries of these identified sets imply bounds on the contrasts:
\begin{align}\label{eq - bounds on DF}
    \underline{D}_{N_A}^s(x) - \overline{D}_{N_B}^s(x) \leq D_{N_A}^s (x) - D_{N_B}^s (x) \leq \overline{D}_{N_A}^s(x) -\underline{D}_{N_B}^s(x), \quad \forall x \in [\underline{t},\overline{t}].
\end{align}
The identified set~(\ref{eq - bounds on DF}) captures all the information on the contrasts under the maintained assumption on nonresponse. Now imposing sign restrictions on its bounds is advantageous as it leads to a method of inferring restricted stochastic dominance on the distributions in question. In particular, consider the testing problem:
\begin{align}\label{eq - boundary of null}
H^{1}_0:\max_{x\in[\underline{t},\overline{t}]}\left( \overline{D}_{N_A}^s(x) -\underline{D}_{N_B}^s(x)\right)\geq0\;\;\text{vs.}\;\; H^{1}_1:\max_{x\in[\underline{t},\overline{t}]}\left( \overline{D}_{N_A}^s(x) -\underline{D}_{N_B}^s(x)\right)<0.
\end{align}
Rejecting $H^{1}_0$ in favor of $H^{1}_1$ in~(\ref{eq - boundary of null}) implies rejection of $H_0$ in favor of $H_1$ in~(\ref{eq - test problem}) since $$\max_{x\in[\underline{t},\overline{t}]}\left( \overline{D}_{N_A}^s(x) -\underline{D}_{N_B}^s(x)\right)<0\implies D^{s}_{N_A}(x)\leq\overline{D}_{N_A}^{s}(x)<\underline{D}_{N_B}^{s}(x)\leq D^{s}_{N_B} (x),$$ holds, $\forall x\in[\underline{t},\overline{t}]$.

\par Next, we present a general estimating function approach that targets the contrasts \\ $\left\{\overline{D}_{N_A}^{s}(x)-\underline{D}_{N_B}^{s}(x),x\in[\underline{t},\overline{t}]\right\}$. We also demonstrate this approach's scope in applications with several important examples. For each $x\in[\underline{t},\overline{t}]$, consider the estimating function $h_s\left(Y^A, Y^B,Z^A,Z^B,\theta(x),\varphi(x)\right)$, given by
\begin{equation}
\begin{aligned}
&\varphi_3(x)\mathbbm{1}\left[Z^A=1\right]-\theta(x)\mathbbm{1}\left[Z^A=1\right]-\frac{(x-Y^B)^{s-1}}{(s-1)!}\mathbbm{1}\left[Y^B\leq x,Z^A=Z^B=1\right]\varphi_4(x)\\
&+\frac{(x-Y^A)^{s-1}}{(s-1)!} \left[\mathbbm{1}\left[Y^A\leq x,Z^A=Z^B=1\right]\varphi_1(x)+\mathbbm{1}\left[Y^A\leq x,Z^A=1,Z^B=0\right]\varphi_2(x)\right],
\end{aligned}\label{eq - Estimating Function}
\end{equation}
where $\varphi=[\varphi_1,\varphi_2,\varphi_3,\varphi_4]$ is a vector of nuisance functionals. For a given vector $\varphi_{N_A,N_B}$, the estimand in this estimating function is the parameter $\theta(\cdot)$. It solves the census estimating equations
\begin{align}\label{eq - Census Est Eq. Bounds}
\frac{1}{N_A}\frac{1}{N_B}\sum_{i=1}^{N_A}\sum_{j=1}^{N_B}h_s\left(Y_i^A, Y_j^B,Z_i^A,Z_j^B,\theta(x),\varphi_{N_A,N_B}(x)\right)=0\quad \forall x\in[\underline{t},\overline{t}].
\end{align}
and has the form
\begin{align*}
\theta_{N_A,N_B}(x;\varphi_{N_A,N_B}) & = \mathbbm{E}_{F_A\left(\cdot| Z^A=Z^B=1\right)}\left[\frac{(x-Y^A)^{s-1}}{(s-1)!}\,\mathbbm{1}\left[Y^A\leq x\right]\right]\frac{\delta_{11}}{\delta_{11}+\delta_{10}}\,\varphi_{1,N_A,N_B}(x) \\
& + \mathbbm{E}_{F_A\left(\cdot| Z^A=1,Z^B=0\right)}\left[\frac{(x-Y^A)^{s-1}}{(s-1)!}\,\mathbbm{1}\left[Y^A\leq x\right]\right]\frac{\delta_{10}}{\delta_{11}+\delta_{10}}\,\varphi_{2,N_A,N_B}(x)\\
& - \mathbbm{E}_{F_B\left(\cdot| Z^A=Z^B=1\right)}\left[\frac{(x-Y^B)^{s-1}}{(s-1)!}\,\mathbbm{1}\left[Y^B\leq x\right]\right]\frac{\delta_{11}}{\delta_{11}+\delta_{10}}\,\varphi_{4,N_A,N_B}(x)\\
& +\varphi_{3,N_A,N_B}(x),
\end{align*}
for each $x\in[\underline{t},\overline{t}]$.
This estimand targets $\overline{D}_{N_A}^{s}(x)-\underline{D}_{N_B}^{s}(x)$ under specifications of $\varphi_{N_A,N_B}$ that encode the maintined nonresponse assumptions. The next section illustrates the scope of our approach using several examples.

\subsection{Examples}\label{Subsection Examples}
\par To fix ideas, we introduce examples of sharp bounds that our general formulation covers. We defer a formal analysis to Section~\ref{Appendix D} in the Supplementary Material, where we also derive the identified sets of the dominance functions $D_{N_A}^{s}(\cdot)$ and $D_{N_B}^{s}(\cdot)$ under the informational assumptions of the examples.

\par The first example concerns the scenario where there is no information on the nonresponse-generating mechanism available to the practitioner. The worst-case bounds on the contrasts must be used in practice. These bounds summarize what the data, and only the data, say about the dominance contrast~(\ref{eq - contrasts}). While they may be wide in practice, it does not necessarily preclude their use for conducting distributional comparisons.~\cite{Tabri2021} make this point but in the context of ordinal data and first-order stochastic dominance.

\begin{example}\label{Example - WC Bounds}
{\bf Worst-Case Bounds}. Proposition~\ref{prop - WC ID Set} reports the worst-case identified set of the dominance functions and describes their boundaries. The worst-case bounds on the dominance contrast can be derived using this result. In this scenario, the following specification of $\varphi_{N_A,N_B}$ applies: for each $x\in[\underline{t},\overline{t}]$,
\begin{align*}
\varphi_{1,N_A,N_B}(x) & =\varphi_{4,N_A,N_B}(x)=\varphi_{2,N_A,N_B}(x)=\delta_{11}+\delta_{10}\quad\text{and}\\ 
\varphi_{3,N_A,N_B}(x) & =\frac{(x-\underline{Y}^A)^{s-1}}{(s-1)!}\,\delta_{00}\quad\text{where}\quad \underline{Y}^A=\inf\mathcal{X}_A.
\end{align*}
This specification defines $\theta_{N_A,N_B}(\cdot;\varphi_{N_A,N_B})=\overline{D}_{N_A}^{s}(\cdot)-\underline{D}_{N_B}^{s}(\cdot)$, where for each $x\in[\underline{t},\overline{t}]$,
\begin{align*}
\overline{D}_{N_A}^{s}(x) & = \mathbbm{E}_{F_A\left(x| Z^A=Z^B=1\right)}\left[\frac{(x-Y^A)^{s-1}}{(s-1)!}\,\mathbbm{1}\left[Y^A\leq x\right]\right]\,\delta_{11}\nonumber \\
& \qquad +\mathbbm{E}_{F_A\left(x| Z^A=1,Z^B=0\right)}\left[\frac{(x-Y^A)^{s-1}}{(s-1)!}\,\mathbbm{1}\left[Y^A\leq x\right]\right]\,\delta_{10}+\frac{(x-\underline{Y}^A)^{s-1}}{(s-1)!}\,\delta_{00}\,\text{and}\\
\underline{D}_{N_B}^{s}(x) & =  \mathbbm{E}_{F_B\left(x| Z^A=Z^B=1\right)}\left[\frac{(x-Y^B)^{s-1}}{(s-1)!}\,\mathbbm{1}\left[Y^B\leq x\right]\right]\,\delta_{11}.
\end{align*}
Under Assumption~\ref{Assumption - Compactness of X_K}, these bounds are finite, as $\underline{Y}^A\in\mathbb{R}$, holds.
\end{example}

\par The second example presents bounds that arise from imposing the MCAR assumption on unit nonresponders. While implausible for modeling this type of nonresponse in voluntary surveys, the dominant practice by survey-designers has been to use weights to implement this assumption.
\begin{example}\label{Example - MCAR and WC Bounds}
{\bf MCAR for Unit Nonresponse}. Assuming MCAR for unit nonresponse means
\begin{align}\label{eq - Example MCAR}
F_K\left(x| Z^A=Z^B=1\right)=F_K\left(x| Z^A=Z^B=0\right)\quad\forall x\in\mathcal{X}_K,
\end{align}
holds, for $K=A,B$. Proposition~\ref{prop - MCAR UNR WC WNR ID Set} reveals that this assumption point-identifies $D_{N_A}^{s}(\cdot)$ and only partially identifies $D_{N_B}^{s}(\cdot)$, for each $s\in\mathbb{Z}_+$, since $F_B(x|Z^A=1,Z^B=0)$ is not pinned down by the conditions~(\ref{eq - Example MCAR}). Thus, in conjunction with the WC lower bound for $F_B(x|Z^A=1,Z^B=0)$, the following specification of $\varphi_{N_A,N_B}$ encodes conditions~(\ref{eq - Example MCAR}):
\begin{align*}
\varphi_{1,N_A,N_B}(x) =\varphi_{4,N_A,N_B}(x)=\frac{(\delta_{11}+\delta_{00})(\delta_{11}+\delta_{10})}{\delta_{11}},\quad\varphi_{2,N_A,N_B}(x) =1,
\end{align*}
and $\varphi_{3,N_A,N_B}(x) = 0$.
This specification defines $\theta_{N_A,N_B}(\cdot;\varphi_{N_A,N_B})=\overline{D}_{N_A}^{s}(\cdot)-\underline{D}_{N_B}^{s}(\cdot)$, where for each $x\in[\underline{t},\overline{t}]$,
\begin{align*}
\overline{D}_{N_A}^{s}(x) & = \mathbbm{E}_{F_A\left(x| Z^A=Z^B=1\right)}\left[\frac{(x-Y^A)^{s-1}}{(s-1)!}\,\mathbbm{1}\left[Y^A\leq x\right]\right]\,(\delta_{11}+\delta_{00})\nonumber \\
& \qquad +\mathbbm{E}_{F_A\left(x| Z^A=1,Z^B=0\right)}\left[\frac{(x-Y^A)^{s-1}}{(s-1)!}\,\mathbbm{1}\left[Y^A\leq x\right]\right]\,\delta_{10}\quad\text{and}\\
\underline{D}_{N_B}^{s}(x) & =  \mathbbm{E}_{F_B\left(x| Z^A=Z^B=1\right)}\left[\frac{(x-Y^B)^{s-1}}{(s-1)!}\,\mathbbm{1}\left[Y^B\leq x\right]\right]\,(\delta_{11}+\delta_{00}).
\end{align*}
\end{example}

\par The third example presents bounds that arise from constraining the nonresponse propensities conditional on the outcome, based on ideas in Section 4.2 of~\cite{manski2016}. The interesting aspect of this example is that information in the form of shape constraints on $\text{Prob}\left(Z^A=Z^B=0| Y^K\leq x\right)$ for $K=A,B$, and $\text{Prob}\left(Z^A=1,Z^B=0| Y^B\leq x\right)$ can be incorporated into our inference procedure.
\begin{example}\label{Example - NBD of MCAR}
{\bf Restrictions on Unit and Wave Nonresponse Propensities}. This example constrains the distribution of missing data through bounds on the conditional probabilities
$\text{Prob}\left(Z^A=1,Z^B=0| Y^B\leq x\right)$ and $\text{Prob}\left(Z^A=Z^B=0| Y^K\leq x\right)$ for $K=A,B$. For each $x\in\mathcal{X}_A$, suppose that $$\delta_{00}\,L^{A}_{00}(x)\leq \text{Prob}\left(Z^A=Z^B=0| Y^A\leq x\right)\leq U^{A}_{00}(x) \,\delta_{00},$$ where $\text{Prob}\left(Z^A=Z^B=0| Y^A\leq x\right) = \frac{\sum_{i,j}\mathbbm{1}\left[Y_i^A\leq x, Z_i^A=Z^B_j=0\right]}{\sum_{i}\mathbbm{1}\left[Y^A_i\leq x\right]}$,
and $L^{A}_{00}$ and $U^{A}_{00}$ are CDFs that are predesignated by the practitioner and satisfy $L^{A}_{00}\leq U^{A}_{00}$. Bounds on the probability $\text{Prob}\left(Z^A=Z^B=0| Y^B\leq x\right)$ similar to the above also hold, but with CDFs $L^{B}_{00}$ and $U^{B}_{00}$. Furthermore, consider bounds on the probability $\text{Prob}\left(Z^A=1,Z^B=0| Y^B\leq x\right)$: for each $x\in\mathcal{X}_B$, $\delta_{10}\,L_{10}(x) \leq \text{Prob}\left(Z^A=1,Z^B=0| Y^B\leq x\right)\leq U_{10}(x)\,\delta_{10}$, where $L_{10}$ and $U_{10}$ are CDFs that satisfy $L_{10}\leq U_{10}$, and $\text{Prob}\left(Z^A=1,Z^B=0| Y^B\leq x\right)= \frac{\sum_{i,j}\mathbbm{1}\left[Y_j^B\leq x, Z_i^A=1,Z^B_j=0\right]}{\sum_{j}\mathbbm{1}\left[Y^B_j\leq x\right]}$. In practice, shape constraints on the above conditional probabilities can be incorporated into our inference procedure by imposing them on the bounding functions.

\par Now consider the CDFs given by $\overline{G}_A(x)=\frac{1-\delta_{00}}{1-U^{A}_{00}(x)\delta_{00}}$ for each $x\in\mathcal{X}_A$ and $\underline{G}_B(x) =\frac{\delta_{11}}{1-L^{B}_{00}(x)\delta_{00}-L_{10}(x)\delta_{10}}$ for each $x\in\mathcal{X}_B$.
We define their dominance functions recursively: for $G\in\{\overline{G}_A,\underline{G}_B\}$, these functions are $D_{G}^1 (x)= G(x)$ and  $D^{s}_{G}(x) =\int_{-\infty}^{x}D^{s-1}_{G}(u)\, \text{d}u$ for $s=2,3,4,\ldots$. Furthermore, we define the recursively defined functions on $\mathbb{R}^2$ given by $R_0(y,x)=\mathbbm{1}[y\leq x]$, $R_j(y,x)=\int_{-\infty}^x R_{j-1}(y,u)\,du$ for $j=1,2,\ldots$.

\par Proposition~\ref{prop - NBD of MCAR ID Set} reports the identified sets of the dominance functions under these informational conditions. Using  this result, the specification of $\theta_{N_A,N_B}(\cdot;\varphi_{N_A,N_B})=\overline{D}_{N_A}^{s}(\cdot)-\underline{D}_{N_B}^{s}(\cdot)$, where for each $x\in[\underline{t},\overline{t}]$, takes the form
\begin{align*}
\overline{D}_{N_A}^{s}(x) & =\left(\frac{\delta_{11}}{1-\delta_{00}}\,F_A\left(x| Z^A=Z^B=1\right)+\frac{\delta_{10}}{1-\delta_{00}}\,F_A\left(x| Z^A=1,Z^B=0\right)\right)D^{s}_{\overline{G}_A}(x)\\
& \qquad-\mathbbm{1}[s>1]\frac{\delta_{11}}{1-\delta_{00}}\sum_{j=0}^{s-2}\mathbbm{E}_{F_A\left(\cdot| Z^A=Z^B=1\right)}\left[D^{s}_{\overline{G}_A}(Y^A)\,R_j(Y^A,x)\right]\\
& \qquad-\mathbbm{1}[s>1]\frac{\delta_{10}}{1-\delta_{00}}\sum_{j=0}^{s-2}\mathbbm{E}_{F_A\left(\cdot| Z^A=1,Z^B=0\right)}\left[D^{s}_{\overline{G}_A}(Y^A)\,R_j(Y^A,x)\right]\\
\underline{D}_{N_B}^{s}(x) & =F_B\left(x| Z^A=Z^B=1\right)\,D^{s}_{\underline{G}_B}(x)\\
 & \qquad -\mathbbm{1}[s>1]\sum_{j=0}^{s-2}\mathbbm{E}_{F_B\left(\cdot| Z^A=Z^B=1\right)}\left[D^{s}_{\underline{G}_B}(Y^B)\,R_j(Y^B,x)\right].
\end{align*}
Therefore, the following specification of $\varphi_{N_A,N_B}$ encodes this assumption on unit and wave nonrepsonse: for each $x\in[\underline{t},\overline{t}]$, $\varphi_1(x)=\varphi_2(x)=D^{s}_{\overline{G}_A}(x)$, $\varphi_4(x) = \frac{1-\delta_{00}}{\delta_{11}}\,D^{s}_{\underline{G}_B}(x)$, and
\begin{align*}
\varphi_3(x) & =\mathbbm{1}[s>1]\sum_{j=0}^{s-2}\mathbbm{E}_{F_B\left(\cdot| Z^A=Z^B=1\right)}\left[D^{s}_{\underline{G}_B}(Y^B)\,R_j(Y^B,x)\right]\\
 & \qquad-\mathbbm{1}[s>1]\frac{\delta_{11}}{1-\delta_{00}}\sum_{j=0}^{s-2}\mathbbm{E}_{F_A\left(\cdot| Z^A=Z^B=1\right)}\left[D^{s}_{\overline{G}_A}(Y^A)\,R_j(Y^A,x)\right]\\
& \qquad-\mathbbm{1}[s>1]\frac{\delta_{10}}{1-\delta_{00}}\sum_{j=0}^{s-2}\mathbbm{E}_{F_A\left(\cdot| Z^A=1,Z^B=0\right)}\left[D^{s}_{\overline{G}_A}(Y^A)\,R_j(Y^A,x)\right].
\end{align*}
\end{example}

\par The fourth example is a neighborhood-based approach to a sensitivity analysis of empirical conclusions to departures from the MCAR assumption for unit and wave nonresponse. It is based on~\cite{klinesantos2013}, who put forward a construction using the maximal Kolmogorov-Smirnov distance between the distributions of missing and observed outcomes, which allows a determination of the critical level of selection for which hypotheses regarding the dominance contrast cannot be rejected.
\begin{example}\label{Example - Kline and Santos}
The MCAR assumption for both unit and wave nonresponse is
\begin{equation}
\label{eq - Example MCAR-KS}
\begin{aligned}
F_K\left(x| Z^A=Z^B=1\right) & =F_K\left(x| Z^A=Z^B=0\right)\quad\forall x\in\mathcal{X}_K,K=A,B,\quad\text{and}\\
F_B\left(x| Z^A=Z^B=1\right) & =F_B\left(x| Z^A=1,Z^B=0\right)\quad\forall x\in\mathcal{X}_B.
\end{aligned}
\end{equation}
The approach of~\cite{klinesantos2013} is to build neighborhoods for the non-identified CDFs, \\ $F_K\left(\cdot| Z^A=Z^B=0\right)$ and $F_B\left(\cdot| Z^A=1,Z^B=0\right)$, according to the maximal Kolmogorov-Smirnov distance to quantify their divergence from the identified CDFs. Proposition~\ref{prop - Kline Santos ID Set} reports the identified sets of the dominance functions $D_{N_K}^{s}(\cdot)$ under the conditions~(\ref{eq - Example MCAR-KS}). For each $K\in\{A,B\}$, the boundary of the identified set depends on parameter $\gamma_A,\gamma^{00}_B,\gamma^{10}_B\in[0,1]$. The parameters $\gamma_A$ and $\gamma^{00}_B$ represent the fraction in populations $A$ and $B$, respectively, of unit nonresponders whose outcome variable is not well represented by the distribution $F_A\left(\cdot| Z^A=Z^B=1\right)$ and $F_B\left(\cdot| Z^A=Z^B=1\right)$, respectively. Similarly, the parameter $\gamma^{10}_B$ represent the fraction in population $B$ of wave nonresponders whose outcome variable is not well represented by the distribution $F_B\left(\cdot| Z^A=Z^B=1\right)$.

\par Using this result, we obtain $\theta_{N_A,N_B}(\cdot;\varphi_{N_A,N_B})=\overline{D}_{N_A}^{s}(\cdot)-\underline{D}_{N_B}^{s}(\cdot)$, where for each $x\in[\underline{t},\overline{t}]$ and $\gamma_A,\gamma_B\in[0,1]$
\begin{align*}
\overline{D}_{N_A}^{s}(x) & = \mathbbm{E}_{F_A\left(\cdot| Z^A=Z^B=1\right)}\left[\frac{(x-Y^A)^{s-1}}{(s-1)!}\,\mathbbm{1}\left[Y^A\leq x\right]\right]\,(\delta_{11}+(1-\gamma_A)\delta_{00})\nonumber \\
& \qquad +\mathbbm{E}_{F_A\left(\cdot| Z^A=1,Z^B=0\right)}\left[\frac{(x-Y^A)^{s-1}}{(s-1)!}\,\mathbbm{1}\left[Y^A\leq x\right]\right]\,\delta_{10}
+ \frac{(x-\underline{Y}^A)^{s-1}}{(s-1)!}\,\gamma_A \delta_{00}\,\text{and}\\
\underline{D}_{N_B}^{s}(x) & =  \mathbbm{E}_{F_B\left(\cdot| Z^A=Z^B=1\right)}\left[\frac{(x-Y^B)^{s-1}}{(s-1)!}\,\mathbbm{1}\left[Y^B\leq x\right]\right]\,(\delta_{11}+(1-\gamma^{00}_B)\delta_{00}+(1-\gamma^{10}_B)\delta_{10}).
\end{align*}
Therefore, the following specification of $\varphi_{N_A,N_B}$ encodes this assumption on unit and wave nonrepsonse: for each $x\in[\underline{t},\overline{t}],$
\begin{align*}
\varphi_1(x) & =\frac{(\delta_{11}+(1-\gamma_A)\delta_{00})(\delta_{11}+\delta_{10})}{\delta_{11}},\;\varphi_2(x)=\delta_{11}+\delta_{10}, \\
\,\varphi_3(x) & =\frac{(x-\underline{Y}^A)^{s-1}}{(s-1)!}\,\gamma_A \delta_{00},\;\;\text{and}\;\;\varphi_4(x)=\frac{(\delta_{11}+(1-\gamma^{00}_B)\delta_{00}+(1-\gamma^{10}_B)\delta_{10})(\delta_{11}+\delta_{10})}{\delta_{11}}.
\end{align*}
\end{example}

\subsection{Sample Estimating Equations}
\par In the above examples, the estimating function~(\ref{eq - Estimating Function}) depends only on observables and the nuisance parameter. For a maintained assumption on nonresponse, these examples show that there is a unique value of the nuisance parameter that encodes it into the contrasts $\left\{\overline{D}_{N_A}^{s}(x) -\underline{D}_{N_B}^{s}(x),x\in[\underline{t},\overline{t}]\right\}$. We denote this value of the nuisance parameter by $\varphi_{N_A,N_B}$. It can be consistently estimated using a plug-in procedure with the achieved sample $\{Y_i,i\in V\}$ and the inclusion probabilities $\pi_i=\text{Prob}\left[i\in V\right]$ for $i\in V$, where $Y_i$ is defined in~(\ref{eq - Sample pairs}) for each $i$, $V\subset \{1,\ldots, N_A\}$ is a survey sample of units from population $\mathcal{P}_{A,N_A}$ (i.e., wave 1). In practice the inclusion probabilities are reported as \emph{design} weights, $\{W_{\ell},\ell\in V\}$ satisfying the normalization $W_\ell/k=\pi^{-1}_\ell/\sum_{i\in V}\pi^{-1}_i$, where $k=\sum_{i\in V}$.
\par The estimator of $\theta_{N_A,N_B}(\cdot;\varphi_{N_A,N_B})$ solves the sample-analogue census estimating equations~(\ref{eq - Census Est Eq. Bounds}):$$k^{-1}\sum_{i\in V}W_i\,h_s\left(Y_i^A, Y_i^B,Z_i^A,Z_i^B,\theta(x),\hat{\varphi}(x)\right)=0\;\quad \forall x\in[\underline{t},\overline{t}],$$ where $\hat{\varphi}$ is the plug-in estimator of $\varphi_{N_A,N_B}$. We can simplify this set of equations using the fact that the moment function $h_s\left(Y_i^A, Y_i^B,Z_i^A,Z_i^B,\theta(x),\varphi(x)\right)$ equals zero for each $i\in V-U$, where $U=\{i\in V: Z_i^{A}=1\}$ is the subsample corresponding to the unit responders, and that
\begin{align}\label{eq - scaled design weight}
\frac{W_i}{k}=\frac{\pi^{-1}_i}{\sum_{i\in V}\pi^{-1}_i}=\frac{\pi^{-1}_i}{\sum_{i\in U}\pi^{-1}_i}\,\frac{\sum_{i\in U}\pi^{-1}_i}{\sum_{i\in V}\pi^{-1}_i},
\end{align}
holds. As the ratio $\frac{\sum_{i\in U}\pi^{-1}_i}{\sum_{i\in V}\pi^{-1}_i}$ is positive, let $\{W^{\prime}_i,i\in U\}$ be such that
\begin{align}\label{eq - Design weights-responders}
\frac{W^{\prime}_i}{n}=\frac{\pi^{-1}_i}{\sum_{i\in U}\pi^{-1}_i}\quad\forall i\in U,\quad \text{where} \quad n=\sum_{i\in U}.
\end{align}
Then, we simplify the sample-analogue of the census estimating equations as \\
$0 =\sum_{i\in V}\frac{W_i}{k}\,h_s\left(Y_i^A, Y_i^B,Z_i^A,Z_i^B,\theta(x),\hat{\varphi}(x)\right)=\sum_{i\in U}\frac{W^\prime_i}{n}\,h_s\left(Y_i^A, Y_i^B,Z_i^A,Z_i^B,\theta(x),\hat{\varphi}(x)\right)$ for each $x\in[\underline{t},\overline{t}]$,
by substituting out $W_i/k$ using~(\ref{eq - scaled design weight}) and dividing out the common factor $\frac{\sum_{i\in U}\pi^{-1}_i}{\sum_{i\in V}\pi^{-1}_i}$. The solution of the resulting equation is the sample-analogue estimator
\begin{align}\label{eq - theta hat}
\hat{\theta}(x;\hat{\varphi}(x))=\frac{1}{n}\sum_{i\in U}W^\prime_i H_i (x;\hat{\varphi}(x)) \quad \forall x\in[\underline{t},\overline{t}],
\end{align}
where for each $i\in U$, $H_i (x;\hat{\varphi}(x))$ is given by
\begin{equation}
\begin{aligned}
\frac{(x-Y_i^A)^{s-1}}{(s-1)!} & \left[\mathbbm{1}\left[Y_i^A\leq x,Z_i^A=Z_i^B=1\right]\hat{\varphi}_1(x)+\mathbbm{1}\left[Y_i^A\leq x,Z_i^A=1,Z_i^B=0\right]\hat{\varphi}_2(x)\right]\\
 & \qquad+\hat{\varphi}_3(x)\mathbbm{1}\left[Z_i^A=1\right]-\frac{(x-Y_i^B)^{s-1}}{(s-1)!}\mathbbm{1}\left[Y_i^B\leq x,Z_i^A=Z_i^B=1\right]\hat{\varphi}_4(x).
 \label{eq: moment function}
\end{aligned}
\end{equation}

\section{Testing Procedure}\label{Section - Testing Procedure}
\par This section introduces the statistical procedure for implementing the hypothesis testing problem~(\ref{eq - boundary of null}). The procedure is based on the empirical likelihood test of~\cite{davidsonduclos2013}. The test focuses on the boundary of $H_0^1$ in~(\ref{eq - boundary of null}). For a pair of CDFs $F_{N_A}$ and $F_{N_B}$ of the outcome variable of interest, the rejection probability will be highest on the subset of the boundary of the null hypothesis $H_0^1$ where we have exactly one $x\in[\underline{t},\overline{t}]$ such that $\overline{D}_{N_A}^s (x) = \underline{D}_{N_B}^s (x)$. Therefore, we impose the restriction corresponding to the boundary of $H_0^1$ for a single $x\in[\underline{t},\overline{t}]$. To maximize the pseudo-empirical likelihood function (PELF) under this restriction, for each $x\in[\underline{t},\overline{t}]$, compute the maximum PELF whilst imposing $\overline{D}_{N_A}^s (x) = \underline{D}_{N_B}^s (x)$, which is equivalent to $\theta_{N_A,N_B}(x;\varphi_{N_A,N_B}(x))=0$. This corresponds to the maximization problem:
\begin{equation}\label{eq: PELF}
\begin{aligned}
\max_{\Vec{p} \in (0,1]^{n}} \sum_{i\in U} W^\prime_i \log(p_i)\; &\text{s.t.}\;\; p_i>0 \; \forall i,\;\;\sum_{i\in U} W^\prime_i p_i=1,\;\;\text{and}\\
&\sum_{i\in U} W^\prime_i p_i H_i (x;\hat{\varphi}(x)) = 0,
\end{aligned}
\end{equation}
where $\hat{\varphi}$ is a plug-in estimator of $\varphi$, and $H_i (x;\hat{\varphi}(x))$ is the moment function~(\ref{eq: moment function}). The moment condition in \eqref{eq: PELF} imposes the restriction $\overline{D}_{N_A}^s (x) = \underline{D}_{N_B}^s (x)$, by tilting the estimator $\hat{\theta}(x;\hat{\varphi}(x))$ through the probabilities $\{p_i, i\in U\}$ on the subsample $\{Y_i, i\in U\}$.

\par For a fixed $x \in [\underline{t},\overline{t}]$ denote by $L_R (x)$ the maximized value of the constrained maximization problem \eqref{eq: PELF}. Additionally, let $L_{UR} = \sum_{i\in U} W^\prime_i \log(n^{-1})$, which is the unconstrained maximum value of the PELF which corresponds to \eqref{eq: PELF} omitting the constraint $\sum_{i\in U} W^{\prime}_i p_i H_i (x;\hat{\varphi}(x)) = 0$. Then the pseudo-empirical likelihood-ratio statistic for the test problem~(\ref{eq - boundary of null}) is
\begin{align}
\label{LR test statistic}
    LR = \begin{cases}
    \min\limits_{x \in [\underline{t},\overline{t}]} 2\left(L_{UR} - L_R(x)\right)/\widehat{\text{Deff}}(x) & \text{if} \ \hat{\theta}(x;\hat{\varphi}(x))<0 \quad \forall x \in [\underline{t},\overline{t}] \\
    0, & \text{otherwise}
    \end{cases}
\end{align}
where $\widehat{\text{Deff}}(x)$ is an estimator of the design-effect
\begin{align}\label{eq - deff}
\text{Deff}(x) = \left[n^{-1}\sum_{i\in U}(W^\prime_i/n)H^2_i (x;\hat{\varphi}(x))\right]^{-1}\,Var\left(\hat{\theta}(x;\hat{\varphi}(x)) \right).
\end{align}
The variance calculation in the numerator of this expression is with respect to the randomness emanating from the survey's design, so that $\widehat{\text{Deff}}(\cdot)$ coincides with $\text{Deff}(\cdot)$ in~(\ref{eq - deff}), except that we replace the design-variance $Var\left(\hat{\theta}(x;\hat{\varphi}(x)) \right)$ with a consistent estimator of it. Such estimators are abundant and well-established, and which one to use in practice depends on how much information the practitioner has from the survey designers; e.g., joint selection probabilities enables the use of Taylor linearization, and the availability of replication design weights enables the use of the jackknife (see, for example, Chapter 4 of~\citealp{fuller2009}). For each $x\in [\underline{t},\overline{t}]$, if the moment equality constraint in~(\ref{eq: PELF}) holds in the population, so that $\theta_{N_A,N_B}(x;\varphi_{N_A,N_B}(x))=0,$ then the denominator of $\text{Deff}(x)$, given by $\sum_{i\in U}(W^\prime_i/n)H^2_i (x)$, is the H\'ajek estimator (\citealp{hajek1971}) of the population variance
 \begin{align}\label{eq - S_NA_NB}
 S_{N_A,N_B}(x)=\frac{1}{N_A} \frac{1}{N_B}\sum_{i=1}^{N_A}\sum_{j=1}^{N_B}\left[h_s\left(Y_i^A, Y_j^B,Z_i^A,Z_j^B,0,\varphi_{N_A,N_B}(x)\right)\right]^2.
 \end{align}

\par As the test statistic \eqref{LR test statistic} is formulated piece-wise, we only implement the procedure if we observe $\hat{\theta}(x;\hat{\varphi}(x))<0$ for all $x \in [\underline{t},\overline{t}]$, in the sample. For a fixed nominal level $\alpha \in (0,1)$, the decision rule of the test is to
\begin{align}\label{eq - decision rule}
    \text{reject} \ H_0^1 \iff LR > c(\alpha),
\end{align}
where $c(\alpha)$ is be the $1-\alpha$ quantile of the chi-squared distribution with one degree of freedom. The next section presents the asymptotic framework for this testing procedure and establishes its asymptotic properties under the null and alternative hypotheses.

\begin{remark}\label{Remark - Design Effect}
The design-effect arises in the Taylor expansion of $L_{UR} - L_R(x)$ under the null and alternative hypotheses; see Lemmas~\ref{Rami Lemma 1} and~\ref{Rami Lemma 2} for results under the null, and see Lemmas~\ref{power theorem Lemma 1} and~\ref{power theorem Lemma 2} for results under the alternative. If we do not account for it as in~(\ref{LR test statistic}), then the testing procedure would suffer from a size distortion under the null, since the distributional limit of the test statistic would be a scaled $\chi^2_1$, with scale involving the design variance $Var\left(\hat{\theta}(x;\hat{\varphi}(x)\right)$. Under the alternative, the test would suffer from distortion in type 2 error, and hence, distort its power. Thus, accounting for the survey's design through the design-effect adjustment in $LR$, corrects these distortions.
\end{remark}
\section{Asymptotic Framework and Results}\label{Section - Asymptotic Framework and Results}
Given a sampling scheme for wave 1, denote by $\mathcal{U}_{N_A} = \{ V : V \subset \{1,2,\hdots, N_A \}\}$ the set of all possible samples from the sampling frame $\{1,2,\hdots, N_A \}$ according to the survey's design. Given a sample $V \in \mathcal{U}_{N_A}$, the survey follows those same sampled units through time to period $B$ to generate a panel survey between periods $A$ and $B$, yielding the dataset: $\left\{(Y_i^A, Z_i^A), (Y_i^B, Z_i^B) \right\}_{i \in U}$. For each $K\in\{A,B\}$, let $\mathcal{M}_{N_K} = \left\{\mathcal{P}_{N_K} : Y^K_i \in \mathcal{X}_K \ \text{and} \ Z^K_i \in \{0,1\} \ \text{for} \ i=1,2,\hdots, N_K \right\}$. This set collects all possible finite populations $\mathcal{P}_{N_K}$. Then the set of all finite populations for both periods $A$ and $B$ of sizes $N_A$ and $N_B$ is given by:
\begin{equation*}
    \mathcal{M}_{N_A, N_B} = \left\{\Pi = \{\mathcal{P}_{N_A}, \mathcal{P}_{N_B} \} : \mathcal{P}_{N_K} \in  \mathcal{M}_{N_K}, \ K = A,B \right\}.
\end{equation*}
Following~\cite{bleuerkratina2005}, for a fixed finite population $\Pi_{N_A,N_B}$ the probability sampling design associated with a sampling scheme on $\Pi_{N_A,N_B}$ is the function:
\begin{equation}
    P : \sigma_{N_A} \times \mathcal{M}_{N_A, N_B} \rightarrow [0,1]
    \label{eq: probability measure}
\end{equation}
such that,
\begin{enumerate}[(i)]
    \item $\sigma_{N_A}$ is a sigma-algebra generated by $\mathcal{U}_{N_A}$;
    \item $P(V,\cdot) > 0 \ \text{is Borel measurable in } \mathcal{M}_{N_A, N_B}, \ \text{for all} \ V \in \mathcal{U}_{N_A}$; and
    \item $P(\cdot, \Pi) \ \text{is a probability measure on } \mathcal{U}_{N_A}, \ \text{for all} \ \Pi \in \mathcal{M}_{N_A, N_B}$.
\end{enumerate}
The design probability space is $(\mathcal{U}_{N_A}, \sigma_{N_A}, P)$ with $P(V, \cdot)> 0$ for all $V \in \mathcal{U}_{N_A}$ and $\sum\limits_{V \in \mathcal{U}_{N_A}} P(U,\cdot) = 1$. Under this setup the survey sample size, $k = \sum\limits_{i \in V}$, is a random variable and all uncertainty is generated from the probability sampling scheme $P$. We follow the notation in finite population literature to indicate that ``$\mid \Pi_{N_A,N_B}$" means the sample, $\left\{(Y_i^A, Z_i^A), (Y_i^B, Z_i^B) \right\}_{i \in U}$, is drawn from the population $\Pi_{N_A,N_B}$. Therefore, for a fixed $\Pi_{N_A, N_B} \in \mathcal{M}_{N_A, N_B}$, $\mathbbm{E}(\cdot \mid \Pi_{N_A, N_B})$ and $Var(\cdot \mid \Pi_{N_A, N_B})$ denote the expectation and variance taken over all possible samples from $\Pi_{N_A,N_B}$ with respect to the probability space $(\mathcal{U}_{N_A}, \sigma_{N_A}, P)$. The survey design-weights $W_i$ satisfy the normalization~(\ref{eq - scaled design weight}), where $\pi_i = \sum_{\{V \in \mathcal{U}_{N_A} \ : \ i \in V \}} P(V, \Pi_{N_A,N_B})$ is the inclusion probability of element $i$ into the sample. The next sections employ this finite population framework to establish the asymptotic properties of the proposed testing procedure in~(\ref{eq - decision rule}), as $N_A$ and $N_B$ diverge.

\subsection{Asymptotic Null Properties}\label{Section - Unifomr Asy Valid}
Firstly, we define the set of finite populations that are compatible with the null hypothesis $H_0^1$. For any $N_A, N_B \in \mathbb{N}$ this set is given by
\begin{align*}
    \mathcal{M}_{N_A , N_B}^0  & = \left\{ \Pi \in \mathcal{M}_{N_A , N_B} : \max_{x \in [\underline{t}, \overline{t}]} \theta_{N_A,N_B}(x;\varphi_{N_A,N_B}(x)) \geq 0 \right\},
\end{align*}
where $\theta_{N_A,N_B}(x;\varphi_{N_A,N_B}(x))=\overline{D}_{N_A}^s (x) - \underline{D}_{N_B}^s (x)$  for each $x \in [\underline{t}, \overline{t}]$. Then, the true population, $\Pi_ 0$, satisfies $H_0^1$ if and only if $\Pi_0 \in \mathcal{M}_{N_A , N_B}^0$. For any $N_A, N_B \in \mathbb{N}$, the size of the test is given by:  $\sup_{\Pi \in \mathcal{M}_{N_A , N_B}^0 } \mathbbm{E}\left(\mathbbm{1}[LR > c(\alpha)] \, \mid \, \Pi \right)$. For a fixed $\Pi\in\mathcal{M}_{N_A , N_B}^0$, $\mathbbm{E}\left(\mathbbm{1}[LR > c(\alpha)] \, \mid \, \Pi \right)$ is the probability of rejecting $H_0^1$ taken over all possible samples from $\Pi$ under the probability sampling design \eqref{eq: probability measure}. Therefore, the size of the test is the largest rejection probability over all finite populations in the model $\mathcal{M}_{N_A , N_B}^0$. To approximate the asymptotic size, we embed $\mathcal{M}_{N_A , N_B}^0$ into a hypothetical sequence of models $\left\{ \mathcal{M}_{N_A , N_B}^0, \; N_A, N_B = 1,2,\hdots \right\}$
that satisfy enough restrictions so that for a given nominal level $\alpha \in (0,1)$:
\begin{equation}
    \limsup_{N_A, N_B \rightarrow \infty} \sup_{\Pi \in \mathcal{M}_{N_A , N_B}^0 } \mathbbm{E}\left(\mathbbm{1}[LR > c(\alpha)] \, \mid \, \Pi \right) \leq \alpha.
    \label{eq: asymptotic size}
\end{equation}
The approach to proving~(\ref{eq: asymptotic size}) uses a characterization of it in terms of sequences of finite populations,
$\left\{\Pi_{N_A,N_B}=\{\mathcal{P}_{A,N_A},\mathcal{P}_{B,N_B}\},N_A,N_B =1,2,\ldots\right\},$ where $\Pi_{N_A,N_B}\in\mathcal{M}^0_{N_A,N_B}$ for all $N_A$ and $N_B,$ and the asymptotic distribution of $\{LR\mid \Pi_{N_A,N_B}\}_{N_A,N_B=1}^{+\infty}$ is calculated along this hypothetical infinite sequence. Recall that $LR\mid \Pi_{N_A,N_B}$ means the statistic, $LR,$ is a function of the survey samples selected from population $\Pi_{N_A,N_B}$.

\par The bounds generated from nonresponse assumptions that we consider are sharp. Under Assumption~\ref{Assumption - Compactness of X_K}, this sharpness implies $\varphi_i\in L^\infty ([\underline{t}, \overline{t}])$ holds, for each $i=1,2,3,4$, where $L^\infty ([\underline{t}, \overline{t}])$ is the space of uniformly bounded measurable functions from $[\underline{t}, \overline{t}]$ into $\mathbb{R}$. The reason is that the worst-case bounds are finite everywhere on $[\underline{t}, \overline{t}]$ under this assumption (see Example 1), so that $\varphi_i$ being unbounded on $[\underline{t},\overline{t}]$ results in bounds that are \emph{not} sharp. Let $\Psi$ denote the vector space of 4-dimensional valued functions, with each component an element of $L^\infty ([\underline{t}, \overline{t}])$. For $\varphi\in\Psi$, the norm of this space is $\|\varphi\|_{\Psi}=\sup_{i=1,2,3,4}\sup_{x\in[\underline{t},\overline{t}]}|\varphi_i(x)|$.

\par Next, we describe the conditions on the surveys' designs for obtaining~(\ref{eq: asymptotic size}). For a given sequence of finite populations $\left\{\Pi_{N_A,N_B}\right\}_{N_A,N_B=1}^{+\infty},$ the conditions we impose on the designs of the surveys are given by the following assumption.
\begin{assumption}\label{Assumption 1}
Fix $s \in \mathbb{N}$. For a sequence of finite populations $\{ \Pi_{N_A , N_B} \}_{{N_A , N_B} = 1}^\infty$ we impose the following conditions on the survey's design.
\begin{enumerate}
    \item $\mathbbm{E}( n \mid \Pi_{N_A, N_B}) \rightarrow \infty \ \text{as} \ N_A, N_B \rightarrow \infty$.
    \item $\|\hat{\varphi}-\varphi_{N_A,N_B}\|_{\Psi} \mid \Pi_{N_A, N_B}\pconv 0$ as $N_A, N_B \rightarrow \infty $.
    \item $\max\limits_{x \in [\underline{t}, \overline{t}]} \left| \hat{\theta}(x;\hat{\varphi}(x)) - \theta_{N_A,N_B}(x;\varphi_{N_A,N_B}(x)) \right| \mid \Pi_{N_A, N_B} \pconv 0$ as $N_A, N_B \rightarrow \infty $.
    \item $\text{For} \ N_A, N_B = 1,2,\hdots, \ Var\left(\hat{\theta}(x;\hat{\varphi}(x)) \mid \Pi_{N_A, N_B} \right) > 0 \ \text{for each} \ x \in [\underline{t},\overline{t}]$.
    \item $\dfrac{\hat{\theta}(\cdot;\hat{\varphi}(\cdot)) -  \theta_{N_A,N_B}(\cdot;\varphi_{N_A,N_B}(\cdot))}{\sqrt{Var\left(\hat{\theta}(\cdot;\hat{\varphi}(\cdot))\mid \Pi_{N_A, N_B} \right)}} \mid \Pi_{N_A, N_B} \rightsquigarrow \mathbb{G} \ \text{in} \ L^\infty ([\underline{t}, \overline{t}])$ as $N_A,N_B \rightarrow\infty$, \\ \\
    where $\rightsquigarrow$ denotes weak convergence and $\mathbb{G}$ is a zero mean Gaussian process.
    \item $\widehat{Var}\left(\hat{\theta}(\cdot;\hat{\varphi}(\cdot))\right)$ satisfies ${\displaystyle\max_{x \in [\underline{t}, \overline{t}]} \left| \frac{Var\left(\hat{\theta}(x;\hat{\varphi}(x))\right)}{\widehat{Var}\left(\hat{\theta}(x;\hat{\varphi}(x))\right)} - 1 \right| \mid \Pi_{N_A, N_B}\pconv 0}$  \\ as $N_A, N_B \rightarrow \infty$.
    \item The above conditions hold for all subsequences $\left\{ \Pi_{{N_A}_m , {N_B}_m } \right\}_{m=1}^\infty$ of $\left\{ \Pi_{N_A, N_B} \right\}_{N_A, N_B =1}^\infty$.
\end{enumerate}
\end{assumption}
\noindent These conditions are versions of commonly used large-sample properties in the survey sampling and partial identification inference literatures; see, for example,~\cite{zhaowu2020} and~\cite{Andrews-Soares2010}. Condition 1 imposes the divergence of the mean subsample size of $U$ as the population totals diverge. Condition 2 imposes the design-consistency of $\hat{\varphi}$ in the norm $\|\cdot\|_{\Psi}$. Condition 3 imposes design-consistency of $\hat{\theta}(\cdot;\hat{\varphi})$, with uniformity over $[\underline{t}, \overline{t}]$, which can be justified by the preceding condition and the Continuous Mapping Theorem, as the estimand is a linear function of the nuisance parameter. Condition 4 imposes positive design-variance for the statistic $\hat{\theta}(\cdot;\hat{\varphi})$. Condition 5 requires that a design-based functional central limit theorem holds for the standardized version of $\hat{\theta}(\cdot;\hat{\varphi})$. Condition 6 imposes design-consistency of the design-variance's estimator, with uniformity over $[\underline{t}, \overline{t}]$. Condition~7 is important for establishing~(\ref{eq: asymptotic size}) via Theorem~\ref{Rami Theorem 1} below.

\par The embedding sequence of null models for developing~(\ref{eq: asymptotic size}) is made precise in the following definition.
\begin{definition}
Suppose that the sequence of null models $\left\{ \mathcal{M}_{N_A , N_B}^0, \; N_A, N_B = 1,2,\hdots \right\}$ is such that every sequence of finite populations $\{ \Pi_{N_A, N_B} \}_{N_A , N_B = 1}^\infty$, with $\Pi_{N_A, N_B} \in \mathcal{M}_{N_A , N_B}^0$ for each $N_A , N_B = 1,2,\ldots$, satisfy the conditions of Assumptions~\ref{Assumption - Compactness of X_K} and \ref{Assumption 1}. Let $\mathbb{W}_0$ denote the set of all such sequences $\{ \Pi_{N_A, N_B} \}_{N_A , N_B = 1}^\infty$ that satisfy Assumption \ref{Assumption 1}.
\label{Definition 1}
\end{definition}
\noindent The first result presents a characterization of~(\ref{eq: asymptotic size}) in terms of $\{\mathcal{M}^0_{N_A,N_B},N_A,N_B=1,2,\ldots\}$ that satisfies Definition~\ref{Definition 1}.
\begin{theorem}\label{Rami Theorem 1}
Let $\alpha \in (0,1)$ and $\mathbb{W}_0$ be as in Definition \ref{Definition 1}. Then \eqref{eq: asymptotic size} is equivalent to
\begin{align}
    \limsup_{N_A , N_B \rightarrow \infty} \mathbbm{E}\left( \mathbbm{1}[LR > c(\alpha)] \, \mid \, \Pi_{N_A, N_B} \right) \leq \alpha \quad \forall \; \{ \Pi_{N_A, N_B} \}_{N_A , N_B = 1}^\infty \in \mathbb{W}_0.
    \label{eq: asymptotic size sequences}
\end{align}
\end{theorem}
\begin{proof}
See Appendix~\ref{Proof of Rami Theorem 1}.
\end{proof}
\noindent An important distinction between our framework and the conventional approach in the literature on inference for finite populations is that, like~\cite{Tabri2021}, we develop the behavior of the test over a set of sequences of finite populations, whereas that literature's focus has been on a single sequence of that sort (e.g.,~\citealp{wurao2006}, and~\citealp{zhaowu2020}). The result of Theorem~\ref{Rami Theorem 1} shows that this distinction is analogous to the difference between uniform and pointwise asymptotics in the partial identification literature.

\par The next result establishes the uniform asymptotic validity of the testing procedure, which is essential for reliable inference in large finite populations where the test statistic's limiting distribution is discontinuous as a function of the underlying population sequence.
\begin{theorem}\label{Rami Theorem 2}
Let $\left\{ \mathcal{M}_{N_A , N_B}^0, \; N_A, N_B = 1,2,\hdots \right\}$, $\mathbb{W}_0$ and $\alpha$ be the same as in Theorem \ref{Rami Theorem 1}. Then \eqref{eq: asymptotic size sequences} holds.
\end{theorem}
\begin{proof}
See Appendix \ref{Proof of Rami Theorem 2}.
\end{proof}
\noindent The key technical steps in the proof of Theorem~\ref{Rami Theorem 2} is to determine the asymptotic distribution of $\{ 2(L_{UR} - L_R (x))/\widehat{\text{Deff}}(x) \mid \Pi_{N_A, N_B} \}_{N_A,N_B=1}^\infty$ along sequences $\{ \Pi_{N_A, N_B} \}_{N_A , N_B = 1}^\infty \in \mathbb{W}_0$ that drift to/on the boundary of the model of the null hypothesis $H_0^1$ where the rejection probability is highest. We show that the limiting distribution for those sequences is $\chi_1^2$. Hence, the test achieves level-$\alpha$ asymptotically. Consequently, a rejection of $H^1_0$ based on this test using a small significance level constitutes very strong evidence in favor of $H^1_1$, and hence, is very strong evidence in favor of $H_1$ defined in~(\ref{eq - test problem}), under the maintained nonresponse assumptions. 

\subsection{Asymptotic Power Properties}\label{Section - Pwer Properties}
We now develop asymptotic properties of the testing procedure along sequences of finite populations under the alternative hypothesis $H_1^1$. In a similar vein to the method used to prove the results for asymptotic size, let $\left\{ \mathcal{M}_{N_A , N_B}^1, \ N_A , N_B = 1,2,\hdots \right\}$ be the embedding sequence where $$\mathcal{M}_{N_A , N_B}^1= \left\{ \Pi \in \mathcal{M}_{N_A , N_B} :\max_{x \in [\underline{t}, \overline{t}]} \theta_{N_A,N_B}(x;\varphi_{N_A,N_B}(x))< 0 \right\},$$ for each $N_A , N_B$. In this formulation, the true population, $\Pi_0$, satisfies $H_1^1$ if and only if $\Pi_0 \in \mathcal{M}_{N_A , N_B}^1$. Thus, $\mathcal{M}_{N_A , N_B}^1$ corresponds to the model of the alternative hypothesis $H_1^1$.

\par For any $N_A , N_B \in \mathbb{N}$, test power is given by $\mathbbm{E}\left( \mathbbm{1}[LR > c(\alpha) ] \, \mid \, \Pi_0 \right)$ when $\Pi_0 \in \mathcal{M}_{N_A , N_B}^1$. Along a sequence of finite populations $\{ \Pi_{N_A , N_B} \}_{N_A,N_B = 1}^\infty$ under the alternative hypothesis, the asymptotic power is given by
\begin{equation}
    \lim_{N_A, N_B \rightarrow \infty} \mathbbm{E} \left( \mathbbm{1}[LR > c(\alpha)] \, \mid \, \Pi_{N_A, N_B} \right).
    \label{asymptotic power}
\end{equation}
As $\alpha \in (0,1)$ and $c(\alpha)$ is a fixed critical value the stochastic behavior the test statistic $\{ LR \mid \Pi_{N_A, N_B} \}_{N_A, N_B}^\infty$ drives the asymptotic power of the testing procedure.  We impose the following conditions on the sequences of finite populations $\{ \Pi_{N_A, N_B} \}_{N_A, N_B = 1}^\infty$ under $H_1^1$, which are useful for deriving the asymptotic behavior of the test statistic.
\begin{assumption}
For a given sequence of finite populations $\{ \Pi_{N_A , N_B} \}_{{N_A , N_B} = 1}^\infty$, we impose the following conditions on the survey's design:
\begin{enumerate}
    \item $\dfrac{\mathbbm{E}( n \mid \Pi_{N_A, N_B})}{n} \pconv d\in \mathbb{R}_{++}$ as $N_A ,N_B \rightarrow \infty$.

    \item For each $x\in [\underline{t}, \overline{t}]$:
    \begin{enumerate}[(i)]
        \item $\lim\limits_{N_A, N_B \rightarrow \infty } Var\left(\hat{\theta}(x;\hat{\varphi}(x))\mid \Pi_{N_A, N_B} \right) = 0$.
        \item $\lim\limits_{N_A, N_B \rightarrow \infty } \mathbbm{E}( n \mid \Pi_{N_A, N_B}) Var\left(\hat{\theta}(x;\hat{\varphi}(x)) \mid \Pi_{N_A, N_B} \right) \in \mathbb{R}_{++}$.
        \item  $\lim\limits_{N_A, N_B \rightarrow \infty} S_{N_A,N_B}(x)\in \mathbb{R}_{++}$, where $S_{N_A,N_B}$ is given by~(\ref{eq - S_NA_NB}).
    \end{enumerate}
\end{enumerate}
\label{Assumption 2}
\end{assumption}

\par The conditions of Assumption \ref{Assumption 2} are commonly imposed on the sampling design of surveys (e.g., see Assumption 3 of \citealp{wang2012}). Condition 1 imposes a restriction on the design so that the average sample size and the sample size of $\{Y_i,i\in U\}$ have a similar growth rate with increasing population sizes.  On Condition 2, Part (i) imposes diminishing design-variance to zero as the population sizes diverge, Part (ii) pins the rate of decrease of the design-variance in (i) with respect to the average sample size of $\{Y_i,i\in U\}$. and Part (iii)  imposes finite limiting population variance.

\par Next we describe the embedding of the sequences of finite populations under the alternative hypothesis with respect to which we compute limiting power~(\ref{asymptotic power}).
\begin{definition}
\label{sequences of alternative hypothesis}
Let $\mathcal{M}_{N_A , N_B}^1 (\varepsilon) = \left\{ \Pi \in \mathcal{M}_{N_A , N_B} : \max_{x \in [\underline{t}, \overline{t}]} \theta_{N_A,N_B}(x;\varphi_{N_A,N_B}(x)) < -\varepsilon \right\}$ for each  $N_A , N_B \in \mathbb{N}$, with given $\varepsilon>0$. This is the set of finite populations whose dominance contrasts are negative and uniformly bounded away from zero by $\varepsilon$. Let $\mathbb{W}_1 (\varepsilon)$ denote the set of all $\{ \Pi_{N_A , N_B} \}_{N_A , N_B = 1}^\infty$ such that $\Pi_{N_A , N_B} \in \mathcal{M}_{N_A , N_B}^1 (\varepsilon)$ for each $N_A, N_B \in \mathbb{N}$, satisfying Assumptions~\ref{Assumption - Compactness of X_K}, \ref{Assumption 1}, and \ref{Assumption 2}.
\end{definition}
\noindent The index set of binding inequalities enters the calculation of the limit~(\ref{asymptotic power}) for the sequences of finite populations in this setup. The following function plays a central role in characterizing the impact of this set on limiting power: for a given sequence of finite populations $\{ \Pi_{N_A, N_B} \}_{N_A, N_B = 1}^\infty$, define the function
\begin{equation}
    C(x) =  \lim_{N_A, N_B \rightarrow \infty} \frac{\theta_{N_A,N_B}(x;\varphi_{N_A,N_B}(x))}{\sqrt{Var\left(\hat{\theta}(x;\hat{\varphi}(x))\mid \; \Pi_{{N_A}, {N_B} }  \right)}} \; \mid \; \Pi_{{N_A}, {N_B} } \quad \text{for} \ x \in [\underline{t}, \overline{t}].
\label{C function}
\end{equation}

\par We have the following result.
\begin{theorem}
Fix $\alpha \in (0,1)$. Let $\mathbb{W}_1 (\varepsilon)$ be defined as in Definition \ref{sequences of alternative hypothesis} and $C(x)$ be defined as in~(\ref{C function}). The following statements hold.
\begin{enumerate}
    \item Fix $\varepsilon > 0$, then for a sequence of finite populations $\{\Pi_{N_A , N_B} \}_{N_A, N_B=1}^\infty \in \mathbb{W}_1 (\varepsilon)$:
    $$\lim\limits_{N_A , N_B \rightarrow \infty} \mathbbm{E}\left( \mathbbm{1}[LR > c(\alpha)] \, \mid \, \Pi_{N_A , N_B} \right) = 1.$$
    \item For a sequence of finite populations $\{\Pi_{N_A , N_B} \}_{N_A, N_B=1}^\infty \in \mathbb{W}_1 (0)$ define $X_0 = \{ x \in [\underline{t}, \overline{t}] : C(x) = 0\}$.
    If $X_0 \neq \emptyset$, then
    $$\lim\limits_{N_A , N_B \rightarrow \infty} \mathbbm{E}\left( \mathbbm{1}[LR > c(\alpha) ] \, \mid \, \Pi_{N_A , N_B} \right)=\text{Prob}\left[\mathbb{G}^{2}(x)>c(\alpha),\mathbb{G}(x)<0,x\in X_0\right].$$
\end{enumerate}
\label{power theorem}
\end{theorem}

\begin{proof}
See Appendix \ref{Proof of power theorem}.
\end{proof}
\section{Discussion}\label{Section - Discussion}
This section presents a discussion of the scope of our main results and implications for empirical practice. Section~\ref{Subsection - Discussion - H^1_0} discusses the interpretation of $H^1_0$. Section~\ref{Subsection - Discussion - Fakih} compares the method of this paper with that of~\cite{Tabri2021}, and Section~\ref{Subsection - Discussion - Other Tests} compares and contrasts our method to other tests in the moment inequalities inference literature. Section~\ref{Subsection - Discussion - B dominates A} describes the adjustments of our method for inferring that $B$ dominates $A$. Finally, Section~\ref{Subsection - Discussion - Longitudinal weights} puts forward an adjustment of calibration to enable pairwise comparisons with waves beyond the first one.

\subsection{Interpreting $H^1_0$}\label{Subsection - Discussion - H^1_0}

\par For a given order of stochastic dominance $s\in\mathbb{N}$ and range $[\underline{t},\overline{t}]$,  this paper's statistical procedure aims to infer $D_{N_A}^s (x) < D_{N_B}^s (x)$ for all $x\in[\underline{t},\overline{t}]$ using a paired sample from a panel survey with a maintained assumption on nonresponse. The inference relies on ranking the bounds of these dominance functions, where a rejection event implies $D_{N_A}^s (x) < D_{N_B}^s (x)$ throughout the interval $[\underline{t},\overline{t}]$ (see testing problem~(\ref{eq - boundary of null})). Under the null hypothesis $H_0^1$, an $x^*$ exists within $[\underline{t},\overline{t}]$ where $D_{N_A}^s (x^) \geq D_{N_B}^s (x^*)$, a condition that allows $D_{N_A}^s (x) < D_{N_B}^s (x)$ elsewhere within the interval. This ambiguity makes $H_0^1$ uninformative due to the partial identification of $D_{N_A}^s (\cdot)-D_{N_B}^s (\cdot)$. Failure to reject this hypothesis yields no definitive conclusions about the populations' rankings under restricted $s$th order stochastic dominance. In such a situation, we recommend empirical researchers perform a sensitivity analysis of this empirical conclusion (i.e., non-rejection of $H^1_0$) with respect to plausible assumptions on the nonresponse-generating process, using our testing procedure. Additionally, considering higher orders of dominance may be beneficial, especially in poverty and inequality analysis (e.g., \citealp{atkinson1970,atkinson1987,fostershorrocks1988b, Deaton}). Section~\ref{Section - Illustration} provides an empirical illustration using HILDA survey data.

\subsection{Comparison to~\cite{Tabri2021}}\label{Subsection - Discussion - Fakih}
\par~\cite{Tabri2021} developed a pseudo-empirical likelihood testing procedure for a test problem akin to~(\ref{eq - into test problem}), but for, first-order stochastic dominance, ordinal variables, worst-case bounds, and survey data from independent cross-sections. While their procedure is beneficial for such applications, it is narrow in scope. The worst-case bounds can be uninformative in practice, and they do not consider data arising from panel surveys to capture dynamics. Furthermore, 
the data are discrete and they focus on first-order dominance, which excludes important applications, such poverty and inequality analysis (e.g.,~\citealp{fostershorrocks1988b,Deaton}), collusion detection in industrial organization (\citealp{Aryal-Gabrielli}), and the ranking of strategies in management sciences (e.g., \citealp{Harris-Mapp,FONG20101237,minviel_benoit_2022}). See also Chapter 1 of~\cite{whang2019} and the references therein for other applications of stochastic dominance.

\par Contrastingly, this paper's setup is more complex than the setup of~\cite{Tabri2021}. Paired data from panel surveys possess more complex forms of nonresponse. Unit and item nonresponse can occur within period and wave nonresponse across periods, as well as attrition where sampled units who have previously responded permanently exit the survey. Furthermore, the testing procedure enables the incorporation of prior assumptions on nonresponse, allowing researchers to examine the informational content of their assumptions and their impacts on the inferences made. Another important difference is in the treatment of the design-effect in the statistical procedure. This paper's procedure employs an estimator of the design-effect as it must be estimated in practice. By contrast, the procedure of~\cite{Tabri2021} ignores design-effect estimation as it assumes that it is asymptotically equal to unity with uniformity --- see Condition (v) of Assumption~1 in their paper. While they do explain how to adjust their procedure to include an estimator of the design-effect, they do not explicitly account for it in the statements of their results. Reliable estimation of the design-effect becomes more important when considering high order of stochastic dominance, since asymptotically the design-variance would be based on the behaviour of random functions with powers of $s-1$, which could have a profound effect on the testing procedure.

\subsection{Comparison to Other Tests}\label{Subsection - Discussion - Other Tests}
\par This paper contributes to the vast literature on inference for parameters defined by moment inequalities. While most testing procedures in this literature assume random sampling and test for the opposite of our hypotheses our focus is different. In the context of this paper, those procedures, such as that of~\cite{ANDREWS2017275}, apply to the test problems  
\begin{align*}
H_0 :\,\max_{x\in[\underline{t},\overline{t}]}\left( \overline{D}_{N_A}^s(x) -\underline{D}_{N_B}^s(x)\right)\leq0\;\;\text{vs.}\;\; H_1: \max_{x\in[\underline{t},\overline{t}]}\left( \overline{D}_{N_A}^s(x) -\underline{D}_{N_B}^s(x)\right)>0,
\end{align*}
but under the random sampling assumption. Unlike typical approaches that infer non-dominance from a null of dominance, we posit a null of non-dominance to infer strict dominance, addressing the methodological challenge where failure to reject the null does not necessarily confirm it, unless test power is high (\citealp{davidsonduclos2013}, p. 87). This choice reflects our testing objective, which is to infer strict dominance.

\par Few tests in the moment inequalities literature consider a null of non-dominance, with notable exceptions by \cite{Tabri2021}, \cite{davidsonduclos2013}, and a test proposed by \cite{KPS} based on the minimum $t$-statistic for second-order stochastic dominance under complete data and random sampling. \cite{davidsonduclos2013} also adapt it for first-order stochastic dominance. Our Supplementary Material’s Lemmas~\ref{Rami Lemma 1} and~\ref{Rami Lemma 2} demonstrate that under the null hypothesis $H_0^1$, the statistic $2(L_{UR}-L_{R}(x))/\widehat{\text{Deff}}(x)$ is asymptotically equivalent to the square of a $t$-statistic. Additionally, Lemmas~\ref{power theorem Lemma 1} and~\ref{power theorem Lemma 2} confirm analogous results under local alternatives, establishing the asymptotic local equivalence of our pseudo-empirical likelihood test and the test using the minimum of the squared $t$-statistics. Therefore, practitioners can also use the minimum $t$-statistic instead of our pseudo-empirical likelihood ratio statistic~(\ref{LR test statistic}) and obtain asymptotically valid tests by comparing it to a quantile from the $\chi^2_1$ distribution using the same decision rule~(\ref{eq - decision rule}).

\par Moreover, econometricians have also considered design-based approaches to inference for Lorenz dominance but for particular survey designs (e.g.,~\citealp{Zheng,BHATTACHARYA2005145,BHATTACHARYA2007674}). While Lorenz dominance tests are similar in spirit to various stochastic dominance tests in the literature, their treatment is different as their test statistics are more complicated functionals of the underlying distribution functions. Importantly, these testing procedures have demonstrated that complex sampling designs can greatly affect the estimates of standard errors required for inference. However, they are applicable to complete data, which is their main drawback, as nonresponse is inevitable in socioeconomic surveys.

\subsection{Testing $B$ Dominates $A$}\label{Subsection - Discussion - B dominates A}
In our framework, populations $A$ and $B$ represent wave 1 and a subsequent wave, respectively, and interest centers on establishing that $A$ dominates $B$ at a predesignated order $s$ and over a range $[\underline{t},\overline{t}]$. The practitioner may also be interested in establishing that $B$ dominates $A$, which entails considering the testing problem
\begin{align}\label{eq - Testing Problem B Dom A}
H_0:\max_{x\in[\underline{t},\overline{t}]}\left(D_{N_B}^{s}(x)-D_{N_A}^{s}(x)\right)\geq0\;\;\text{vs.}\;\; H_1:\max_{x\in[\underline{t},\overline{t}]}\left(D_{N_B}^{s}(x)-D_{N_A}^{s}(x)\right)<0.
\end{align}
The framework of Section~\ref{Section - Setup} covers this scenario through specifying $\varphi_{N_A,N_B}$ appropriately to obtain the sharp upper bound $\overline{D}^s_{N_B}(\cdot)-\underline{D}^s_{N_A}(\cdot)$ on the contrast $D_{N_B}^{s}(\cdot)-D_{N_A}^{s}(\cdot)$, so that estimation and testing can proceed as described above.

\par We repeat the same examples as in Section~\ref{Subsection Examples} but for the testing problem~(\ref{eq - Testing Problem B Dom A}), and relegate those details to Appendix~\ref{Subsection - Appendix D - B Dom A} for brevity. The technical derivations of the bounds in each example follow steps similar to the ones in the examples of Section~\ref{Subsection Examples}. The general idea with this set of examples is to carefully specify $\varphi_{N_A,N_B}$ to obtain the desired form of $\theta_{N_A,N_B}(\cdot)=\overline{D}^s_{N_B}(\cdot)-\underline{D}^s_{N_A}(\cdot)$.

\subsection{Comparisons of Two Waves Beyond Wave 1}\label{Subsection - Discussion - Longitudinal weights} 

\par This section presents a discussion of how to carry out inference for stochastic dominance with two populations beyond the first wave within our framework. Firstly, it points to a limitation of a leading practice in the literature called calibration that has been used to address nonresponse. Second, it proposes a bounds approach to inference for stochastic dominance based on restricted tests (\citealp{Aitchison-1962}). This approach combines the basic idea of calibration with the framework of the previous sections in the paper.

\par A longitudinal weight between these waves must be utilized instead of the design weights to compare populations of two subsequent waves beyond the first. The idea behind longitudinal weights is to reflect the finite population at the initial wave of the comparison by accounting for the dynamic nature of the panel, where the achieved sample in wave 1 evolves with time due to following rules and population changes. The process of calibration estimation obtains such weights. The seminal paper of~\cite{Deville_Sarndal-1992} formalized the general concept and techniques of calibration estimation in the context of survey sampling with complete data to improve estimators' efficiency and ensure coherency with population information. The idea of using extra information to enhance inference has a long tradition in statistics and econometrics; see, for example, ~\cite{Aitchison-1962,Imbens-Lancaster-1994,PARENTE20171}.

\par The calibration estimator uses calibrated weights, which are as close as possible, according to a given distance measure, to the original sampling design weights while also respecting a set of constraints, which represent known auxiliary population benchmarks/information, to ensure the resulting weighted estimates match (typically external) high-quality totals or means of the initial wave in the comparison.\footnote{See~\cite{Wu-Lu-2016} for examples of distance measures used in calibration.} This method is ideal for incorporating such information in the setup with complete responses and a representative sample of the target population.

\par Calibration is also widely applied under nonresponse, but with the achieved sample of responding units -- a procedure that is known as \emph{weighting} (see Chapter 5 of~\citealp{Sarndal-Lundstrom}). Under nonresponse, weighting thus relies on the assumption that the supports of the outcome variable conditional on response and nonresponse coincide. While suitable for specific empirical applications, this assumption can be untenable in other applications. It is implausible when there are reasons to expect that the missing units tend to belong to a subpopulation of the target population. For example, a widely shared view of household surveys is that the missing incomes correspond to households at the top of the income distribution (e.g.,~\citealp{Lustig-2017, Bourguignon-2018}). The issue is that the achieved sample on the outcome variable (i.e., of responding units) in the first wave is biased because it represents a subpopulation of the target population. Following the units of this sample over time could perpetuate their selectivity bias in subsequent waves, where the known auxiliary population benchmarks/information in the subsequent waves may not be compatible with the subpopulation that the achieved sample represents. Consequently, calibrating design weights using the achieved (biased) sample can result in misleading inferences.

\par The challenge with calibration in our missing data setup is that the achieved sample is generally not representative of the target population. The achieved sample could represent a subpopulation of the target population, which may not satisfy the auxiliary information. If there are known auxiliary benchmarks/information on that subpopulation, then the design weights can be calibrated using this sample and auxiliary benchmarks/information. This calibration approach can be combined with our bounds approach of the previous sections to develop a statistical procedure for inference on stochastic dominance for the target populations under nonresponse. The known auxiliary benchmarks/information on the subpopulations may be derived from or implied by their counterparts on the target population in the form of subpopulation totals or means, for instance.



\par More concretely, $A$ now represents the population of the initial wave in the comparison and $B$ now represents the later wave. The subpopulation $\mathcal{P}_{N_{A_1}}= \left\{ \{Y^A_i, Z^A_i \} \in \mathcal{P}_{N_A} : Z_i^{W_1}=1\right\},$ is the one the sample represents, where $Z_i^{W_1}$ is the 0/1 binary variable indicating on response in the population corresponding to the first wave of the survey. Let $\{G_i:i=1,\ldots, N_{A_1}\}$ where $G_i\in\mathbb{R}$ be the population value of the auxiliary variable for each $i$, and $N_{A_1}$ is the subpopulation total. The information we have is the constraint $\frac{1}{N_{A_1}}\sum_{i=1}^{N_{A_1}}G_i=G_0$, but because of nonresponse, the data on this auxiliary variable will have missing values, which would be the units $\{i\in U: Z^{A_1}_i=0\}$. If there are known bounds on $G_i$, i.e., $\underline{G}\leq G_i\leq \overline{G}$ for all $i$ in the subpopulation $A_1$, then using the aforementioned bounds, the following inequality restrictions must hold
\begin{align}
\frac{1}{N_{A_1}}\left(\sum_{i=1:Z^{A_1}_i=1}^{N_{A_1}}G_i+\sum_{i=1:Z^{A_1}_i=0}^{N_{A_1}}\overline{G}\right) \geq G_0\,\text{and}\,\frac{1}{N_{A_1}}\left(\sum_{i=1:Z^{A_1}_i=1}^{N_{A_1}}G_i+\sum_{i=1:Z^{A_1}_i=0}^{N_{A_1}}\underline{G}\right)\leq G_0.\label{eq - WC bounds Aux}
\end{align}
Employing the information~(\ref{eq - WC bounds Aux}) in testing entails the consideration of the \emph{restricted} testing problem:
\begin{equation}\label{eq - test problem Longitudinal}
\begin{aligned}
H_0^2:& \max_{x\in[\underline{t},\overline{t}]}\left(\overline{D}_{N_A}^{s}(x)-\underline{D}_{N_B}^{s}(x)\right)\geq0,~\text{and}~(\ref{eq - WC bounds Aux}) \quad \text{versus}\\
H^2_1:& \max_{x\in[\underline{t},\overline{t}]}\left(\overline{D}_{N_A}^{s}(x)-\underline{D}_{N_B}^{s}(x)\right)<0\,~\text{and}~(\ref{eq - WC bounds Aux}),
\end{aligned}
\end{equation}
where the bounds, $\overline{D}_{N_A}^{s}$ and $\underline{D}_{N_B}^{s}$, account for nonresponse through the use of side information or a maintained assumption on nonresponse.

\par A testing procedure using the method of pseudo-empirical likelihood is feasible since it can implement the information~(\ref{eq - WC bounds Aux}) as additional constraints in formulating the pseudo-empirical likelihood-ratio statistic. These constraints are imposed under both the null and alternative hypotheses for \emph{internal consistency} (\citealp{Wu-Lu-2016}), so that the transformed design weights are compatible with subpopulation $A_1$. We conjecture the asymptotic form of the test statistic is equivalent to generalized likelihood ratio statistic for a cone-based testing problem on a multivariate normal mean vector (e.g., Theorem~3.4 of~\citealp{Raubertas1986}). There are challenges with implementing the testing procedure, as the subset of inequalities in~(\ref{eq - WC bounds Aux}) that are active/binding enters the design effect; see Appendix~\ref{Appendix E Longitudinal weights} for the details. This set is unknown and difficult to estimate reliably, which can profoundly impact the overall reliability of the testing procedure. It renders the distribution as non-pivotal, which creates challenges for calculation of its quantiles. One possible avenue forward is to adapt the bootstrap procedure of~\cite{Wang-Peng-Kim-2022}
to the setup of inequality restrictions in order circumvent this computational difficulty --  we leave it for future research.

\section{Empirical Application}\label{Section - Illustration}
The empirical analysis aims to investigate the temporal orderings of poverty among Australian households using data from the HILDA panel survey. This survey's design, described in detail by~\cite{watson2002}, is complex, involving multiple stages, clustering, stratification, and unequal probability sampling. The unit of observation is the household, which means that each household contributes one data point to the dataset. We use equivalized household net income (EHNI) based on the OECD's equivalence scale as the measure of material resources available to a household .\footnote{This scale assigns a value of 1 to the first household member, of 0.7 to each additional adult and of 0.5 to each child.}~The equivalization adjustments of households' net incomes are essential for meaningful comparisons between different types of households, as it accounts for variations in household size and composition and considering the economies of scale that arise from sharing dwellings.

\par The HILDA panel survey has been conducted annually since 2001 and provides reports that present selected empirical findings on Australian households and individuals across the survey waves. A recent issue,~\cite{HILDA_Stat_Report_2022}, indicates an upward trend in the median and average EHNI. This suggests a steady decline in household poverty between 2001 and 2022, as shown in Figures 3.1 and 3.2 of their report. However, this empirical evidence is rather weak and potentially misleading as it does not consider (i) household poverty lines, which concerns the lower tails of the EHNI distributions, and (ii) a poverty index.

\par Our empirical analysis utilizes restricted stochastic dominance (RSD) orderings to effectively rank distributions of households'EHNI in terms of poverty. The concept of poverty orderings based on RSD conditions was introduced by~\cite{fostershorrocks1988b} to provide robust comparisons of poverty. This approach overcomes the challenge of defining a single poverty line to identify and categorize households as poor. Instead, it compares distributions across a range of poverty lines. In this context, the range of poverty lines corresponds to different levels of EHNI for various types of households. A household is considered to be in poverty if its EHNI falls below its corresponding poverty line. By employing RSD orderings and considering multiple poverty lines, our analysis offers a more comprehensive and nuanced understanding of poverty and avoids the limitations of relying on a single threshold to define poverty.

\par The robust comparison of households' EHNI distributions using the HILDA survey data is complicated by the presence of nonresponse, because the distributions are only partially identified. In producing their annual reports, the survey designers have implemented the following assumptions on nonresponse. They address unit nonresponse through re-weighting responding households by distributing the weight of nonresponding households to other like responding households, assuming nonrespondents are the same as respondents. In other words, they implement the MCAR assumption for unit nonresponse -- see~\cite{watsonfry2002} for more details. Furthermore, they address wave and item nonresponse through imputation.

\par The assumptions made by the survey designers on nonresponse, specifically unit and wave nonresponse, are implausible in practice. There is descriptive empirical evidence showing nonresponse rates tend to be highest at both ends of the income distribution. For example, logistic regression analyses demonstrate a U-shaped relationship between unit nonresponse and median weekly household incomes (see Table~A3.1 in~\citealp{watsonfry2002}) and wave nonresponse is also U-shaped with respect to individual EHNI (see Table~10.2 in~\citealp{Watson-Wooden.ch10}). However, the imputation of missing household incomes in the case of item nonresponse is considered reliable due to the availability of comprehensive information on partially responding households' characteristics and the income components of responding individuals. This additional information allows for reasonably accurate imputations -- see~\cite{watson2004} for an example of their imputation method for missing income data in wave 2.

\par The empirical analysis focuses on the first four waves, and uses the forgoing evidence on nonresponse to establish bounds on dominance contrasts, which are then implemented in our statistical procedure. The rest of this section is organized as follows. Section~\ref{Subsection - Emp Sec - Descrip} describes the dataset and Section~\ref{Subsection - Emp Sec - Results} presents our results with a discussion.

\subsection{Description of Dataset}\label{Subsection - Emp Sec - Descrip}

\par We have focused on pairwise comparisons between the first wave and waves 2, 3, and 4. In the notation introduced earlier, $A$ represents the population of Australian households in 2001, which corresponds to wave 1. $B$ would then represent either of the population of Australian households in the years 2002, 2003, and 2004, corresponding to waves 2, 3, and 4, respectively.

\par In the first wave of the HILDA survey, there were 12252 addresses issued which resulted in 804 addresses being identified as out of scope (as they were vacant, non-residential, or all members of the household were not living in Australia for 6 months or more). In addition, there were 245 households added to the sample due to multiple households living at one address. This resulted in $11693$ in-scope households of which 7682 responded. These responding households are followed through time.

\par Households can grow, split, and dissolve over time, which creates unbalanced panels. As the split households's were followed, to form a balanced panel, we double counted the split households, but equally distributed their design weights. For example, between waves 1 and 2, 712 of (responding) households in wave 1 split. This approach yields a balanced panel with $n=7682+712=8395$ unit responding households a from a sample size of $k=11693+712=12405$. Table~\ref{Table Nonresponse} reports the values of $k$ and $n$ for all of the wave pairs based on this construction of the balanced panel. There are other balancing approaches and we do not take a stand on which one to favour; see, for example, the "fare shares approach" described in~\cite{taylor2010} for a different re-balancing method.


\par Now we provide details on estimation of the fractions $\delta_{00}$, $\delta_{10}$, and $\delta_{11}$. A standard approach uses the survey design weights $\{\omega_i: i\in V\}$ in their estimation. However, the HILDA survey does not provide the subset of those weights corresponding to unit nonresponding households. This means the practitioner only has access to $\{\omega_i: i\in U\}$ and not $\{\omega_i: i\in V-U\}$. Furthermore, the survey team has scaled the provided weights so that $\sum_{i\in U}\omega_i=N_A$, where $N_A=7404297$ is the total number of Australian households in 2001. Despite this setup, we can still estimate $\delta_{00}$, $\delta_{10}$, and $\delta_{11}$ as follows. Estimate $\delta_{00}$ as $\hat{\delta}_{00}=(k-n)/k$, which is an unweighted estimator of $\delta_{00}$ on account of not having the complete set of weights. Since the provided weights satisfy $\sum_{i\in U}\omega_i=N_A$, we re-scale them so that the resulting weights $\{\omega^\prime_i: i\in U\}$ sum to $(n/k)N_A$. Then we estimate $\delta_{10}$ and $\delta_{11}$ as $\hat{\delta}_{10}=N^{-1}_A\sum_{i\in U:Z^B=0}\omega^\prime_i$ and $\hat{\delta}_{11}=1-\hat{\delta}_{10}-\hat{\delta}_{00}$, respectively. Table~\ref{Table Nonresponse} reports the estimates of these fractions for all wave pairs.

\begin{table}[pt]
\caption{Preliminary Calculations}\label{Table Nonresponse}
\centering
\bigskip
\resizebox{15cm}{!}{
\begin{tabular}{c*{8}{c}}
Population & $k$   & $n$    & $\hat{\delta}_{00}$ & $\hat{\delta}_{10}$  & $\hat{\delta}_{11}$ & $[\underline{t},\overline{t}]$ & Deflation Factor \\
\hline \\
Waves 1 and 2 & 12405    &  8395          & 0.3233       & 0.0855           & 0.5912 & $[6917.01,14872]$ & 1.03\smallskip\\
         \hline \\
Waves 1 and 3 & 12876    &  8865         &  0.3115      & 0.1275          & 0.561  & $[6622,14872]$ & 1.06 \smallskip\\
         \hline \\
Waves 1 and 4 & 13255   &  9244         & 0.3026     & 0.1601           & 0.5373 & $[6917.01,15262]$ & 1.08
 \smallskip\\
         \hline
\end{tabular}
}
\end{table}

\begin{figure}[t]
\centering
\includegraphics[width=16cm, height=10cm]{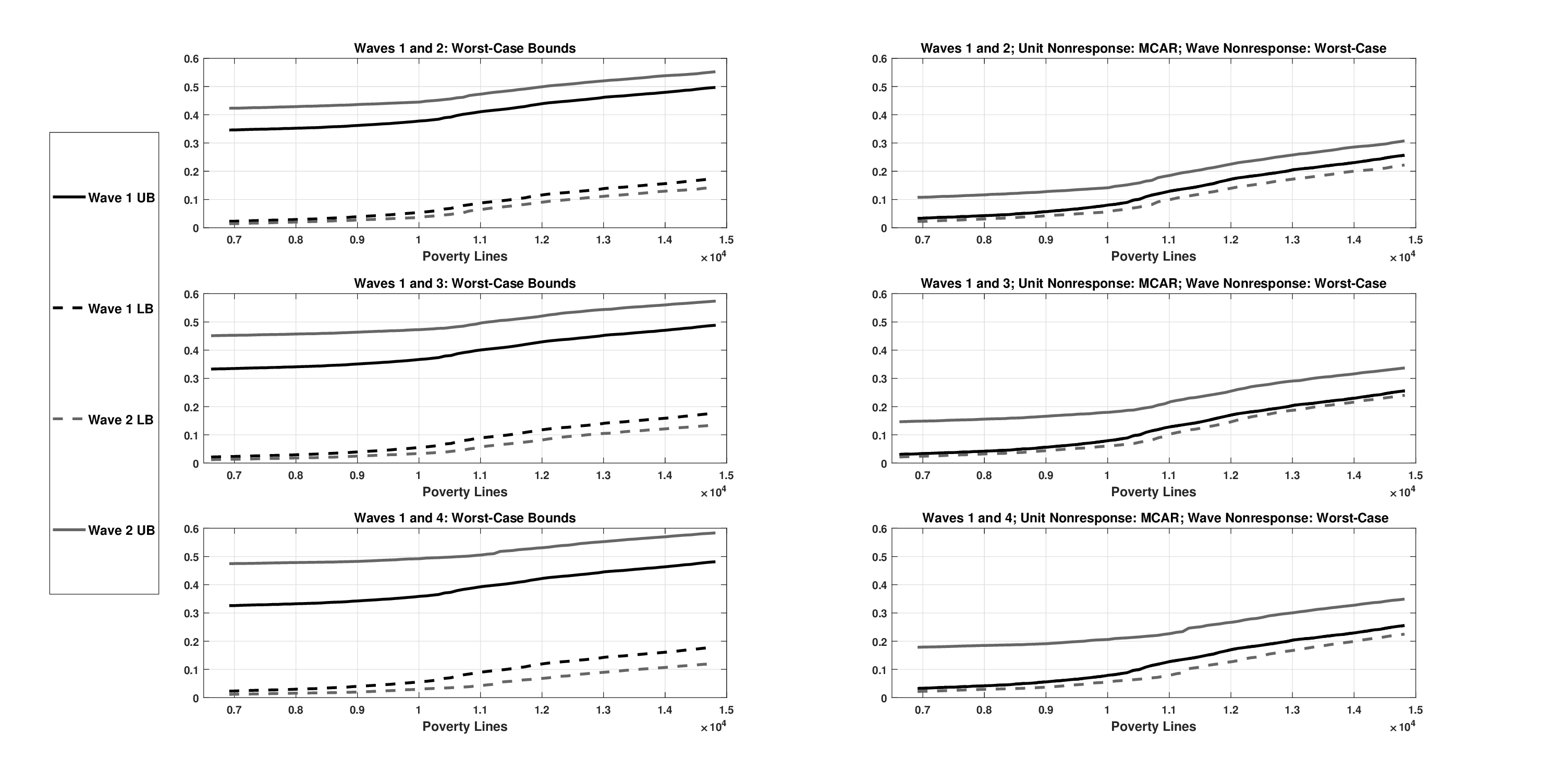}
\caption{Figures report bounds on $D^1_{N_A}$ and $D^1_{N_B}$ over the range of poverty lines under two sets of assumptions on nonresponse. The left figure reports the worst-case bounds. The right figure reports the bounds under the MCAR assumption on unit nonresponse and worst-case scenario on wave nonresponse.}\label{Figure: WC bounds and MCAR UNR}
\end{figure}

\par The weights $\{W^\prime_i,i\in U\}$ we use in the statistical procedure and in the derivation of our results are obtained by re-scaling $\{\omega^\prime_i: i\in U\}$ so that they sum to $n$. The range of poverty lines we consider are reported in Table~\ref{Table Nonresponse}, all denominated in 2001 AUD using the deflation factors obtained from the Reserve Bank of Australia's inflation calculator.\footnote{The inflation calculator's URL is \text{https://www.rba.gov.au/calculator/}.} The poverty line ranges are based on the tables reported in~\cite{Wilkins2001,Wilkins2002,Wilkins2003,Wilkins2004} after equivalizing them according to the OECD equivalence scale.

\par We do not expect the EHNI population distribution to change rapidly across the first 4 waves. We encode this restriction by setting $\mathcal{X}_A=\mathcal{X}_B$ for $B\in\{\text{wave 2, wave 3, wave 4}\}$. Furthermore, we have set $\mathcal{X}_A=\mathcal{X}_B=[-150000,1000000]$, where the upper and lower bounds are larger and smaller, respectively, than the observed values, because of the high incidence of unit nonresponse. Of course, this assumption is irrefutable; however, it is credible since there is evidence from logistic regression analyses reported in Table~A3.1 of~\cite{watsonfry2002} demonstrating a U-shaped relationship between unit nonresponse and median weekly household incomes of different neighborhoods. In consequence, the EHNI of unit nonresponders is likely to be in the tails of the EHNI population distribution. In practice, one can also study the sensitivity of outcomes based on it.

\par Figure~\ref{Figure: WC bounds and MCAR UNR} reports estimates of the bounds on the dominance functions $D^1_{N_A}$ and $D^1_{N_B}$ under two sets of assumptions on nonresponse for each pair of waves in our study. The left figure depicts their worst-case bounds. These bounds summarize what the data, and only the data, say about $D^1_{N_A}$ and $D^1_{N_B}$. They are instructive since it establishes ``a domain of consensus among researchers who may hold disparate beliefs about what assumptions are appropriate'' (\citealp{HOROWITZ2006445}). However, they are not informative in our setup because we find the identified set of $D^1_{N_A}$ is a proper subset of its $D^1_{N_B}$ counterpart for each pair of waves. The right panels of the figure depict the bounds under the MCAR assumption on unit nonresponse without any assumption on wave nonresponse. The HILDA survey implements this assumption on unit nonresponse in its data releases. It is a very strong assumption that point-identifies $D^1_{N_A}$. Without any assumptions on wave nonresponse, $D^1_{N_B}$ is only partially identified. As with the worst-case bounds, these bounds are also not informative because the point estimate of $D^1_{N_A}$ is an element of $D^1_{N_B}$'s identified set, for each pair of waves.

\par Finally, to approximate the design effect~(\ref{eq - deff}), we have used the jackknife and the replication design weights (provided by the survey's release). See~\cite{hayes2008} for a general description of this procedure to calculate standard errors of estimators using HILDA survey data.

\subsection{Results}\label{Subsection - Emp Sec - Results}
\par We would like to evaluate the dynamics of poverty between waves 1 and 2, waves 1 and 3, and waves 1 and 4, using the \emph{headcount ratio} for each poverty line in their respective range $[\underline{t},\overline{t}]$. Therefore, the testing problem of interest is given by~(\ref{eq - Testing Problem B Dom A}) with $s=1$; that is
\begin{align}\label{eq - Test Problem Empirical s1}
H_0 :\max_{x\in[\underline{t},\overline{t}]}\left(\overline{D}_{N_B}^{1}(x)-\underline{D}_{N_A}^{1}(x)\right)\geq0\;\;\text{vs.}\;\; H_1  :\max_{x\in[\underline{t},\overline{t}]}\left(\overline{D}_{N_B}^{1}(x)-\underline{D}_{N_A}^{1}(x)\right)<0.
\end{align}

\par This section reports the results of a sensitivity analysis of the event ``Reject $H_0$'' in the forgoing test problem for each pair of waves at the 5\% significance level, using nonresponse assumptions based on Examples~\ref{Example - NBD of MCAR} and~\ref{Example - Kline and Santos}, which present alternative perspectives on how missing data differ from observed outcomes. In those examples, the focus was on the contrast $x\mapsto D^1_{N_A}(x)-D^1_{N_B}(x)$, and in this testing problem, we must consider the contrast $x\mapsto D^1_{N_B}(x)-D^1_{N_A}(x)$ instead. Using the results of Propositions~\ref{prop - NBD of MCAR ID Set} and~\ref{prop - Kline Santos ID Set}, which report the forms of $\overline{D}_{N_B}^{1}$ and $\underline{D}_{N_A}^{1}$ for these two types of neighborhood assumptions, we can derive the corresponding identified sets of the contrast of interest.

\subsubsection{Neighborhood of MCAR: Kolmogorov-Smirnov Distance}
\par The result of Proposition~\ref{prop - Kline Santos ID Set} delivers the forms of $\overline{D}_{N_B}^{1}$ and $\underline{D}_{N_A}^{1}$. In particular, for each $x\in[\underline{t},\overline{t}]$ and  $\gamma_A,\gamma^{00}_B,\gamma^{10}_B\in[0,1]$, the dominance functions are given by$ \underline{D}_{N_A}^{1}(x)= F_A\left(x| Z^A=Z^B=1\right)\,(\delta_{11}+(1-\gamma_A)\delta_{00})+F_A\left(x| Z^A=1,Z^B=0\right)\,\delta_{10}$ and $\overline{D}_{N_B}^{1}(x)=  F_B\left(x| Z^A=Z^B=1\right)\,(\delta_{11}+(1-\gamma^{00}_B)\delta_{00}+(1-\gamma^{10}_B)\delta_{10})+(\gamma^{00}_B \delta_{00}+\gamma^{10}_B\delta_{10})$.

\par Next, we report results of the sensitivity analysis of the  event ``Reject $H_0$'' in the test problem~(\ref{eq - Test Problem Empirical s1}) for each pair of waves, with respect all values of the triple $(\gamma_A,\gamma^{00}_B,\gamma^{10}_B)$ in the 3-dimensional grid $\Xi=\{0,0.1,0,2,\ldots,0.9,1\}^3$. Let $\Xi^{1,j}$ denote the subset of $\Xi$ where this rejection of the null occurs for the waive pairs $1$ and $j$, for $j=2,3,4$. Interestingly, we did not find considerable deviations from the MCAR assumptions at the 5\% significance level, as the tests did not reject this null hypothesis for the majority of points in $\Xi$ for each pair of waves. Specifically, we have obtained $\Xi^{1,2}=\Xi^{1,3}=\left\{(0,0,0), (0.1,0,0)\right\}$ and $\Xi^{1,4}=\left\{(0,0,0),(0.1,0,0),(0.2,0,0),(0.3,0,0)\right\}$, where the MCAR assumption arises under the parameter specification $(\gamma_A,\gamma^{00}_B,\gamma^{10}_B)=(0,0,0)$. The implication is that the decline in poverty over time is not robust to deviations from the MCAR assumption with respect to the Kolmogorov-Smirnov distances between $F_K\left(\cdot| Z^A=Z^B=1\right)$ and $F_K\left(\cdot| Z^A=Z^B=0\right)$ for $K=A,B$, and between $F_B\left(\cdot| Z^A=Z^B=1\right)$ and $F_B\left(\cdot| Z^A=1,Z^B=0\right)$.

\par The non-rejections arise because $H_0$ in~(\ref{eq - Test Problem Empirical s1}) occurs in the sample, forcing the test statistic to equal zero. Therefore, we cannot conclude anything informative about the ranking of the two EHNI distributions using first-order restricted stochastic dominance, for $(\gamma_A,\gamma^{00}_B,\gamma^{10}_B)$ in $\Xi-\Xi^{1,j}$ and $j=2,3,4$. Hence, we have also considered sensitivity with respect to the poverty index by testing using second-order restricted dominance (i.e., $s=2$) and the values of $(\gamma_A,\gamma^{00}_B,\gamma^{10}_B)$ in $\Xi-\Xi^{1,j}$ for $j=2,3,4$. Ranking EHNI distributions using second-order restricted dominance corresponds to ranking them robustly according to the \emph{per capita income gap} poverty index. Using the same significance level, the tests did not reject this null hypothesis for the majority of points in $\Xi-\Xi^{1,j}$ and $j=2,3,4$. Let $\Xi^{1,j}_2\subset \Xi-\Xi^{1,j}$ be the set where this rejection of the null occurs for the waive pairs $1$ and $j$, for $j=2,3,3,4$. They are given by $\Xi^{1,2}_2=\Xi^{1,2}_2=\left\{(0.2,0,0), (0.3,0,0)\right\}$ and $\Xi^{1,4}_2=\left\{(\gamma_A,0,0):\gamma_A=0.4,0.5,0.6,0.7\right\}$. The situation is similar to the case of first-order restricted dominance above: the decline in poverty over time, now with respect to per capita income gap, is not robust to deviations from the MCAR assumption with respect to the Kolmogorov-Smirnov distances between $F_K\left(\cdot| Z^A=Z^B=1\right)$ and $F_K\left(\cdot| Z^A=Z^B=0\right)$ for $K=A,B$, and between $F_B\left(\cdot| Z^A=Z^B=1\right)$ and $F_B\left(\cdot| Z^A=1,Z^B=0\right)$.

\subsubsection{A U-Shape Restriction on Nonresponse Propensities}

\par This section's sensitivity analysis refines the worst-case bounds on the contrasts $\{D^1_{N_B}(x)-D^1_{N_A}(x),x\in [\underline{t},\overline{t}]\}$ using the descriptive empirical evidence on the U-shaped form of the unit and wave nonresponse propensities with respect to income measures. This evidence on nonresponse are from exploratory logistic regressions whose estimation output is reported in Table~A3.1 in~\cite{watsonfry2002} and Table~10.2 in~\cite{Watson-Wooden.ch10} for unit and wave nonresponse, respectively. Furthermore, they have found no statistically significant relationship with the income measures, which can be considered as weak evidence for this U-shaped form. It should be noted, however, that the validity of their tests rely on the correct specification of their model, which is likely misspecified. Thus, to take account of their findings and the likely misspecification of their model, we implement the U-shaped restriction on $\text{Prob}\left(Z^A=Z^B=0| Y^A=x\right)$ and $\text{Prob}\left(Z^A=1,Z^B=0| Y^B=x\right)$ through the structure described in Example~\ref{Example - NBD of MCAR}.

\par For each $x\in\mathcal{X}_K$ and $K\in\{A,B\}$, an application of Bayes' Theorem to the conditional probability $P\left(Z^A=Z^B=0| Y_K\leq x\right)$ reveals its form as $\frac{1}{F_K(x)}\int_{-\infty}^{x}P\left(Z^A=Z^B=0| Y_K=v\right)\,dF_K(v)$,
with a similar expression for $\text{Prob}\left(Z^A=1,Z^B=0| Y^B\leq x\right)$. Therefore, we can encode this shape information on the nonresponse propensities using the bounds
\begin{align}
 \delta_{00}\,L^{K}_{00}(x) & \leq \text{Prob}\left(Z^A=Z^B=0| Y^K\leq x\right)\leq U^{K}_{00}(x) \,\delta_{00},\quad K=A,B,\text{and}\label{eq - Prop Score 1} \\
 \delta_{10}\,L^B_{10}(x) & \leq \text{Prob}\left(Z^A=1,Z^B=0| Y^B\leq x\right)\leq U^B_{10}(x)\,\delta_{10},\label{eq - Prop Score 2}
\end{align}
where $L^{A}_{00}$, $U^{A}_{00}$, $L^{B}_{00}$, $U^{B}_{00}$, $L^{B}_{10}$, and $U^B_{10}$ are CDFs on the common support $[-150000,10^6]$, satisfying $L^{K}_{00}(\cdot)\leq U^{K}_{00}(\cdot)$ for $K=A,B,$ and $L^{B}_{10}(\cdot)\leq U^{B}_{10}(\cdot)$, and have U-shaped densities. Proposition~\ref{prop - NBD of MCAR ID Set} describes the bounds on the contrast implied by this assumption for $s\in\mathbb{Z}_+$. For $s=1$, they are given by
\begin{align}
\underline{D}_{N_A}^{1}(x) & =  F_A\left(x| Z^A=Z^B=1\right)\,\frac{\delta_{11}}{1-L^{A}_{00}(x)\delta_{00}}\nonumber\\
 & \qquad+F_A\left(x| Z^A=1,Z^B=0\right)\,\frac{\delta_{10}}{1-L^{A}_{00}(x)\delta_{00}}\quad\text{and}\label{eq - Emp Sec LB D_NA}\\
\overline{D}_{N_B}^{s}(x) & = F_B\left(x| Z^A=Z^B=1\right)\,\frac{\delta_{11}}{1-U^{B}_{00}(x)\delta_{00}-U^{B}_{10}(x)\delta_{10}},\label{eq - Emp Sec UB D_NB}
\end{align}
for each $x\in[\underline{t},\overline{t}]$. Within our estimating function approach, they arise under the following specification of $\varphi_{N_A,N_B}$: for each $x\in[\underline{t},\overline{t}]$, $\varphi_1(x)=\varphi_2(x)=-\frac{\delta_{11}+\delta_{10}}{1-L^{A}_{00}(x)\delta_{00}}$, $\varphi_3(x)=0$, and $\varphi_4(x)=-\frac{\delta_{11}+\delta_{10}}{1-U^{B}_{00}(x)\delta_{00}-U^{B}_{10}(x)\delta_{10}}$.

\begin{sidewaysfigure}
\centering
\subfloat[Waves 1 and 2]{\includegraphics[width=7cm, height=7cm]{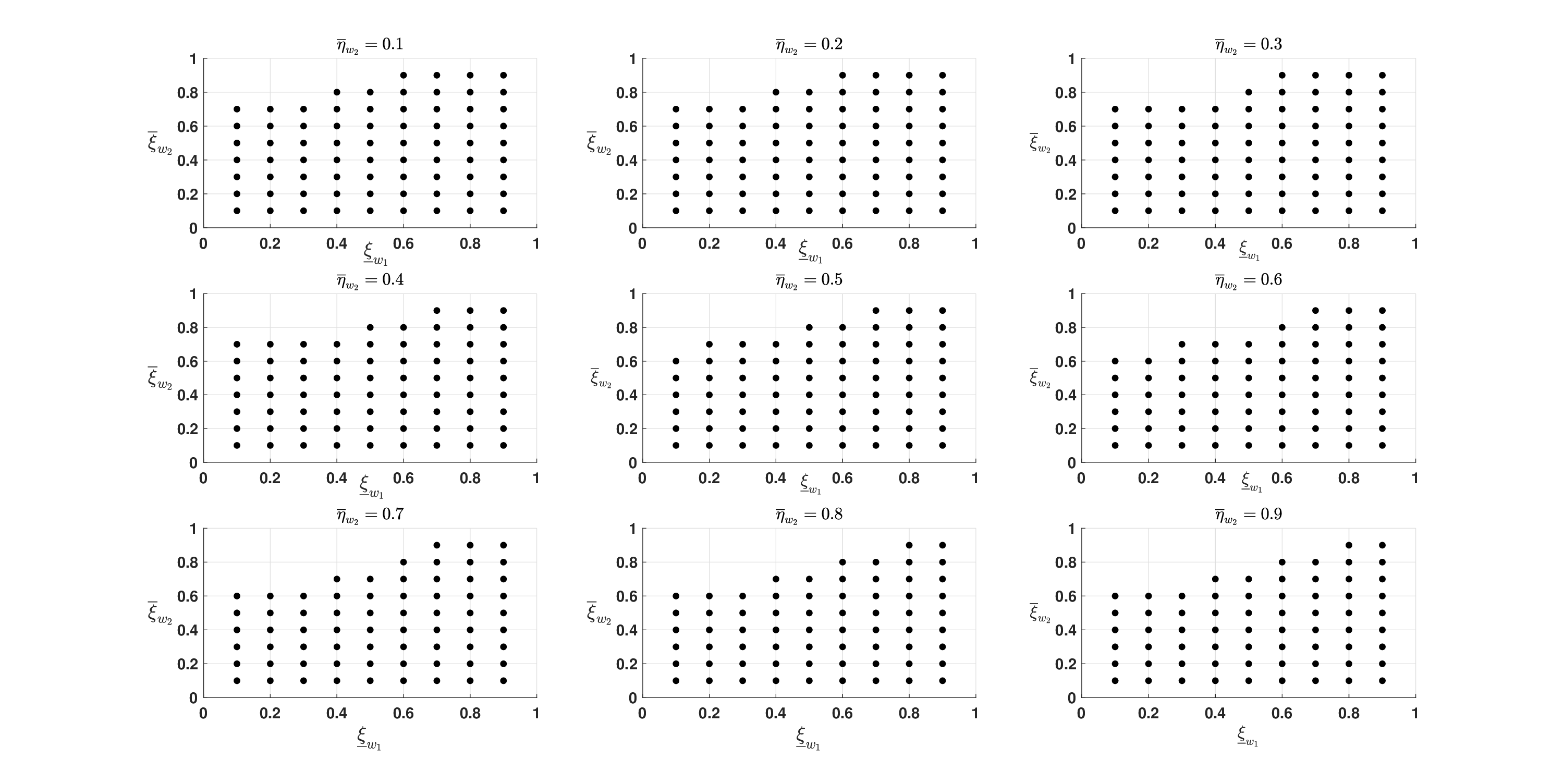}\label{Figure: Sensitivity Waves 1 and 2}}\quad
\subfloat[Waves 1 and 3]{\includegraphics[width=7cm, height=7cm]{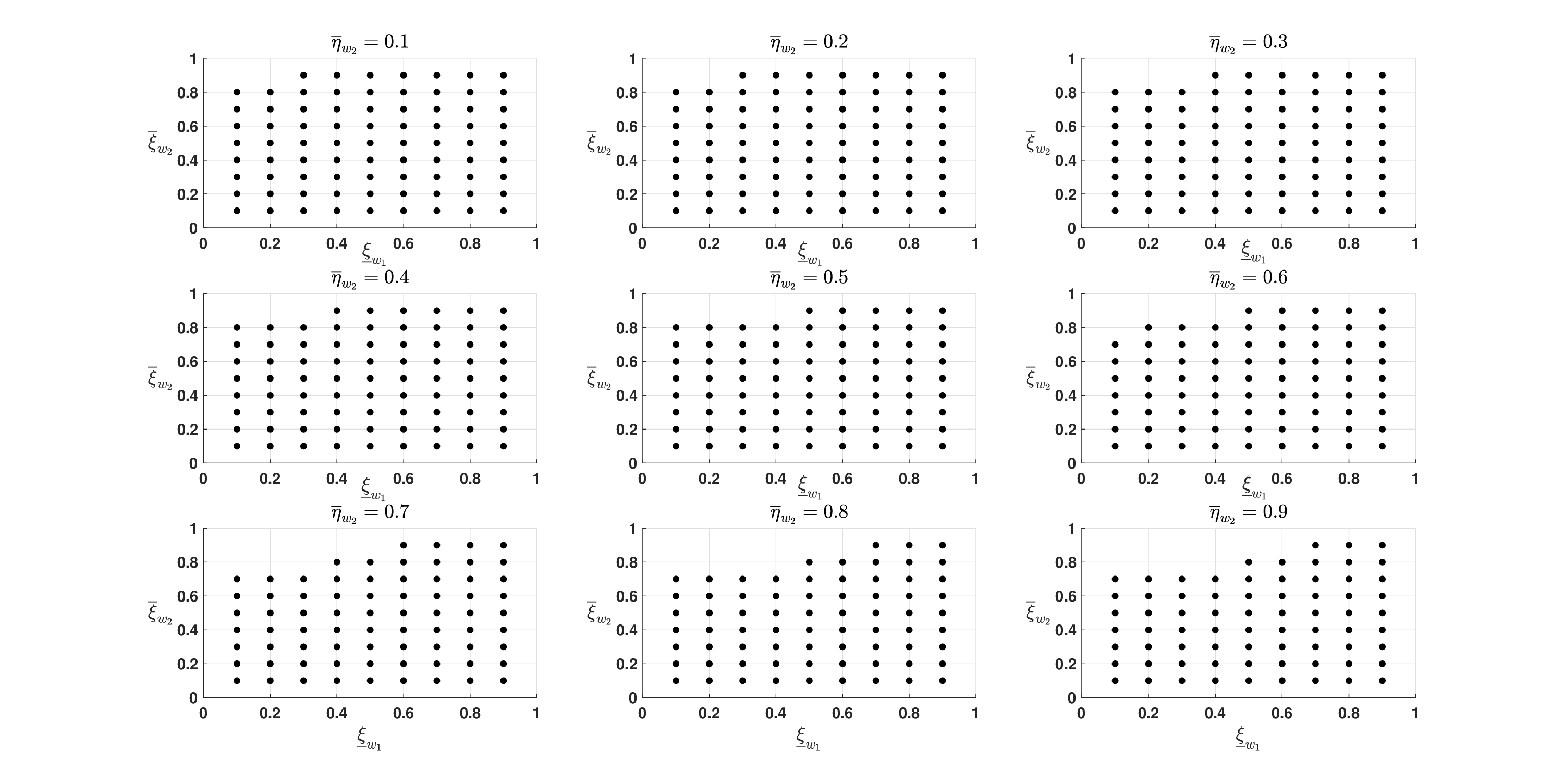}\label{Figure: Sensitivity Waves 1 and 3}}\quad
\subfloat[Waves 1 and 4]{\includegraphics[width=7cm, height=7cm]{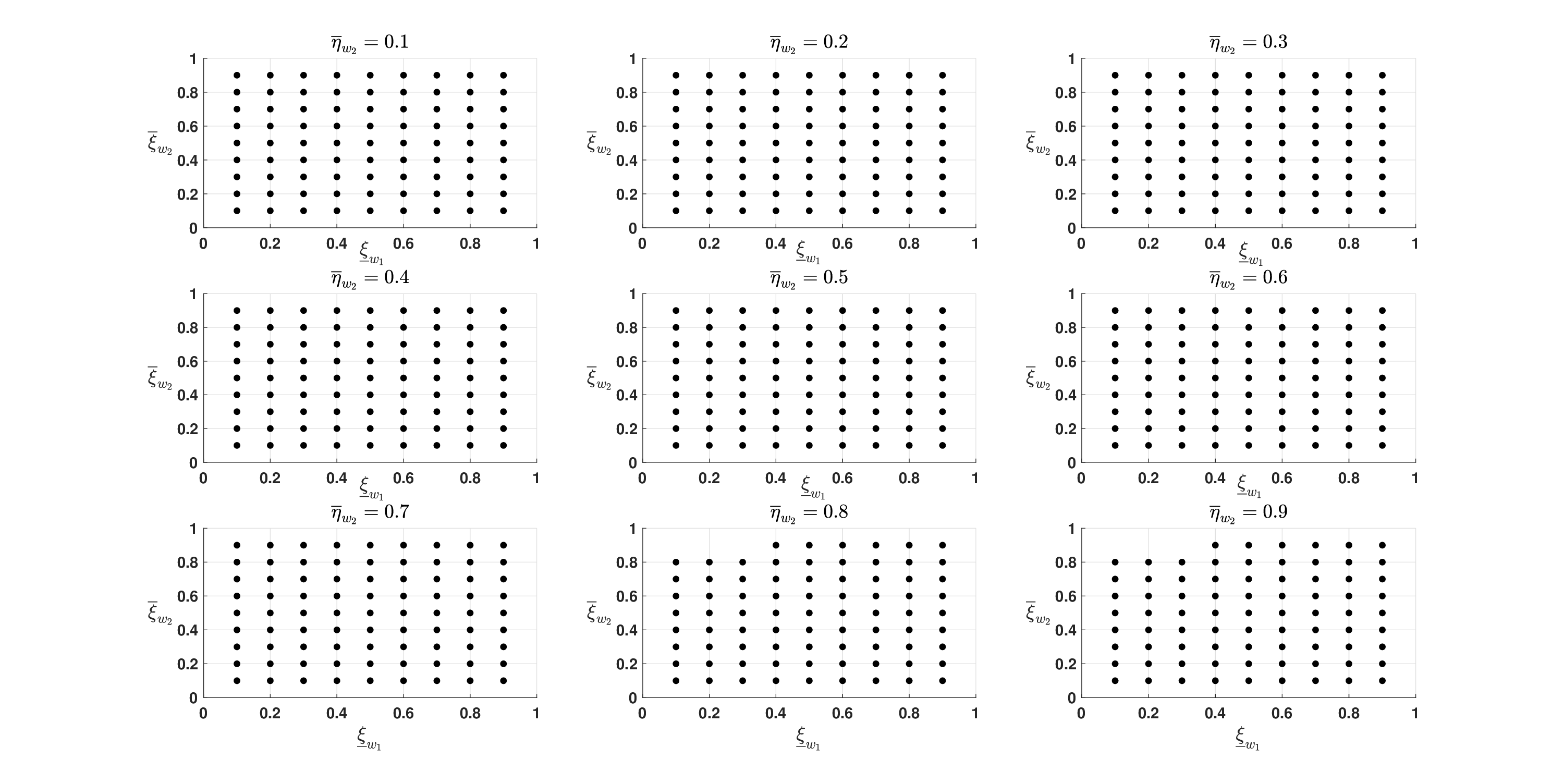}\label{Figure: Sensitivity Waves 1 and 4}}
\caption{Scatter plots of $(\underline{\xi}_{w_1},\overline{\xi}_{w_2},\overline{\eta}_{w_2})$ corresponding to ``Reject the null'' in (\ref{eq - Test Problem Empirical s1})}\label{Figure: Sensitivity all}
\end{sidewaysfigure}

\par Given the form of the bounds on the dominance contrast, we only need to specify $L^{A}_{00}$, $U^{B}_{00}$ and $U^{B}_{10}$ to implement them. We specify them as elements of the family of Generalized Arcsine parametric family. This family has U-shaped PDFs given by
\begin{align*}
g(x;\xi)=\pi^{-1}\,\sin\left(\pi\,\xi\right)\,(x-\underline{x})^{-\xi}\,(\overline{x}-x)^{\xi-1},\quad x\in[\underline{x},\overline{x}],
\end{align*}
where $\xi\in(0,1)$ is a shape parameter, and $[\underline{x},\overline{x}]=[-150000,10^6]$ is the support. This parametric family is ordered with respect to the parameter $\xi$ as follows: $\xi_1\leq \xi_2\implies G(x;\xi_1)\leq  G(x;\xi_2)$, where $G(x,\xi_i)=\int_{-150000}^{x} g(r;\xi_i)\,dr$ and $\xi_i\in(0,1)$ for $i=1,2$. The sensitivity analysis studies the sensitivity of the empirical outcome with respect to configurations of $L^{A}_{00}$, $U^{B}_{00}$ and $U^{B}_{10}$ in this parametric family. Notably, the uniform distribution $U[-150000,10^6]$ is not an member of this parametric family. Hence, the inequalities (\ref{eq - Prop Score 1}) and (\ref{eq - Prop Score 2}) do not define a neighbourhood of the MCAR nonresponse propensities under this parametric shape restriction.

\par Next, we report the results of the sensitivity analysis of the  event ``Reject $H_0$'' in the test problem~(\ref{eq - Test Problem Empirical s1}) for each wave pair, with respect to the choice of $L^{A}_{00}$, $U^{B}_{00}$ and $U^{B}_{10}$ within the Generalized Arcsine parametric family, where $D_{N_A}^{1}(\cdot)$ and $D_{N_B}^{1}(\cdot)$ are described in~(\ref{eq - Emp Sec LB D_NA}) and~(\ref{eq - Emp Sec UB D_NB}), respectively. Let $L^{A}_{00}(\cdot)=G(\cdot,\underline{\xi}_{w_1})$, $U^{B}_{00}(\cdot)=G(\cdot,\overline{\xi}_{w_2})$ and $U^{B}_{10}(\cdot)=G(\cdot,\overline{\eta}_{w_2})$. The testing procedure was implemented for all values of the triple $(\underline{\xi}_{w_1},\overline{\xi}_{w_2},\overline{\eta}_{w_2})$ in the 3-dimensional grid $\Xi=\{0.1,0,2,\ldots,0.9\}^3$.

\par Figure~\ref{Figure: Sensitivity all} reports the results of these tests using scatter plots, and all tests were conducted at the 5\% significance level. The scatter of black dots in each panel of these figures correspond to the subset of $\Xi$ given by
$\Xi_{1}=\left\{(\underline{\xi}_{w_1},\overline{\xi}_{w_2},\overline{\eta}_{w_2})\in\Xi:\, \text{``Reject $H_0$ in~(\ref{eq - Test Problem Empirical s1})''}\right\}$. The subset $\Xi_0=\Xi-\Xi_{1}$ of $\Xi$ is where non-rejection of $H_0$ has occurred. Let $\Xi^{1,2}_{1}$, $\Xi^{1,3}_{1}$, $\Xi^{1,4}_{1}$ be the set $\Xi_{1}$ but corresponding to the tests for pairwise comparison of waves 1 and 2, 1 and 3, and 1 and 4, respectively. Furthermore, define $\Xi^{1,j}_{0}=\Xi-\Xi^{1,j}_{1}$ for $j=2,3,4$. From these figures observe that
$\Xi^{1,2}_{1}\subset\Xi^{1,3}_{1}\subset\Xi^{1,4}_{1}$ and $\Xi^{1,4}_{0}\subset\Xi^{1,3}_{0}\subset\Xi^{1,2}_{0}$, hold. These subset relationships also hold for each subfigure of Figures~\ref{Figure: Sensitivity Waves 1 and 2} -~\ref{Figure: Sensitivity Waves 1 and 4}. This result is expected since the wave nonresponse rates, given by $\hat{\delta}_{10}$ in Table~\ref{Table Nonresponse}, double when moving from comparisons between waves 1 and 2 to that of waves 1 and 4. This doubling has lowered the worst-case lower bound of population $B$ in the comparisons, enabling more configurations of the contrasts within the Generalized Arcsine parametric family to satisfy $H_1$ in the sample, and hence, a chance at rejecting $H_0$.

\par Despite the prevalence of nonresponse, overall, the rejection event is a little sensitive to values of $(\underline{\xi}_{w_1},\overline{\xi}_{w_2},\overline{\eta}_{w_2})\in \Xi$, as $\Xi^{1,j}_{1}$ consists of most elements of $\Xi$ and $\Xi^{1,j}_{0}$ consists only of elements of $\Xi$ where the test statistic equalled zero, for $j=2,3,4$. In fact, the rejection event's sensitivity declines with comparisons of wave 1 with later waves, as $\Xi^{1,4}_{1}$ is almost the entire grid $\Xi$. This empirical finding provides credible evidence at the 5\% significance level that poverty among Australian households has declined between years 2001 to 2002, 2001 and 2003, and 2001 and 2004, according to the headcount ratio over their corresponding set of poverty lines given in Table~\ref{Table Nonresponse}.

\section{Conclusion}\label{Section - Conclusion}
We have proposed a comprehensive design-based framework for executing tests of restricted stochastic dominance with paired data from survey panels that accounts for the identification problem created by nonresponse. The methodology employs an estimating function procedure with nuisance functionals that can encode a broad spectrum of assumptions on nonreponse. Hence, practitioners can use our framework to perform a sensitivity analysis of testing conclusions based on the assumptions they are willing to entertain. We have illustrated the scope of our procedure using data from the HILDA survey in studying the sensitivity of their documented decrease in poverty between 2001 and 2004 in Australia using two types of assumptions on nonresponse. The first assumption embodies the divergence between missing data and observed outcomes through the Kolmogorov-Smirnov distance, while the second assumption constrains the shape of unit and wave nonresponse with incomes. We have found this decrease in poverty is (i) sensitive to departures from ignorability with the Kolmogorov-Smirnov neighborhood assumption, and (ii) relatively robust within a class of nonresponse propensities whose boundary is modeled semiparmetrically using CDFs of the Generalized Arcsine family of distributions.

\section{Acknowledgement}
Rami Tabri expresses gratitude to Brendan K. Beare, Elie T. Tamer, Isaiah Andrews, Aureo de Paula, Mervyn J. Silvapulle, and Christopher Walker for their valuable feedback. Special thanks to the Economics Department at Harvard University and the HILDA team at the Melbourne Institute: Applied Economic and Social Research, University of Melbourne, for their hospitality during his visit. We also thank Sarah C. Dahmann for her assistance with the data preparation for the empirical illustration
\bibliographystyle{chicago}
\bibliography{mybib}
\appendix
\setcounter{page}{1}
\section{Appendix}
This appendix provides the supplementary material to the paper. It is organized as follows.
\begin{itemize}
    \item Appendix \ref{Appendix B} presents the proofs of Theorems~\ref{Rami Theorem 1},~\ref{Rami Theorem 2}, and \ref{power theorem}.
    \item Appendix \ref{Appendix C} presents the technical Lemmas used in the proofs of Theorems \ref{Rami Theorem 2} and \ref{power theorem}.
    \item Appendix \ref{Appendix D} presents the technical result steps for deriving the bounds in the Examples outlined in Section \ref{Subsection Examples}.
    \item Appendix \ref{Appendix E} present the technical details concerning Section \ref{Section - Discussion}.
\end{itemize}
\section{Proofs of Theorems~\ref{Rami Theorem 1} - \ref{power theorem}}\label{Appendix B}

\subsection{Proof of Theorem~\ref{Rami Theorem 1}}\label{Proof of Rami Theorem 1}  
\begin{proof}
The proof follows steps identical to those in the proof of Theorem 1 in \cite{Tabri2021}, but with adjustments for the test statistic. We first show \eqref{eq: asymptotic size}$\implies$~\eqref{eq: asymptotic size sequences}. The proof proceeds by the direct method. Suppose that \eqref{eq: asymptotic size} holds, and let $ \left\{\Pi_{N_A,N_B}\right\}_{N_A,N_B=1}^{+\infty}\in\mathbb{W}_0.$ Then
\begin{align}\label{proof lemma Uniformity 0}
\mathbbm{E}\left[1[LR>c(\alpha)]\mid\Pi_{N_A,N_B}\right]\leq\sup_{\Pi\in\mathcal{M}^0_{N_A,N_B}}E\left[1[LR>c(\alpha)]\mid\Pi\right] \;\forall N_A,N_B,
\end{align}
holds. Taking the limit superiors on both sides of the inequality~(\ref{proof lemma Uniformity 0}) implies the inequality~\eqref{eq: asymptotic size sequences}. As the sequence $\left\{\Pi_{N_A,N_B}\right\}_{N_A,N_B=1}^{+\infty}\in\mathbb{W}_0$ was arbitrary, this inequality holds for all such sequences.

\par Now we shall prove the reverse direction:\eqref{eq: asymptotic size sequences}$\implies$~\eqref{eq: asymptotic size}. The proof proceeds by contraposition. Suppose that~\eqref{eq: asymptotic size} does not hold, i.e.,
\begin{align}\label{proof lemma Uniformity 1}
\limsup_{N_A,N_B\rightarrow+\infty}\sup_{\Pi\in\mathcal{M}^0_{N_A,N_B}}\mathbbm{E}\left[1[LR>c(\alpha)]\mid\Pi\right] & > \alpha.
\end{align}
Then we have to construct a sequence $\left\{\Pi_{N_A,N_B}\right\}_{N_A,N_B=1}^{+\infty}\in\mathbb{W}_0$ such that
\begin{align*}
\limsup_{N_A,N_B\rightarrow+\infty}\mathbbm{E}\left[1[LR^{(A,B)}>c(\alpha)]\mid\Pi_{N_A,N_B}\right] >\alpha
\end{align*}
to prove the result. To that end, the condition~(\ref{proof lemma Uniformity 1}) implies the largest subsequential limit of the sequence $\left\{\sup_{\Pi\in\mathcal{M}^0_{N_A,N_B}}\mathbbm{E}\left[1[LR>c(\alpha)]\mid\Pi\right]\right\}_{N_A,N_B=1}^{+\infty}$ exceeds $\alpha.$ Thus, there is a sequence $\{N_{A,m},N_{B,m}\}_{m=1}^{+\infty}$ such that
the limit of $\left\{\sup_{\Pi\in\mathcal{M}^0_{N_{A,m},N_{B,m}}}\mathbbm{E}\left[1[LR>c(\alpha)]\mid\Pi\right]\right\}_{m=1}^{+\infty}$ exceeds $\alpha,;$ e.g., the limit is equal to $\alpha+\nu$ where $\nu>0.$ Now let $\epsilon>0$ be such that $\nu>\epsilon>0.$ For each $m$ there exists $\Pi^{\prime}_{N_{A,m},N_{B,m}}\in\mathcal{M}^0_{N_{A,m},N_{B,m}}$ such that
\begin{align}\label{proof lemma Uniformity 2}
\mathbbm{E}\left[1[LR>c(\alpha)]\mid\Pi^{\prime}_{N_{A,m},N_{B,m}}\right]>\sup_{\Pi\in\mathcal{M}^0_{N_{A,m},N_{B,m}}}E\left[1[LR>c(\alpha)]\mid\Pi\right]-\epsilon.
\end{align}
Now taking limit superior of both sides of~(\ref{proof lemma Uniformity 2}) with respect to $m,$ yields
\begin{align*}
\limsup_{m\rightarrow+\infty}\mathbbm{E}\left[1[LR>c(\alpha)]\mid\Pi^{\prime}_{N_{A,m},N_{B,m}}\right]& \geq\limsup_{m\rightarrow+\infty}\sup_{\Pi\in\mathcal{M}^0_{N_{A,m},N_{B,m}}}\mathbbm{E}\left[1[LR>c(\alpha)]\mid\Pi\right]-\epsilon \\
& > \alpha+\nu-\epsilon>\alpha.
\end{align*}
Thus, we have constructed a sequence of populations $\left\{\Pi^{\prime}_{N_{A,m},N_{B,m}}\right\}_{m=1}^{+\infty}\in\mathbb{W}_0$ with the desired property. This concludes the proof.
\end{proof}

\subsection{Proof of Theorem \ref{Rami Theorem 2}}
\label{Proof of Rami Theorem 2}

\begin{proof}
We proceed by the direct method. Fix an $\alpha \in (0,1)$.  This proof demonstrates that the testing procedure is uniformly asymptotically valid by showing that the test has level $\alpha$ for all possible subsequences of finite populations based on sequences in $\mathbb{W}_0$. The subsequences where the asymptotic size \eqref{eq: asymptotic size sequences} will be largest is where we have dominance within sample and we are on the boundary of the model of the null hypothesis $H_0^1$. Let $\Upsilon$ denote the event of dominance in sample, that is: $\widehat{\overline{D}}_A^{s}(x) - \widehat{\underline{D}}_B^{s}(x)=\hat{\theta}(x;\hat{\varphi}(x)) < 0 \quad \forall x \in [\underline{t},\overline{t}].$
Let $\{ \Pi_{{N_A}_m , {N_B}_m } \}_{m = 1}^\infty$ be a subsequence of $\{ \Pi_{N_A , N_B} \}_{N_A , N_B = 1}^\infty \in \mathbb{W}_0$. We focus on subsequences that give the event $\Upsilon$ the highest probability (as otherwise $LR = 0$). Fix an $x \in [\underline{t}, \overline{t}]$ arbitrarily, and let $\hat{\tau}_m(x)=\dfrac{\hat{\theta}(x;\hat{\varphi}(x)) }{\sqrt{Var\left(\hat{\theta}(x;\hat{\varphi}(x))\mid \Pi_{{N_A}_m , {N_B}_m } \right)}} \mid \Pi_{{N_A}_m , {N_B}_m }$ $\tau_m(x)=\dfrac{\theta_{{N_{A}}_m,{N_{B}}_m}\left(x;\varphi_{{N_{A}}_m,{N_{B}}_m}(x)\right)}{\sqrt{Var\left(\hat{\theta}(x;\hat{\varphi}(x))\mid \Pi_{{N_A}_m , {N_B}_m } \right)}}\mid \Pi_{{N_A}_m , {N_B}_m }$. Observe that  $\hat{\tau}_m(x)=\hat{\tau}_m(x)-\tau_m(x)+\tau_m(x)$, so that under Conditions 5 and 7 of Assumption \ref{Assumption 1}, as $m \rightarrow \infty$, $\hat{\tau}_m(\cdot)$ must converges weakly to a Gaussian process with mean
\begin{align*}
C(x)=\lim_{m\rightarrow+\infty}\tau_m(x)=\lim_{m\rightarrow+\infty}\dfrac{\theta_{{N_{A}}_m,{N_{B}}_m}\left(x;\varphi_{{N_{A}}_m,{N_{B}}_m}(x)\right)}{\sqrt{Var\left(\hat{\theta}(x;\hat{\varphi}(x))\mid \Pi_{{N_A}_m , {N_B}_m } \right)}}\mid \Pi_{{N_A}_m , {N_B}_m }.
\end{align*}
Therefore, ${\displaystyle\lim_{m \rightarrow \infty} \mathbbm{E}\left[ \mathbbm{1}[\Upsilon] \mid \Pi_{{N_A}_m , {N_B}_m}\right]= \lim_{m \rightarrow \infty} \text{Prob} \left(  \max\limits_{x \in [\underline{t}, \overline{t}]} \hat{\theta}(x;\hat{\varphi}(x))< 0  \mid \Pi_{{N_A}_m , {N_B}_m} \right)}$, which equals $\lim_{m \rightarrow \infty} \text{Prob} \left(  \max\limits_{x \in [\underline{t}, \overline{t}]}\hat{\tau}_m(x) < 0  \mid \Pi_{{N_A}_m , {N_B}_m} \right)$. Now using the weak convergence of $\hat{\tau}_m(\cdot)$, this probability simplifies to
\begin{align*}
 \lim_{m \rightarrow \infty} \mathbbm{E}\left[ \mathbbm{1}[\Upsilon] \mid \Pi_{{N_A}_m , {N_B}_m}\right]=\text{Prob} \left(  \max\limits_{x \in [\underline{t}, \overline{t}]} \left( \mathbb{G} (x) + C(x) \right) < 0  \right). \stepcounter{equation}\tag{\theequation}\label{size dominance probability}
\end{align*}

\par This probability \eqref{size dominance probability} will be highest for subsequences of finite populations within the model of the null hypothesis $H_0^1$ where we have a unique $x^* \in [\underline{t}, \overline{t}]$ such that  $C(x^*) = 0$ and $C(x) = -\infty \ \forall x \in [\underline{t}, \overline{t}]\setminus \{x^*\}$. This corresponds to finite populations in the boundary of the model of the null hypothesis $H_0^1$. Thus, consider two types of subsequences of finite populations. Firstly, subsequences that drift on the boundary: $\Pi_{{N_A}_m , {N_B}_m } \in \partial\mathcal{M}^0_{{N_A}_m , {N_B}_m }$  $\forall m \in \mathbb{N}$. Secondly, subsequences of the finite populations that drift to the boundary of the null model from its interior. Starting with the first case, as we are on the boundary for each $m\in \mathbb{N}$, we have at least one $x \in [\underline{t},\overline{t}]$ such that $\overline{D}_{{N_A}_m}^s (x) = \underline{D}_{{N_B}_m}^s (x)$. For a given $m \in \mathbb{N}$ let $x_m=\min \left\{ x \in [\underline{t} ,\overline{t}]  \; : \; \theta_{{N_A}_m,{N_B}_m}\left(x;\varphi_{{N_A}_m,{N_B}_m}(x)\right)=0 \right\}$. Then using the definition of the LR test statistic \eqref{LR test statistic}: for each  $m\in \mathbb{N}$
\begin{equation}
    \mathbbm{E}\left( \mathbbm{1}[LR > c(\alpha)] \, \mid \, \Pi_{{N_A}_m , {N_B}_m } \right) \leq \mathbbm{E}\left( \mathbbm{1}\left[\frac{2(L_{UR} - L_R (x_m ))}{\widehat{\text{Deff}}(x_m)} > c(\alpha)\right] \, \mid \, \Pi_{{N_A}_m , {N_B}_m } \right)
    \label{LR inequality}
\end{equation}
holds. Now by Lemma \ref{Rami Lemma 1} and taking the limit superior on both sides of \eqref{LR inequality}, it follows that \\ $\limsup_{m \rightarrow \infty } \mathbbm{E}\left( \mathbbm{1}[LR > c(\alpha)] \, \mid \, \Pi_{{N_A}_m , {N_B}_m } \right)\leq \alpha$ since
\begin{align*}
\limsup_{m \rightarrow \infty } \mathbbm{E}\left( \mathbbm{1}\left[\frac{2(L_{UR} - L_R (x_m ))}{\widehat{\text{Deff}}(x_m)} > c(\alpha)\right] \, \mid \, \Pi_{{N_A}_m , {N_B}_m } \right)= \alpha,
\end{align*}
and the limiting distribution of $\{ 2(L_{UR} - L_R (x_m ))/\widehat{\text{Deff}}(x_m) \mid \Pi_{{N_A}_m , {N_B}_m } \}_{m = 1}^\infty$ is $\chi_1^2$.

\par Now focusing on the second case, we have that the subsequence of finite populations drifting to the boundary as $m \rightarrow \infty$ and let $x_e \in [\underline{t}, \overline{t}]$ be such that $C(x_e) = 0$ and $C(x) = -\infty $ for all $x \in [\underline{t} , \overline{t}] \setminus \{ x_e \}$. We utilise the inequality \eqref{LR inequality} and replace $x_m$ with $x_e$. Then taking the limit superior over both sides and applying Lemma \ref{Rami Lemma 2}, we must have that $\limsup_{m \rightarrow \infty } \mathbbm{E}\left( \mathbbm{1}[LR > c(\alpha)] \, \mid \, \Pi_{{N_A}_m , {N_B}_m } \right) \leq\alpha$ holds, since
\begin{align*}
\limsup_{m \rightarrow \infty }  \mathbbm{E}\left( \mathbbm{1}\left[\frac{2(L_{UR} - L_R (x_e ))}{\widehat{\text{Deff}}(x_e)} > c(\alpha)\right] \, \mid \, \Pi_{{N_A}_m , {N_B}_m } \right) = \alpha.
\end{align*}
Note that Lemma \ref{Rami Lemma 2} establishes that limiting distribution of $\{ 2(L_{UR} - L_R (x_e))/\widehat{\text{Deff}}(x_e) \, \mid \, \Pi_{{N_A}_m , {N_B}_m } \}_{m = 1}^\infty$ as $\chi_1^2$. This concludes the proof.
\end{proof}


\subsection{Proof of Theorem \ref{power theorem}}
\label{Proof of power theorem}
\begin{proof}
\par There are two possible cases for the sequences of finite populations which we examine under the alternative. Firstly, we consider sequences of finite populations that are uniformly bounded away from the boundary of the null hypothesis. Secondly, we consider sequences of finite populations that converges to $\partial \mathcal{M}_{N_A,N_B}^1$, the boundary of the model of the null hypothesis as $N_A,N_B \rightarrow \infty$.

\par Proceeding under the first case via the direct method, fix an $\varepsilon > 0$. Consider a sequence of finite populations $\{ \Pi_{N_A, N_B} \}_{N_A, N_B=1}^\infty \in \mathbb{W}_1 (\varepsilon)$. For any $\{ \Pi_{N_A, N_B} \}_{N_A, N_B=1}^\infty \in \mathbb{W}_1 (\varepsilon)$ we have that $C(x) = -\infty$ for all $x \in [\underline{t},\overline{t}]$. Let $\Upsilon$ denote the event of dominance within sample: i.e., the event $\max_{x\in[\underline{t},\overline{t}]}\hat{\theta}(x;\hat{\varphi}(x))<0$. Using that $C(x) = - \infty$ for all $x \in [\underline{t},\overline{t}]$ and arguments similar to those in the proof of Theorem~\ref{Rami Theorem 2}, it follows that  $\lim_{N_A, N_B \rightarrow \infty}   \mathbbm{E}[\mathbbm{1}[\Upsilon] \, \mid \, \Pi_{N_A, N_B}] = 1$. Consequently, the event $\Upsilon$ occurs with probability tending to unity as $N_A,N_B\rightarrow+\infty$, and hence, $LR=\min_{ x \in [\underline{t} ,\overline{t}]} \frac{2(L_{UR} - L_R (x))}{\widehat{\text{Deff}}(x)}$ with probability tending to unity as $N_A,N_B\rightarrow+\infty$. Now for each $x \in [\underline{t},\overline{t}]$, applying Lemma \ref{power theorem Lemma 1} to the sequence $\{ 2(L_{UR} - L_R (x))/\widehat{\text{Deff}}(x) \mid  \Pi_{N_A, N_B} \}_{N_A ,N_B = 1}^\infty$, we have that $2(L_{UR} - L_R (x))/\widehat{\text{Deff}}(x) \pconv \infty \ \text{as} \ N_A ,N_B \rightarrow \infty$. Accordingly, the form of $LR$ implies
   $LR\mid \Pi_{N_A,N_B} \stackrel{P}{\longrightarrow}+\infty$. Since the critical value $c(\alpha)$ is fixed and finite, it follows that $\lim_{N_A,N_B\rightarrow+\infty} \mathbbm{E}\left[\mathbbm{1}LR>c(\alpha)]\mid\Pi_{N_A,N_B}\right]=1$. This concludes the proof for this case.

\par Now focusing on the second case, we again proceed via the direct method. Let $\varepsilon = 0$; therefore, the sequence of finite populations $\{ \Pi_{N_A, N_B} \}_{N_A, N_B=1}^\infty \in \mathbb{W}_1 (0)$. Consequently, there will be at least one $x \in [\underline{t},\overline{t}]$ such that $C(x) = 0$.  Let $X_0 = \{ x \in [\underline{t}, \overline{t}] : C(x) = 0\}$. By contrast with the previous case, the limiting probability of the event $\Upsilon$ is $\text{Prob}\left( \max\limits_{x \in X_0} \left( \mathbb{G}(x) + C(x) \right) < 0  \right)=\text{Prob}\left( \max\limits_{x \in X_0}  \mathbb{G}(x) < 0  \right)$,
by arguments similar to those in the proof of Theorem~\ref{Rami Theorem 2}. As this limiting probability is not necessarily equal to unity, we investigate the asymptotic power of the test on the event $\Upsilon$, since for each $N_A,N_B\in\mathbb{N}$ the equalities
$\mathbbm{E}\left( \mathbbm{1}[LR > c(\alpha)] \mid \Pi_{N_A, N_B} \right)= \mathbbm{E}\left( \mathbbm{1}[LR > c(\alpha),\Upsilon] \mid \Pi_{N_A, N_B} \right)$, hold, because $\Upsilon^c\implies LR=0$. Hence, $\lim_{N_A, N_B \rightarrow \infty} \mathbbm{E}\left( \mathbbm{1}[LR > c(\alpha)] \mid \Pi_{N_A, N_B} \right)  = \lim_{N_A, N_B \rightarrow \infty} \mathbbm{E}\left( \mathbbm{1}[LR > c(\alpha),\Upsilon] \mid \Pi_{N_A, N_B} \right)$, holds. Now Lemma \ref{power theorem Lemma 1} establishes that $2(L_{UR} - L_R (x))/\widehat{\text{Deff}}(x) \pconv \infty \ \text{as} \ N_A ,N_B \rightarrow \infty$, holds,  for each $x \in [\underline{t}, \overline{t}] \setminus X_0$. Furthermore, Lemma \ref{power theorem Lemma 2} establishes that $2(L_{UR} - L_R (x))/\widehat{\text{Deff}}(x) \dconv \chi_1^2\ \text{as} \ N_A ,N_B \rightarrow \infty$, holds, for each $x\in X_0$. Therefore, on the event $\Upsilon$ and asympotically, the rejection event is characterized as $ LR > c(\alpha) \iff \min_{x \in X_0} \frac{2(L_{UR} - L_R (x))}{\widehat{\text{Deff}}(x)} > c(\alpha)$. Hence, $\lim_{N_A, N_B \rightarrow \infty} \mathbbm{E}\left( \mathbbm{1}[LR > c(\alpha),\Upsilon] \mid \Pi_{N_A, N_B} \right)$ equals the limiting probability \\
$\lim_{N_A, N_B \rightarrow \infty} \mathbbm{E}\left( \mathbbm{1}\left[\min_{x \in X_0} \frac{2(L_{UR} - L_R (x))}{\widehat{\text{Deff}}(x)}> c(\alpha),\Upsilon\right] \mid \Pi_{N_A, N_B} \right)$.
By Lemma \ref{power theorem Lemma 2}, this resulting limit equals $\text{Prob}\left( \mathbb{G}^2 (x) > c (\alpha) \, \forall x \in X_0,\max\limits_{x \in X_0}  \mathbb{G}(x) < 0  \right)$
as we show the statistics $\left\{\frac{2(L_{UR} - L_R (x))}{\widehat{\text{Deff}}(x)},x\in X_0\right\}$ converge in distribution to $\left\{\mathbb{G}^2(x),x\in X_0\right\}$ pointwise in $x\in X_0$. This concludes the proof.
\end{proof}
\section{Technical Lemmas}
\label{Appendix C}

\subsection{Solution to PELF \eqref{eq: PELF}}
\label{Solution to PELF}

Here we develop the solution to the PELF, described by \eqref{eq: PELF}. For a given population $\Pi_{N_A, N_B}$ suppose we have access to a panel survey sample $V \in \mathcal{U}_{A}$. Firstly, we determine the solution to the unconstrained problem $L_{UR}$. Then for each $x\in [\underline{t}, \overline{t}]$, we determine the solution to the constrained optimisation problem denoted by $L_R (x)$. The unconstrained optimisation problem is described by:

\begin{equation}
    \max_{\Vec{p} \in (0,1]^n} \sum_{i\in U} W^\prime_i \log(p_i) \quad \text{subject to} \quad \sum_{i\in U} W^\prime_i p_i = 1,
\end{equation}
where $U=\{i\in V: Z^A_i=1\}$, and $\{W^\prime_i:i\in U\}$ satisfies~(\ref{eq - Design weights-responders}) so that $\sum_{i\in U} W^\prime_i=n$, holds.
The corresponding Lagrangian is given by:
\begin{equation*}
    \mathcal{L} = \sum_{i\in U} W^\prime_i \log(p_i) - \lambda \left( \sum_{i\in U} W^\prime_i p_i - 1 \right)
\end{equation*}

The First Order Conditions are, taking them with respect to i for generality:
\begin{align}
    \dfrac{\partial \mathcal{L}}{\partial p_i} &=  \dfrac{W^\prime_i}{p_i} - \lambda W^\prime_i = 0, \quad \forall i\in U    \label{eq: 3a} \\
    \dfrac{\partial \mathcal{L}}{\partial \lambda} &= \sum_{i\in U} W_i p_i - 1 = 0, \ \text{and} \label{eq: 3b}
\end{align}
Multiplying \eqref{eq: 3a} by $p_i$ and then summing over $i=1,2,\hdots,n$
\begin{align*}
    &W^\prime_i - \lambda W^\prime_i p_i = 0 \\
    &\sum_{i\in U} W^\prime_i = \lambda \sum_{i\in U} W^\prime_i p_i \\
    &\text{using \eqref{eq: 3b} and \eqref{eq - Design weights-responders}} \\
    \implies & n = \lambda \stepcounter{equation}\tag{\theequation}\label{eq: 3d}
\end{align*}
Then substituting \eqref{eq: 3d} into \eqref{eq: 3a} and solving for $p_i$, we obtain:
\begin{align*}
    &W^\prime_i - n W^\prime_i p_i = 0 \\
    & 1 - np_i = 0 \\
    \implies &p_i = \dfrac{1}{n} \quad \forall i\in U
    \stepcounter{equation}\tag{\theequation}\label{eq: 3e}
\end{align*}

Therefore, using \eqref{eq: 3e} we have that the unrestricted log-likelihood is given by:
\begin{equation}
    L_{UR} = \sum_{i\in U} W^\prime_i \log \left(\frac{1}{n} \right)
    \label{eq: 3f}
\end{equation}

The constrained maximisation problem of the PELF corresponds to the following optimisation problem, where the last restriction imposes the moment equality for $x\in[\underline{t},\overline{t}]$ corresponding to the boundary of the null hypothesis $H_0^1$:
\begin{equation}
    \begin{aligned}
    \max_{\Vec{p} \in (0,1]^n} & \sum_{i\in U} W^\prime_i \log(p_i) \\
    \text{subject to} \ \quad & \sum_{i\in U} W^\prime_i p_i = 1 \\
    & \sum_{i\in U} W^\prime_i p_i H_i (x;\hat{\varphi}(x)) = 0
    \end{aligned}
    \label{Constrained Problem}
\end{equation}

The corresponding Lagrangian is given by:
\begin{equation*}
    \mathcal{L} = \sum_{i\in U} W^\prime_i \log(p_i) - \lambda \left( \sum_{i\in U} W^\prime_i p_i - 1 \right) - \kappa \left( \sum_{i\in U} W^\prime_i p_i H_i (x;\hat{\varphi}(x)) \right)
\end{equation*}

The First Order Conditions are
\begin{align}
    \dfrac{\partial \mathcal{L}}{\partial p_i} &=  \dfrac{W^\prime_i}{p_i} - \lambda W^\prime_i - \kappa W^\prime_i H_i (x;\hat{\varphi}(x)) = 0, \quad \forall i\in U
    \label{eq: 4a} \\
    \dfrac{\partial \mathcal{L}}{\partial \lambda} &= \sum_{i\in U} W^\prime_i p_i - 1 = 0, \label{eq: 4b} \\
    \dfrac{\partial \mathcal{L}}{\partial \kappa} &= \sum_{i\in U} W^\prime_i p_i H_i (x;\hat{\varphi}(x)) = 0.
    \label{eq: 4d}
\end{align}

Multiplying \eqref{eq: 4a} by $p_i$ then summing across $i\in U$
\begin{align*}
    & W^\prime_i - \lambda W^\prime_ip_i - \kappa W^\prime_i p_i H_i (x;\hat{\varphi}(x)) = 0 \\
    & \sum_{i\in U} W^\prime_i - \lambda \sum_{i\in U} W^\prime_ip_i - \kappa \sum_{i\in U} W^\prime_i p_i H_i (x;\hat{\varphi}(x)) = 0 \\
    &\text{using \eqref{eq: 4b}, \eqref{eq: 4d} and that $\sum_{i\in U} W^\prime_i=n$} \\
    & n - \lambda  = 0 \\
    \implies &\lambda = n.\stepcounter{equation}\tag{\theequation}\label{eq: 4e}
\end{align*}

We then substitute \eqref{eq: 4e} into \eqref{eq: 4a} and solve for $p_i$
\begin{align*}
    & \dfrac{W^\prime_i}{p_i}  - n W^\prime_i - \kappa W^\prime_i H_i (x;\hat{\varphi}(x)) = 0 \\
    & \dfrac{1}{p_i} - n - \kappa H_i (x;\hat{\varphi}(x)) = 0 \\
    & \dfrac{1}{p_i} = n + \kappa H_i (x;\hat{\varphi}(x)) \\
    \implies & p_i = \dfrac{1}{n + \kappa H_i (x;\hat{\varphi}(x))} \quad \forall i\in U. \stepcounter{equation}\tag{\theequation}\label{eq: 4f}
\end{align*}

Interpreting this result, $\kappa$ is the cost of imposing the null hypothesis $H_0^1$.  The more readily the data fits the constraint, the easier it is to impose the constraint, and so $\kappa$ will converge to zero, resulting in the unconstrained value for the $p_i$.

\par The following lemma is helpful for developing the technical result sin the subsequent sections.
\begin{lemma}\label{H_i bounded}
Suppose that Assumption~\ref{Assumption - Compactness of X_K} holds. Then $H_i (x;\hat{\varphi}(x))$, defined in~\eqref{eq: moment function}, is bounded on $[\underline{t},\overline{t}]$ for each $i\in U$.
\end{lemma}
\begin{proof}
The proof proceeds by the direct method. Observe that
\begin{align*}
\left |H_i (x;\hat{\varphi}(x))\right|\leq \|\hat{\varphi}\|_{\Psi}\left(\frac{(\overline{Y}^A-\underline{Y}^A)^{s-1}}{(s-1)!}+\frac{(\overline{Y}^B-\underline{Y}^B)^{s-1}}{(s-1)!}+1\right)<+\infty\quad\forall x\in[\underline{t},\overline{t}],
\end{align*}
where the last equality holds under Assumption~\ref{Assumption - Compactness of X_K}, as it implies that $\|\hat{\varphi}\|_{\Psi}<+\infty$ and $\overline{Y}^K,\underline{Y}^K\in\mathbb{R}$ for $K=A,B$.
\end{proof}

\subsection{Lemmas for Proof of Theorem \ref{Rami Theorem 2}}
\label{Lemmas 1 and 2}

The following Lemmas derive the asymptotic distribution of $\{ 2(L_{UR} - L_R (x))/\widehat{\text{Deff}}(x) \mid \Pi_{{N_A}_m , {N_B}_m } \}_{m = 1}^\infty$ for subsequences of sequences of $\{ \Pi_{{N_A} , {N_B}} \}_{N_A , N_B}^\infty$ that drift on/to the boundary of the model of the null hypothesis $H_0^1$.

\subsubsection{Lemma \ref{Rami Lemma 1}}

\begin{lemma}
\label{Rami Lemma 1}
Let $\left\{ \mathcal{M}_{N_A , N_B}^0, \; N_A, N_B = 1,2,\hdots \right\}$, $\mathbb{W}$ and $\alpha$ be the same as in Theorem \ref{Rami Theorem 1}.  Let $\alpha \in (0,1)$.  For a sequence of finite populations $\{ \Pi_{N_A , N_B} \}_{{N_A , N_B} = 1}^\infty \in \mathbb{W}_0$, suppose there is a subsequence $\{ \Pi_{{N_A}_m , {N_B}_m } \}_{m = 1}^\infty$ such that $\Pi_{{N_A}_m , {N_B}_m } \in \partial \mathcal{M}_{{N_A}_m , {N_B}_m }^0 \; \forall m \in \mathbb{N}$.  For each $m \in \mathbb{N}$ let
\begin{align*}
x_m &  = \min \left\{ x \in [\underline{t} ,\overline{t}]  \; : \; \overline{D}_{{N_A}_m}^s (x) = \underline{D}_{{N_B}_m}^s (x) \right\}\\
  & = \min \left\{ x \in [\underline{t} ,\overline{t}]  \; : \; \theta_{{N_A}_m,{N_B}_m}\left(x;\varphi_{{N_A}_m,{N_B}_m}(x)\right)=0 \right\}.
\end{align*}
Then, $2(L_{UR} - L_R (x_m))/\widehat{\text{Deff}}(x_m) \; \mid \; \Pi_{{N_A}_m , {N_B}_m } \dconv \chi_1^2 \ \text{as} \ m \rightarrow \infty$.
\end{lemma}

\begin{proof}
The proof proceeds using the direct method. For simplicity, we drop $``\mid \; \Pi_{{N_A}_m , {N_B}_m }"$. By the compactness of the set $\left\{ x \in [\underline{t} ,\overline{t}]  \; : \; \theta_{{N_A}_m,{N_B}_m}\left(x;\varphi_{{N_A}_m,{N_B}_m}(x)\right)=0 \right\}$, $x_m$ exists. Note that $x_m$ is non-random and depends only on the sequence of finite populations. Using \eqref{eq: 4f}, notice that $p_i$ can be rewritten as:
\begin{align*}
    p_i &= \dfrac{1}{n} \left[ \dfrac{1}{1 + \dfrac{\kappa H_i (x_m;\hat{\varphi}(x_m))}{n}}\right] \\
    & = \dfrac{1}{n} \left[ \sum_{k=0}^\infty \left(-\frac{\kappa H_i (x_m;\hat{\varphi}(x_m))}{n}\right)^k \right] \\
    & = \dfrac{1}{n} \left[ 1 - \dfrac{\kappa H_i (x_m;\hat{\varphi}(x_m))}{n} + \left(\dfrac{\kappa H_i (x_m;\hat{\varphi}(x_m))}{n}\right)^2 - \hdots \right] \stepcounter{equation}\tag{\theequation}\label{eq: 4g}
\end{align*}

Then using \eqref{eq: 4g}, we substitute this into \eqref{eq: 4d}
\begin{align*}
    & \sum_{i\in U} W^\prime_i H_i (x_m;\hat{\varphi}(x_m)) \left(\dfrac{1}{n} \left[ 1 - \dfrac{\kappa H_i (x_m;\hat{\varphi}(x_m))}{n} + \left(\dfrac{\kappa H_i (x_m;\hat{\varphi}(x_m))}{n}\right)^2 - \hdots \right] \right) = 0 \\
    &\bar{H}(x_m;\hat{\varphi}(x_m)) - \kappa \left( \sum_{i\in U}\dfrac{W^\prime_i}{n^2} H^2_i (x_m;\hat{\varphi}(x_m)) \right) + \kappa^2 \left( \sum_{i\in U} \dfrac{W^\prime_i}{n^3} H^3_i (x_m;\hat{\varphi}(x_m)) \right) - \hdots = 0,
    \stepcounter{equation}\tag{\theequation}\label{eq: 4h}
\end{align*}
\noindent where $\bar{H}(x_m;\hat{\varphi}(x_m)) =\hat{\theta}(x_m;\hat{\varphi}(x_m))= \sum_{i\in U} \frac{W^\prime_i}{n}H_i (x_m;\hat{\varphi}(x_m)) $. Equation~\eqref{eq: 4h} is equivalent to a convergent alternating series in $\kappa$, so $\kappa^j \sum_{i\in U}\dfrac{W^\prime_i}{n^{j+1}} H^{j+1}_i (x_m;\hat{\varphi}(x_m)) \rightarrow 0 $ as $j \rightarrow \infty$. Therefore as $m \rightarrow \infty$, ignoring terms beyond the first order, $\kappa$ is asymptotically equivalent to:
\begin{equation}
    \kappa \aequal \left[\sum_{i\in U} \dfrac{W^\prime_i}{n^2}  H^2_i (x_m;\hat{\varphi}(x_m))\right]^{-1}\bar{H}(x_m;\hat{\varphi}(x_m)).
    \label{eq: 4i}
\end{equation}
Let $V(x_m) = \sum_{i\in U} \dfrac{W^\prime_i}{n^2}  H^2_i (x_m;\hat{\varphi}(x_m))$. Now using the result \eqref{eq: 4f} for $p_i$ in $2(L_{UR} - L_R (x_m))$:
\begin{align*}
    2(L_{UR} - L_R (x_m)) &= 2 \left(\sum_{i\in U} W^\prime_i \log \left(\frac{1}{n} \right) - \sum_{i\in U} W^\prime_i \log \left(p_i \right) \right) \\
    &= 2\sum_{i\in U} W^\prime_i \log \left(\frac{1}{n} \frac{1}{p_i} \right) \\
    &= 2 \sum_{i\in U} W^\prime_i \log \left( 1 + \frac{\kappa H_i (x_m;\hat{\varphi}(x_m))}{n}\right).
    \stepcounter{equation}\tag{\theequation}\label{Lemma 1 log expression}
\end{align*}

Since Lemma~\ref{H_i bounded} establishes that $H_i (x_m;\hat{\varphi}(x_m))$ is uniformly bounded, and $\kappa$ is tending to zero in probability and in design, we then use the standard expansion: $\log(1+u) \aequal u - \dfrac{u^2}{2} \quad \text{for} \ |u| < 1$, with $u=\frac{\kappa H_i (x_m;\hat{\varphi}(x_m))}{n}$.

\begin{align*}
    2 \sum_{i\in U} W^\prime_i \log \left( 1 + \frac{\kappa H_i (x_m;\hat{\varphi}(x_m))}{n} \right) &\aequal 2 \sum_{i\in U} W^\prime_i \left[\frac{\kappa H_i (x_m;\hat{\varphi}(x_m))}{n} - \frac{1}{2}\left(\frac{\kappa H_i (x_m;\hat{\varphi}(x_m))}{n} \right)^2 \right] \\
    &= 2 \kappa \bar{H}(x_m;\hat{\varphi}(x_m)) - \kappa^2 V (x_m) \\
    &\text{then using \eqref{eq: 4i}} \\
    &\aequal \dfrac{2 \bar{H}^2(x_m;\hat{\varphi}(x_m))}{V (x_m)} - \dfrac{\bar{H}^2(x_m;\hat{\varphi}(x_m))}{V (x_m)} \\
    &= \dfrac{\bar{H}^2(x_m;\hat{\varphi}(x_m))}{V (x_m)};
    \stepcounter{equation}\tag{\theequation}\label{eq: 4j}
\end{align*}
therefore,
\begin{equation*}
2(L_{UR} - L_R (x_m)) \aequal \left( \dfrac{\hat{\theta}(x_m;\hat{\varphi}(x_m))}{\sqrt{Var\left(\hat{\theta}(x_m;\hat{\varphi}(x_m))\right)}}\right)^2 \frac{Var\left(\hat{\theta}(x_m;\hat{\varphi}(x_m))\right)}{\sum_{i\in U} \dfrac{W^\prime_i}{n^2} H^2_i (x_m;\hat{\varphi}(x_m))}.
\end{equation*}
Under the restriction of the boundary of the null hypothesis $H_0^1$:
$$\sum_{i\in U} \dfrac{W^\prime_i}{n^2} \left(H_i (x_m;\hat{\varphi}(x_m)) -\bar{H}(x_m;\hat{\varphi}(x_m))\right)^2 = \sum_{i\in U} \dfrac{W^\prime_i}{n^2} H^2_i (x_m;\hat{\varphi}(x_m)) + o_P (1),$$
since $\bar{H}(x_m;\hat{\varphi}(x_m) \pconv 0$ as ${N_A}_m,{N_B}_m \rightarrow\infty$. Hence,
\begin{align*}
2(L_{UR} - L_R (x_m))/\widehat{\text{Deff}}(x_m) &  \aequal \left( \dfrac{\hat{\theta}(x_m;\hat{\varphi}(x_m))}{\sqrt{Var\left(\hat{\theta}(x_m;\hat{\varphi}(x_m))\right)}}\right)^2 \frac{\text{Deff}(x_m)}{\widehat{\text{Deff}}(x_m)},\\
& = \left( \dfrac{\hat{\theta}(x_m;\hat{\varphi}(x_m))}{\sqrt{Var\left(\hat{\theta}(x_m;\hat{\varphi}(x_m))\right)}}\right)^2 \frac{Var\left(\hat{\theta}(x_m;\hat{\varphi}(x_m))\right)}{\widehat{Var}\left(\hat{\theta}(x_m;\hat{\varphi}(x_m))\right)},
\end{align*}
and now using Conditions 4,5,6 and 7 of Assumption \ref{Assumption 1}, and that $$\overline{D}_{{N_A}_m}^s (x_m) = \underline{D}_{{N_B}_m}^s (x_m)\iff\theta_{{N_A}_m,{N_B}_m}\left(x_m;\varphi_{{N_A}_m,{N_B}_m}(x_m)\right)=0,$$ we establish that along this subsequence:
\begin{equation*}
  2(L_{UR} - L_R (x_m))/\widehat{\text{Deff}}(x_m) \dconv \chi_1^2 \ \text{as} \ m \rightarrow \infty.
\end{equation*}
\end{proof}

\subsubsection{Lemma \ref{Rami Lemma 2}}

\begin{lemma}
\label{Rami Lemma 2}
Let $\left\{ \mathcal{M}_{N_A , N_B}^0, \; N_A, N_B = 1,2,\hdots \right\}$, $\mathbb{W}_0$ and $\alpha$ be the same as in Theorem \ref{Rami Theorem 1}. Let $\alpha \in (0,1)$. For a sequence of finite populations $\{ \Pi_{N_A , N_B} \}_{{N_A , N_B} = 1}^\infty \in \mathbb{W}$, suppose there is a subsequence $\{ \Pi_{{N_A}_m , {N_B}_m } \}_{m = 1}^\infty$ such that: $\Pi_{{N_A}_m , {N_B}_m } \in int \left( \mathcal{M}_{{N_A}_m , {N_B}_m }^0 \right) \; \forall m \in \mathbb{N}$ and $\exists! \ x_e \in [\underline{t}, \overline{t}]$ such that $C(x_e) = 0$ and $C(x) = - \infty \ \forall x \in [\underline{t}, \overline{t}] \setminus \{ x_e \}$,
where
\begin{align*}
C(x) & =\lim_{m\rightarrow+\infty}\dfrac{  \overline{D}_{{N_A}_m}^s (x) - \underline{D}_{{N_B}_m}^s (x) }{\sqrt{Var\left(\widehat{\overline{D}}_A^s (x) - \widehat{\underline{D}}_B^s (x) \mid \Pi_{{N_A}_m , {N_B}_m }\right)}} \mid \Pi_{{N_A}_m , {N_B}_m } \\
 & = \lim_{m\rightarrow+\infty}\dfrac{ \theta_{{N_A}_m,{N_B}_m}\left(x;\varphi_{{N_A}_m,{N_B}_m}(x)\right) }{\sqrt{Var\left(\hat{\theta}(x;\hat{\varphi}(x)) \mid \Pi_{{N_A}_m , {N_B}_m }\right)}} \mid \Pi_{{N_A}_m , {N_B}_m }
\end{align*}
Then $2(L_{UR} - L_R (x_e))/\widehat{\text{Deff}}(x_e) \; \mid \; \Pi_{{N_A}_m , {N_B}_m } \dconv \chi_1^2 \ \text{as} \ m \rightarrow \infty$.
\end{lemma}

\begin{proof}
Following identical steps as in Lemma \ref{Rami Lemma 1}, we have:
\begin{align*}
    &2(L_{UR} - L_R (x_e)) \aequal \left( \dfrac{\hat{\theta}(x_e;\hat{\varphi}(x_e))}{\sqrt{Var\left(\hat{\theta}(x_e;\hat{\varphi}(x_e))\right)}}\right)^2 \frac{Var\left(\hat{\theta}(x_e;\hat{\varphi}(x_e))\right)}{\sum_{i\in U} \dfrac{W^\prime_i}{n^2} H^2_i (x_e;\hat{\varphi}(x_e))}\\
    =& \left( \dfrac{\hat{\theta}(x_e;\hat{\varphi}(x_e))-\theta_{{N_A}_m,{N_B}_m}\left(x_e;\varphi_{{N_A}_m,{N_B}_m}(x_e)\right)}{\sqrt{Var\left(\hat{\theta}(x_e;\hat{\varphi}(x_e))\right)}}+\frac{\theta_{{N_A}_m,{N_B}_m}\left(x_e;\varphi_{{N_A}_m,{N_B}_m}(x_e)\right)}{\sqrt{Var\left(\hat{\theta}(x_e;\hat{\varphi}(x_e))\right)}}\right)^2 \\ & \quad \times\frac{Var\left(\hat{\theta}(x_e;\hat{\varphi}(x_e))\right)}{\sum_{i\in U} \dfrac{W^\prime_i}{n^2} H^2_i (x_e;\hat{\varphi}(x_e))}.
\end{align*}
Under the restriction of the boundary of the null hypothesis $H_0^1$:
$$\sum_{i\in U} \dfrac{W^\prime_i}{n^2} \left(H_i (x_e;\hat{\varphi}(x_e)) -\bar{H} (x_e;\hat{\varphi}(x_e))\right)^2 = \sum_{i\in U} \dfrac{W^\prime_i}{n^2} H^2_i (x_e;\hat{\varphi}(x_e) + o_P (1),$$
since $\bar{H} (x_e;\hat{\varphi}(x_e))=\hat{\theta}(x_e;\hat{\varphi}(x_e)) \pconv 0$ as ${N_A}_m,{N_B}_m \rightarrow\infty$. By definition of $x_e$, the second term inside the brackets does not affect the limiting behaviour of the test statistic. Therefore, $2(L_{UR} - L_R (x_e))/\widehat{\text{Deff}}(x_e)$ is asymptotically equivalent to
\begin{align*}
\left( \dfrac{\hat{\theta}(x_e;\hat{\varphi}(x_e))-\theta_{{N_A}_m,{N_B}_m}\left(x_e;\varphi_{{N_A}_m,{N_B}_m}(x_e)\right)}{\sqrt{Var\left(\hat{\theta}(x_e;\hat{\varphi}(x_e))\right)}}+\frac{\theta_{{N_A}_m,{N_B}_m}\left(x_e;\varphi_{{N_A}_m,{N_B}_m}(x_e)\right)}{\sqrt{Var\left(\hat{\theta}(x_e;\hat{\varphi}(x_e))\right)}}\right)^2
  \,\frac{Var\left(\hat{\theta}(x_e;\hat{\varphi}(x_e))\right)}{\widehat{Var}\left(\hat{\theta}(x_e;\hat{\varphi}(x_e))\right)},
    \end{align*}
and now applying Conditions 3 - 7 of Assumption \ref{Assumption 1} to the components of this expression yields
\begin{align*}
  \left( \dfrac{\hat{\theta}(x_e;\hat{\varphi}(x_e))}{\sqrt{Var\left(\hat{\theta}(x_e;\hat{\varphi}(x_e))\right)}}\right)^2 \; \mid \; \Pi_{{N_A}_m , {N_B}_m } \dconv \chi_1^2, & \\
     \frac{Var\left(\hat{\theta}(x_e;\hat{\varphi}(x_e))\right)}{\widehat{Var}\left(\hat{\theta}(x_e;\hat{\varphi}(x_e))\right)} \; \mid \; \Pi_{{N_A}_m , {N_B}_m } \pconv 1. &
\end{align*}
Therefore, along $\{ \Pi_{{N_A}_m , {N_B}_m } \}_{m = 1}^\infty$ it follows that
\begin{equation*}
    2(L_{UR} - L_R (x_e))/\widehat{\text{Deff}}(x_e) \; \mid \; \Pi_{{N_A}_m , {N_B}_m } \dconv \chi_1^2 \ \text{as} \ m \rightarrow \infty.
\end{equation*}

\end{proof}

\subsection{Lemmas for Proof of Theorem \ref{power theorem}}
\label{Lemmas 3 and 4}

The following two Lemmas describe the asymptotic behaviour of the test statistic $\{ 2(L_{UR} - L_R (x))/\widehat{\text{Deff}}(x) \, \mid \, \Pi_{{N_A} , {N_B} } \}_{N_A ,N_B = 1}^\infty$ under the alternative hypothesis. The first Lemma focuses on sequences of finite populations that are uniformly bounded away from the boundary, corresponding to $C(x) = -\infty$ for all $x \in [\underline{t}, \overline{t}]$, where $C(\cdot)$ is given by~(\ref{C function}). The second Lemma focuses on sequences of finite populations that converge towards the boundary of the model of the null hypothesis $H_0^1$ where there exists one or more $x \in [\underline{t}, \overline{t}]$ such that $C(x)=0$.

\par The following result is helpful in establishing these intermediate technical results.

\begin{lemma}\label{Lemma Deff limit}
Given finite population $\Pi_{N_A,N_B}$, under Assumption~\ref{Assumption 2}, for each $x\in[\underline{t},\overline{t}]$ the design-effect satisfies
\begin{align}
\text{Deff}(x)\pconv d^{-1}\frac{\lim\limits_{N_A, N_B \rightarrow \infty } \mathbbm{E}( n \mid \Pi_{N_A, N_B}) Var\left(\hat{\theta}(x;\hat{\varphi}(x)) \mid \Pi_{N_A, N_B} \right)}{\lim\limits_{N_A,N_B\rightarrow+\infty}S_{N_A,N_B}},
\end{align}
and this limit is finite and positive.
\end{lemma}

\begin{proof}
The proof proceed by the direct method. Using the representation of the design-effect, $\text{Deff}(x)$ in~\eqref{eq - deff}, note that
\begin{align}
\text{Deff}(x) & =\frac{nVar\left(\hat{\theta}(x;\hat{\varphi}(x)) \, \mid \, \Pi_{N_A ,N_B} \right)}{\sum_{i\in U}(W^\prime_i/n)H^2_i (x;\hat{\varphi}(x))}\\
 & =\frac{n}{E[n\mid \Pi_{N_A ,N_B}]}\frac{E[n\mid \Pi_{N_A ,N_B}] Var\left(\hat{\theta}(x;\hat{\varphi}(x)) \, \mid \, \Pi_{N_A ,N_B} \right)}{\sum_{i\in U}(W_i/n)H^2_i (x;\hat{\varphi}(x))}.
\end{align}
Now Condition 1 of Assumption~\ref{Assumption 2} implies that $\frac{n}{E[n\mid \Pi_{N_A ,N_B}]}\pconv d^{-1}\in\mathbb{R}_{++}$. Furthermore, the denominator in this expression is a consistent estimator of $S_{N_A,N_B}(x)$ in~(\ref{eq - S_NA_NB}). Combining these two points with the implications of Conditions 2(i) and 2(ii) of Assumption~\ref{Assumption 2} yield
\begin{align}
\text{Deff}(x)\pconv d^{-1}\frac{\lim\limits_{N_A, N_B \rightarrow \infty } \mathbbm{E}( n \mid \Pi_{N_A, N_B}) Var\left(\hat{\theta}(x;\hat{\varphi}(x)) \mid \Pi_{N_A, N_B} \right)}{\lim\limits_{N_A,N_B\rightarrow+\infty}S_{N_A,N_B}}.
\end{align}
Since $\lim\limits_{N_A,N_B\rightarrow+\infty}S_{N_A,N_B}$ and $\lim\limits_{N_A, N_B \rightarrow \infty } \mathbbm{E}( n \mid \Pi_{N_A, N_B}) Var\left(\hat{\theta}(x;\hat{\varphi}(x)) \mid \Pi_{N_A, N_B} \right)$ are both positive and finite by Assumption~\ref{Assumption 2}, the probability limit of $\text{Deff}(x)$ as $N_A,N_B\rightarrow+\infty$ is finite and positive.
\end{proof}

\subsubsection{Lemma \ref{power theorem Lemma 1}}

\begin{lemma}
Let $\mathbb{W}_1 (\varepsilon)$ be as in Definition~\ref{sequences of alternative hypothesis} and $C(x)$ be given by~(\ref{C function}). Fix an $\varepsilon > 0$ and $x \in [\underline{t},\overline{t}]$. For a sequence of finite populations $\{ \Pi_{N_A , N_B} \}_{N_A ,N_B = 1}^\infty \in \mathbb{W}_1 (\varepsilon)$, then
$$2(L_{UR} - L_R (x))/\widehat{\text{Deff}}(x) \pconv \infty \ \text{as} \ N_A , N_B \rightarrow \infty.$$
\label{power theorem Lemma 1}
\end{lemma}
\begin{proof}
We proceed via the direct method.  Fix an $x \in [\underline{t},\overline{t}]$ arbitrarily.  Following identical steps in Lemma \ref{Rami Lemma 1} to obtain \eqref{Lemma 1 log expression}:
\begin{equation*}
    2(L_{UR} - L_R (x)) = 2 \sum_{i\in U} W_i \log \left( 1 + \frac{\kappa H_i (x;\hat{\varphi}(x))}{n} \right)
\end{equation*}
where $\kappa$ is given by \eqref{eq: 4i}.  By the duality of the of the optimisation problem \eqref{Constrained Problem}, we have that:
\begin{equation}
    2(L_{UR} - L_R (x)) = \max_{\xi \in \mathbb{R}} \left[ 2 \sum_{i\in U} W^\prime_i \log \left( \frac{1}{n p_i ( \xi)} \right) \right]
\end{equation}
Where $p_i (\xi)$ is given by:
\begin{equation*}
    p_i (\xi) = \dfrac{1}{n + \xi H_i (x;\hat{\varphi}(x))}
\end{equation*}
Thus, by applying duality, we have that for any arbitrary $\xi \in \mathbb{R}$:
\begin{equation}
    2(L_{UR} - L_R (x)) \geq  2 \sum_{i\in U} W^\prime_i \log \left( 1 + \frac{\xi H_i (x;\hat{\varphi}(x))}{n} \right)
    \label{lower bound}
\end{equation}
We now choose $\xi$ dependent on the sequence of finite populations such that $\xi = o_{P} (n)$ as $N_A, N_B \rightarrow \infty$, that is:
\begin{equation*}
    \frac{\xi}{n} \mid \, \Pi_{N_A, N_B} \pconv 0 \quad \text{as} \ N_A ,N_B \rightarrow \infty
\end{equation*}
This property on $\xi$ is necessary to ensure that the expansion $\log(1+u) \aequal u - \dfrac{u^2}{2} \ \text{for} \ |u| < 1$, where $u = \dfrac{\xi H_i (x;\hat{\varphi}(x_m))}{n}$, is valid. Note that by Lemma~ \ref{H_i bounded}, $H_i (x;\hat{\varphi}(x))$ is uniformly bounded on $[\underline{t},\overline{t}]$. Setting $\xi$ as:
\begin{equation}
    \xi = \frac{-1}{\sqrt{Var\left(\hat{\theta}(x;\hat{\varphi}(x))\, \mid \, \Pi_{N_A ,N_B} \right) }}
    \label{xi substitution}
\end{equation}
Then using the second order logarithm expansion on \eqref{lower bound} and substituting \eqref{xi substitution} for $\xi$ we get:
\begin{align*}
     2(L_{UR} - L_R (x)) &\geq  2 \sum_{i\in U} W^\prime_i \log \left( 1 + \frac{\xi H_i (x;\hat{\varphi}(x))}{n} \right) \\
     & \aequal 2 \sum_{i\in U} W^\prime_i \left[\frac{\xi H_i (x;\hat{\varphi}(x))}{n} - \frac{1}{2} \left( \frac{\xi H_i (x;\hat{\varphi}(x))}{n} \right)^2 \right] \\
     & = -2 \frac{\hat{\theta}(x;\hat{\varphi}(x))}{\sqrt{Var\left(\hat{\theta}(x;\hat{\varphi}(x)) \, \mid \, \Pi_{N_A ,N_B} \right) }} - \frac{1}{\text{Deff}(x)} \\
     & = -2 \frac{\hat{\theta}(x;\hat{\varphi}(x)) - \theta_{N_A,N_B}(x;\varphi_{N_A,N_B}(x))}{\sqrt{Var\left(\hat{\theta}(x;\hat{\varphi}(x))\, \mid \, \Pi_{N_A ,N_B} \right) }} \\
     &\quad -2\frac{\theta_{N_A,N_B}(x;\varphi_{N_A,N_B}(x))}{\sqrt{Var\left(\hat{\theta}(x;\hat{\varphi}(x)) \, \mid \, \Pi_{N_A ,N_B} \right)}} - \frac{1}{\text{Deff}(x)}
     \stepcounter{equation}\tag{\theequation}\label{Rami B.6}
\end{align*}
Hence, asymptotically,
\begin{align*}
\frac{2(L_{UR} - L_R (x))}{\widehat{\text{Deff}}(x)} &\geq \frac{-2}{\widehat{\text{Deff}}(x)}\left[\frac{\hat{\theta}(x;\hat{\varphi}(x)) - \theta_{N_A,N_B}(x;\varphi_{N_A,N_B}(x))}{\sqrt{Var\left(\hat{\theta}(x;\hat{\varphi}(x))\, \mid \, \Pi_{N_A ,N_B} \right) }}\right] \\
     &\quad  -\frac{2}{\widehat{\text{Deff}}(x)}\left[\frac{\theta_{N_A,N_B}(x;\varphi_{N_A,N_B}(x))}{\sqrt{Var\left(\hat{\theta}(x;\hat{\varphi}(x)) \, \mid \, \Pi_{N_A ,N_B} \right)}}\right] - \frac{1}{\widehat{\text{Deff}}(x)\,\text{Deff}(x)}\stepcounter{equation}\tag{\theequation}\label{Rami B.7}
\end{align*}
Conditions 4 and 5 of Assumption \ref{Assumption 1} and  Lemma~\ref{Lemma Deff limit} imply the first term on the right side in \eqref{Rami B.7} converges in distribution to $$N\left(0,\frac{4}{\lim_{N_A,N_B\rightarrow+\infty}\text{Deff}(x)^2}\right).$$
The second term on the right side in \eqref{Rami B.7} converges to $-2C(x)/\lim_{N_A,N_B\rightarrow+\infty}\text{Deff}(x)$. Since the denominator is finite and positive, this limit is determined by the value of $C(x)$. As $C(x)=-\infty$, the limit of the second term diverges to $\infty$. Finally, the limit of third term on the right side in \eqref{Rami B.7} is $1/\lim_{N_A,N_B\rightarrow+\infty}\text{Deff}(x)^2\in\mathbb{R}_{++}$ by Condition 5 of Assumption~\ref{Assumption 1} and Lemma~\ref{Lemma Deff limit}.

\par Thus, as a result we have that $2(L_{UR} - L_R (x))/\widehat{\text{Deff}}(x) \pconv \infty$ as $N_A, N_B \rightarrow \infty$. We now demonstrate this choice of $\xi$ satisfies $\xi = o_P (n)$.  Note that:
\begin{align*}
    \left|\frac{\xi}{n}\right| &= \frac{1}{n \sqrt{Var\left(\hat{\theta}(x;\hat{\varphi}(x))\, \mid \, \Pi_{N_A ,N_B} \right) }} \\
    & = \frac{\sqrt{\dfrac{\mathbbm{E}( n \mid \Pi_{N_A, N_B})}{n}} \sqrt{\dfrac{1}{n}}}{\sqrt{\mathbbm{E}( n \mid \Pi_{N_A, N_B}) Var\left(\hat{\theta}(x;\hat{\varphi}(x)) \, \mid \, \Pi_{N_A ,N_B} \right)}}
\end{align*}
Then using Conditions 1 and 2 of Assumption \ref{Assumption 2}, we find the limiting behaviour is determined by $\frac{1}{\sqrt{n}}$. Therefore, since $n \pconv \infty$ as $N_A, N_B \rightarrow \infty$ (i.e., implied by Condition 1 of Assumption~\ref{Assumption 1}), we have that $\xi=o_P (n)$.
\end{proof}

\subsubsection{Lemma \ref{power theorem Lemma 2}}

\begin{lemma}
Let $\mathbb{W}_1 (\varepsilon)$ be defined the same as in Definition \ref{sequences of alternative hypothesis} and $C(x)$ is given by~(\ref{C function}). For a sequence of finite populations $\{ \Pi_{N_A} , {N_B} \}_{N_A ,N_B = 1}^\infty \in \mathbb{W}_1 (0)$ define the set $X_0 = \{ x \in [\underline{t}, \overline{t}] : C(x) = 0\}$.  If $X_0 \neq \emptyset$ then for $x \in X_0$:
$$2(L_{UR} - L_R (x ))/\widehat{\text{Deff}}(x) \mid \Pi_{N_A, N_B} \dconv \chi_1^2 \ \text{as} \ N_A , N_B \rightarrow \infty.$$
\label{power theorem Lemma 2}
\end{lemma}

\begin{proof}
The asymptotic power of the test will be lowest for sequences of finite populations that converge to the boundary of the model of the alternative hypothesis $H_1^1$.  This will be where for every finite $N_A, N_B$ the finite population is in $\mathcal{M}_{N_A ,N_B}^1$ but converges to the model of the null hypothesis $H_0^1$ in the limit. That is, for one or more $x \in [\underline{t}, \overline{t}]$ we will have $C(x) = 0$.  Let $X_0$ be the set $\{x \in [\underline{t}, \overline{t}] : C(x) = 0 \}$.  Fix an $x\in X_0$.  Following identical steps as in Lemma \ref{Rami Lemma 1}, we obtain:
\begin{align*}
    2(L_{UR} - L_R (x)) \aequal& \left( \frac{\hat{\theta}(x;\hat{\varphi}(x))}{\sqrt{Var\left(\hat{\theta}(x;\hat{\varphi}(x)) \right)}} \right)^2 \frac{Var\left(\hat{\theta}(x;\hat{\varphi}(x)) \right)}{\sum_{i\in U} W^\prime_i \dfrac{1}{n^2} H^2_i (x;\hat{\varphi}(x))} \\
    =& \left( \frac{\hat{\theta}(x;\hat{\varphi}(x)) - \theta_{N_A,N_B}(x;\varphi_{N_A,N_B}(x)) }{\sqrt{Var\left(\hat{\theta}(x;\hat{\varphi}(x)) \right)}} + \frac{\theta_{N_A,N_B}(x;\varphi_{N_A,N_B}(x)) }{\sqrt{Var\left(\hat{\theta}(x;\hat{\varphi}(x)) \right)}} \right)^2 \,\text{Deff}(x)
\end{align*}
Using Conditions 4 and 5 of Assumption \ref{Assumption 1} and that $C(x) = 0$ for $x \in X_0$,
\begin{equation*}
    2(L_{UR} - L_R (x))/\widehat{\text{Deff}}(x)  \dconv \mathbb{G}^2(x) \ \text{as} \ N_A ,N_B \rightarrow \infty.
\end{equation*}
For each $x$, the random variable $\mathbb{G} (x)$ is distributed as standard normal; therefore, $\mathbb{G}^2(x)\sim \chi_1^2$.
\end{proof}

\section{Technical Results for the Examples in Section~\ref{Subsection Examples}}\label{Appendix D}
This section presents the identified sets of $D^s_{N_K}(x)$ for each $x\in[\underline{t},\overline{t}]$ and $K=A,B$, associated with the Examples in Section~\ref{Subsection Examples}. These sets are characterized in terms of sharp bounds $$\underline{D}_{N_K}^{s}(x)\leq D^s_{N_K}(x)\leq \overline{D}_{N_A}^{s}(x)\quad x\in[\underline{t},\overline{t}],$$ and the results of this section describe the bounds under the informational setups of the examples.
\subsection{Example~\ref{Example - WC Bounds}: Worst-Case Bounds}
\begin{proposition}\label{prop - WC ID Set}
For each $x\in[\underline{t},\overline{t}]$ and $s\in\mathbb{N}$, the worst-case bounds are
\begin{align*}
\underline{D}_{N_A}^{s}(x) & = \mathbbm{E}_{F_A\left(x| Z^A=Z^B=1\right)}\left[\frac{(x-Y^A)^{s-1}}{(s-1)!}\,\mathbbm{1}\left[Y^A\leq x\right]\right]\,\delta_{11}\\
& \qquad+\mathbbm{E}_{F_A\left(x| Z^A=1,Z^B=0\right)}\left[\frac{(x-Y^A)^{s-1}}{(s-1)!}\,\mathbbm{1}\left[Y^A\leq x\right]\right]\,\delta_{10},\\
\overline{D}_{N_A}^{s}(x) & = \underline{D}_{N_A}^{s}(x)+\frac{(x-\underline{Y}^A)^{s-1}}{(s-1)!}\,\delta_{00},\\
\underline{D}_{N_B}^{s}(x) & =  \mathbbm{E}_{F_B\left(x| Z^A=Z^B=1\right)}\left[\frac{(x-Y^B)^{s-1}}{(s-1)!}\,\mathbbm{1}\left[Y^B\leq x\right]\right]\,\delta_{11},\\
\overline{D}_{N_B}^{s}(x) & = \underline{D}_{N_B}^{s}(x)+\frac{(x-\underline{Y}^B)^{s-1}}{(s-1)!}\,(\delta_{00}+\delta_{10}),
\end{align*}
where $\underline{Y}^K=\inf\mathcal{X}_K$, for $K=A,B$.
\end{proposition}
\begin{proof}
The proof proceeds by the direct method. Let $K=A$, and fix $s\in\mathbb{N}$ and $x\in[\underline{t},\overline{t}]$. The law of total probability expansion of $D_{N_A}^s (x)$, given by~(\ref{eq - LTP Expansion}), is
\begin{align*}
\mathbbm{E}_{F_A\left(\cdot| Z^A=Z^B=1\right)}& \left[\frac{(x-Y^A)^{s-1}}{(s-1)!}\,\mathbbm{1}\left[Y^A\leq x\right]\right]\,\delta_{11}+\mathbbm{E}_{F_A\left(\cdot| Z^A=1,Z^B=0\right)}\left[\frac{(x-Y^A)^{s-1}}{(s-1)!}\,\mathbbm{1}\left[Y^A\leq x\right]\right]\,\delta_{10}\\
 & \qquad+\mathbbm{E}_{F_A\left(\cdot| Z^A=Z^B=0\right)}\left[\frac{(x-Y^A)^{s-1}}{(s-1)!}\,\mathbbm{1}\left[Y^A\leq x\right]\right]\,\delta_{00}.
\end{align*}
Note that $\mathbbm{E}_{F_A\left(\cdot| Z^A=Z^B=0\right)}\left[\frac{(x-Y^A)^{s-1}}{(s-1)!}\,\mathbbm{1}\left[Y^A\leq x\right]\right]$ is not identified, but can be bounded from below and above by 0 and
$\frac{(x-\underline{Y}^A)^{s-1}}{(s-1)!}$, respectively. These bounds on that conditional expectation imply the desired bounds on $D_{N_A}^s (x)$ by direct substitution. Now since $s\in\mathbb{Z}_+$ and $x\in[\underline{t},\overline{t}]$ were arbitrary, the above bounds hold for each $s\in\mathbb{N}$ and $x\in[\underline{t},\overline{t}]$.

\par Let $K=B$, and fix $s\in\mathbb{N}$ and $x\in[\underline{t},\overline{t}]$. The law of total probability expansion of $D_{N_B}^s (x)$, given by~(\ref{eq - LTP Expansion}), is
\begin{align*}
\mathbbm{E}_{F_B\left(\cdot| Z^A=Z^B=1\right)}& \left[\frac{(x-Y^B)^{s-1}}{(s-1)!}\,\mathbbm{1}\left[Y^B\leq x\right]\right]\,\delta_{11}+\mathbbm{E}_{F_B\left(\cdot| Z^A=1,Z^B=0\right)}\left[\frac{(x-Y^B)^{s-1}}{(s-1)!}\,\mathbbm{1}\left[Y^B\leq x\right]\right]\,\delta_{10}\\
 & \qquad+\mathbbm{E}_{F_B\left(\cdot| Z^A=Z^B=0\right)}\left[\frac{(x-Y^B)^{s-1}}{(s-1)!}\,\mathbbm{1}\left[Y^B\leq x\right]\right]\,\delta_{00}.
\end{align*}
Note that $\mathbbm{E}_{F_B\left(\cdot| Z^A=1,Z^B=0\right)}\left[\frac{(x-Y^B)^{s-1}}{(s-1)!}\,\mathbbm{1}\left[Y^B\leq x\right]\right]$ and $\mathbbm{E}_{F_B\left(\cdot| Z^A=Z^B=0\right)}\left[\frac{(x-Y^B)^{s-1}}{(s-1)!}\,\mathbbm{1}\left[Y^B\leq x\right]\right]$ not identified, but are bounded from below by 0 and from above by $\frac{(x-\underline{Y}^B)^{s-1}}{(s-1)!}$. These bounds on that conditional expectation imply the desired bounds on $D_{N_B}^s (x)$ by direct substitution. Now since $s\in\mathbb{N}$ and $x\in[\underline{t},\overline{t}]$ were arbitrary, the above bounds hold for each $s\in\mathbb{N}$ and $x\in[\underline{t},\overline{t}]$.
\end{proof}
\subsection{Example~\ref{Example - MCAR and WC Bounds}}
Recall that the assumption of MCAR for \emph{unit} nonresponse means
\begin{align}\label{eq - Appendix MCAR Example}
F_K\left(x| Z^A=Z^B=1\right)=F_K\left(x| Z^A=Z^B=0\right)\quad\forall x\in\mathcal{X}_K
\end{align}
holds, for $K=A,B$, and that this condition does not imply any restrictions on wave nonrepsonse. We have the following result.
\begin{proposition}\label{prop - MCAR UNR WC WNR ID Set}
Suppose that~(\ref{eq - Appendix MCAR Example}) holds. For each $x\in[\underline{t},\overline{t}]$ and $s\in\mathbb{N}$, the corresponding bounds are
\begin{align*}
\underline{D}_{N_A}^{s}(x) & = \mathbbm{E}_{F_A\left(x| Z^A=Z^B=1\right)}\left[\frac{(x-Y^A)^{s-1}}{(s-1)!}\,\mathbbm{1}\left[Y^A\leq x\right]\right]\,(\delta_{11}+\delta_{00})\\
& \qquad +\mathbbm{E}_{F_A\left(x| Z^A=1,Z^B=0\right)}\left[\frac{(x-Y^A)^{s-1}}{(s-1)!}\,\mathbbm{1}\left[Y^A\leq x\right]\right]\,\delta_{10}\\
\overline{D}_{N_A}^{s}(x) & = \underline{D}_{N_A}^{s}(x),\\
\underline{D}_{N_B}^{s}(x) & =  \mathbbm{E}_{F_B\left(x| Z^A=Z^B=1\right)}\left[\frac{(x-Y^B)^{s-1}}{(s-1)!}\,\mathbbm{1}\left[Y^B\leq x\right]\right]\,(\delta_{11}+\delta_{00}),\\
\overline{D}_{N_B}^{s}(x) & = \underline{D}_{N_B}^{s}(x)+\frac{(x-\underline{Y}^B)^{s-1}}{(s-1)!}\,\delta_{10},
\end{align*}
where $\underline{Y}^K=\inf\mathcal{X}_K$, for $K=A,B$.
\end{proposition}
\begin{proof}
The proof proceeds by the direct method. Let $K=A$, and fix $s\in\mathbb{N}$ and $x\in[\underline{t},\overline{t}]$. The law of total probability expansion of $D_{N_A}^s (x)$, given by~(\ref{eq - LTP Expansion}), is
\begin{align*}
\mathbbm{E}_{F_A\left(\cdot| Z^A=Z^B=1\right)}& \left[\frac{(x-Y^A)^{s-1}}{(s-1)!}\,\mathbbm{1}\left[Y^A\leq x\right]\right]\,\delta_{11}+\mathbbm{E}_{F_A\left(\cdot| Z^A=1,Z^B=0\right)}\left[\frac{(x-Y^A)^{s-1}}{(s-1)!}\,\mathbbm{1}\left[Y^A\leq x\right]\right]\,\delta_{10}\\
 & \qquad+\mathbbm{E}_{F_A\left(\cdot| Z^A=Z^B=0\right)}\left[\frac{(x-Y^A)^{s-1}}{(s-1)!}\,\mathbbm{1}\left[Y^A\leq x\right]\right]\,\delta_{00}.
\end{align*}
Condition~(\ref{eq - Appendix MCAR Example}) implies $$\mathbbm{E}_{F_A\left(\cdot| Z^A=Z^B=1\right)}\left[\frac{(x-Y^A)^{s-1}}{(s-1)!}\,\mathbbm{1}\left[Y^A\leq x\right]\right]=\mathbbm{E}_{F_A\left(\cdot| Z^A=Z^B=0\right)}\left[\frac{(x-Y^A)^{s-1}}{(s-1)!}\,\mathbbm{1}\left[Y^A\leq x\right]\right],$$
holds, for each $x\in\mathcal{X}_A$. Hence, $D_{N_A}^{s}$ is point-identified and given by
\begin{align*}
D_{N_A}^{s}(x) & = \mathbbm{E}_{F_A\left(x| Z^A=Z^B=1\right)}\left[\frac{(x-Y^A)^{s-1}}{(s-1)!}\,\mathbbm{1}\left[Y^A\leq x\right]\right]\,(\delta_{11}+\delta_{00})\\
& \qquad +\mathbbm{E}_{F_A\left(x| Z^A=1,Z^B=0\right)}\left[\frac{(x-Y^A)^{s-1}}{(s-1)!}\,\mathbbm{1}\left[Y^A\leq x\right]\right]\,\delta_{10}
\end{align*}
for each $x\in\mathcal{X}_A$, implying the desired result.

\par Let $K=B$, and fix $s\in\mathbb{N}$ and $x\in[\underline{t},\overline{t}]$. The law of total probability expansion of $D_{N_B}^s (x)$, given by~(\ref{eq - LTP Expansion}), is
\begin{align*}
\mathbbm{E}_{F_B\left(\cdot| Z^A=Z^B=1\right)}& \left[\frac{(x-Y^B)^{s-1}}{(s-1)!}\,\mathbbm{1}\left[Y^B\leq x\right]\right]\,\delta_{11}+\mathbbm{E}_{F_B\left(\cdot| Z^A=1,Z^B=0\right)}\left[\frac{(x-Y^B)^{s-1}}{(s-1)!}\,\mathbbm{1}\left[Y^B\leq x\right]\right]\,\delta_{10}\\
 & \qquad+\mathbbm{E}_{F_B\left(\cdot| Z^A=Z^B=0\right)}\left[\frac{(x-Y^B)^{s-1}}{(s-1)!}\,\mathbbm{1}\left[Y^B\leq x\right]\right]\,\delta_{00}.
\end{align*}
Condition~(\ref{eq - Appendix MCAR Example}) implies $$\mathbbm{E}_{F_B\left(\cdot| Z^A=Z^B=1\right)}\left[\frac{(x-Y^B)^{s-1}}{(s-1)!}\,\mathbbm{1}\left[Y^B\leq x\right]\right]=\mathbbm{E}_{F_B\left(\cdot| Z^A=Z^B=0\right)}\left[\frac{(x-Y^B)^{s-1}}{(s-1)!}\,\mathbbm{1}\left[Y^A\leq x\right]\right],$$
holds, for each $x\in\mathcal{X}_B$. Consequently, only $\mathbbm{E}_{F_B\left(\cdot| Z^A=1,Z^B=0\right)}\left[\frac{(x-Y^B)^{s-1}}{(s-1)!}\,\mathbbm{1}\left[Y^B\leq x\right]\right]$ in this expansion is not identified, but it is bounded from below by 0 and from above by $\frac{(x-\underline{Y}^B)^{s-1}}{(s-1)!}$. These bounds on that conditional expectation imply the desired bounds on $D_{N_B}^s (x)$ by direct substitution. Now since $s\in\mathbb{N}$ and $x\in[\underline{t},\overline{t}]$ were arbitrary, the above bounds hold for each $s\in\mathbb{N}$ and $x\in[\underline{t},\overline{t}]$.
\end{proof}

\subsection{Example~\ref{Example - NBD of MCAR}: Neighborhood of MCAR}
\par Recall that the neighborhood condition in Example~\ref{Example - NBD of MCAR} specifies the following system of inequalities for the conditional probabilities $\text{Prob}\left(Z^A=Z^B=0| Y^A\leq x\right)$, $\text{Prob}\left(Z^A=Z^B=0| Y^B\leq x\right)$, and \\ $\text{Prob}\left(Z^A=1,Z^B=0| Y^B\leq x\right)$:
\begin{align}
 \delta_{00}\,L^{K}_{00}(x) & \leq \text{Prob}\left(Z^A=Z^B=0| Y^K\leq x\right)\leq U^{K}_{00}(x) \,\delta_{00},\quad K=A,B,\quad\text{and}\label{eq - Appendix NBD MCAR Example 0}\\
 \delta_{10}\,L_{10}(x) & \leq \text{Prob}\left(Z^A=1,Z^B=0| Y^B\leq x\right)\leq U_{10}(x)\,\delta_{10},\label{eq - Appendix NBD MCAR Example 1}
\end{align}
where $L^{A}_{00}$, $U^{A}_{00}$, $L^{B}_{00}$, $U^{B}_{00}$, $L_{10}$, and $U_{10}$ are CDFs that are predesignated by the practitioner which satisfy $L^{K}_{00}(\cdot)\leq U^{K}_{00}(\cdot)$, for $K=A,B$, and $L_{10}\leq U_{10}$.

\par Let
\begin{align*}
\overline{G}_A(x) & =\frac{1-\delta_{00}}{1-U^{A}_{00}(x)\delta_{00}}\quad\text{and}\quad\underline{G}_A(x)=\frac{1-\delta_{00}}{1-L^{A}_{00}(x)\delta_{00}}\quad\forall x\in\mathcal{X}_A,\\
\overline{G}_B(x) & =\frac{\delta_{11}}{1-U^{B}_{00}(x)\delta_{00}-U^{B}_{10}(x)\delta_{10}}\quad\text{and}\quad\underline{G}_B(x)=\frac{\delta_{11}}{1-L^{B}_{00}(x)\delta_{00}-L_{10}(x)\delta_{10}}\quad\forall x\in\mathcal{X}_B,
\end{align*}
which are CDFs. Hence, we can define the dominance functions recursively: for $G\in\{\overline{G}_A,\underline{G}_A,\overline{G}_B,\underline{G}_B\}$, these functions are $D_{G}^1 (x)= G(x)$ and  $D^{s}_{G}(x) =\int_{-\infty}^{x}D^{s-1}_{G}(u)\, \text{d}u$ for $s=2,3,4,\ldots$. Finally, for we also introduce the following recursively defined functions defined on $\mathbb{R}^2$: $R_0(y,x)=\mathbbm{1}[y\leq x]$, $R_j(y,x)=\int_{-\infty}^x R_{j-1}(y,u)\,du$ for $j=1,2,\ldots$.

\begin{proposition}\label{prop - NBD of MCAR ID Set}
Suppose that $\text{Prob}\left(Z^A=Z^B=0| Y^A\leq x\right)$, $\text{Prob}\left(Z^A=Z^B=0| Y^B\leq x\right)$, and \\$\text{Prob}\left(Z^A=1,Z^B=0| Y^B\leq x\right)$ satisfy the inequalities~(\ref{eq - Appendix NBD MCAR Example 0}) and~(\ref{eq - Appendix NBD MCAR Example 1}). For each $x\in[\underline{t},\overline{t}]$ and $s\in\mathbb{N}$, the corresponding bounds on the dominance functions are
\begin{align*}
\underline{D}_{N_A}^{s}(x) & =\left(\frac{\delta_{11}}{1-\delta_{00}}\,F_A\left(x| Z^A=Z^B=1\right)+\frac{\delta_{10}}{1-\delta_{00}}\,F_A\left(x| Z^A=1,Z^B=0\right)\right)D^{s}_{\underline{G}_A}(x)\\
& \qquad-\mathbbm{1}[s>1]\frac{\delta_{11}}{1-\delta_{00}}\sum_{j=0}^{s-2}\mathbbm{E}_{F_A\left(\cdot| Z^A=Z^B=1\right)}\left[D^{s}_{\underline{G}_A}(Y^A)\,R_j(Y^A,x)\right]\\
& \qquad-\mathbbm{1}[s>1]\frac{\delta_{10}}{1-\delta_{00}}\sum_{j=0}^{s-2}\mathbbm{E}_{F_A\left(\cdot| Z^A=1,Z^B=0\right)}\left[D^{s}_{\underline{G}_A}(Y^A)\,R_j(Y^A,x)\right]\\
\overline{D}_{N_A}^{s}(x) & =\left(\frac{\delta_{11}}{1-\delta_{00}}\,F_A\left(x| Z^A=Z^B=1\right)+\frac{\delta_{10}}{1-\delta_{00}}\,F_A\left(x| Z^A=1,Z^B=0\right)\right)D^{s}_{\overline{G}_A}(x)\\
& \qquad-\mathbbm{1}[s>1]\frac{\delta_{11}}{1-\delta_{00}}\sum_{j=0}^{s-2}\mathbbm{E}_{F_A\left(\cdot| Z^A=Z^B=1\right)}\left[D^{s}_{\overline{G}_A}(Y^A)\,R_j(Y^A,x)\right]\\
& \qquad-\mathbbm{1}[s>1]\frac{\delta_{10}}{1-\delta_{00}}\sum_{j=0}^{s-2}\mathbbm{E}_{F_A\left(\cdot| Z^A=1,Z^B=0\right)}\left[D^{s}_{\overline{G}_A}(Y^A)\,R_j(Y^A,x)\right]\\
\underline{D}_{N_B}^{s}(x) & =F_B\left(x| Z^A=Z^B=1\right)\,D^{s}_{\underline{G}_B}(x)-\mathbbm{1}[s>1]\sum_{j=0}^{s-2}\mathbbm{E}_{F_B\left(\cdot| Z^A=Z^B=1\right)}\left[D^{s}_{\underline{G}_B}(Y^B)\,R_j(Y^B,x)\right]\\
\overline{D}_{N_B}^{s}(x) & = F_B\left(x| Z^A=Z^B=1\right)\,D^{s}_{\overline{G}_B}(x)-\mathbbm{1}[s>1]\sum_{j=0}^{s-2}\mathbbm{E}_{F_B\left(\cdot| Z^A=Z^B=1\right)}\left[D^{s}_{\overline{G}_B}(Y^B)\,R_j(Y^B,x)\right].
\end{align*}
\end{proposition}
\begin{proof}
The proof proceeds by the direct method. We start with population $A$. For $s=1$ and $x\in\mathcal{X}_A$, the Law of Total Probability applied to $F_A(x)$ yields
\begin{align*}
F_A(x)=F_A\left(x|Z^A=Z^B=1\right)\,\delta_{11}+ F_A\left(x|Z^A=1,Z^B=0\right)\,\delta_{10}+F_A\left(x|Z^A=Z^B=0\right)\,\delta_{00}.
\end{align*}
Next, apply Bayes' Theorem to $F_A\left(x|Z^A=Z^B=0\right)$ to express it in terms of $\text{Prob}\left(Z^A=Z^B=0| Y^A\leq x\right)$, substitute this resulting expression into the
above representation of $F_A(x)$, and solving for $F_A(x)$, yields
\begin{align*}
F_A(x) & =\frac{F_A\left(x|Z^A=Z^B=1\right)}{1-\text{Prob}\left(Z^A=Z^B=0| Y^A\leq x\right)}\,\delta_{11}+ \frac{F_A\left(x|Z^A=1,Z^B=0\right)}{1-\text{Prob}\left(Z^A=Z^B=0| Y^A\leq x\right)}\,\delta_{10}.
\end{align*}
As $D_{N_A}^{1}(x)=F_A(x)$, we substitute out $\text{Prob}\left(Z^A=Z^B=0| Y^A\leq x\right)$ in the above expression using the inequalities~(\ref{eq - Appendix NBD MCAR Example 0}) to obtain
\begin{align*}
\underline{D}_{N_A}^{1}(x) & =\left(\frac{\delta_{11}}{1-\delta_{00}}\,F_A\left(x| Z^A=Z^B=1\right)+\frac{\delta_{10}}{1-\delta_{00}}\,F_A\left(x| Z^A=1,Z^B=0\right)\right)D^{1}_{\underline{G}_A}(x)\quad \text{and}\\
\overline{D}_{N_A}^{1}(x) & =\left(\frac{\delta_{11}}{1-\delta_{00}}\,F_A\left(x| Z^A=Z^B=1\right)+\frac{\delta_{10}}{1-\delta_{00}}\,F_A\left(x| Z^A=1,Z^B=0\right)\right)D^{1}_{\overline{G}_A}(x).
\end{align*}
As $x\in\mathcal{X}_A$ was arbitrary, the above derivation holds for each $x\in\mathcal{X}_A$.

\par Now we focus on $s=2$. Fix $x\in\mathcal{X}_A$. Since $\underline{D}_{N_A}^{1}(x)\leq D_{N_A}^{1}(x)\leq \overline{D}_{N_A}^{1}(x)$ and $D_{N_A}^{2}(x)=\int_{-\infty}^{x}D_{N_A}^{1}(u)\,du$, we can integrate all sides to obtain
\begin{align*}
\underline{D}_{N_A}^{2}(x)=\int_{-\infty}^{x}\underline{D}_{N_A}^{1}(u)\,du\leq D_{N_A}^{2}(x)\leq\int_{-\infty}^{x} \overline{D}_{N_A}^{1}(u)\,du=\overline{D}_{N_A}^{2}(x).
\end{align*}
Since $\underline{D}_{N_A}^{1}$ and $\overline{D}_{N_A}^{1}$ are mixtures of products of CDFs, we can use integration by parts to obtain expressions for $\underline{D}_{N_A}^{2}(x)$ and $\overline{D}_{N_A}^{2}(x)$, given by
 \begin{align*}
 \left\{F_A\left(x|Z^A=Z^B=1\right)\,\int_{-\infty}^x D^{1}_{\underline{G}_A}(u)\,du-\int_{-\infty}^x\int_{-\infty}^{u} \underline{G}_A(u^\prime)\,du^\prime dF_A\left(u|Z^A=Z^B=1\right)\right\}\frac{\delta_{11}}{1-\delta_{00}}\\
 +\left\{F_A\left(x|Z^A=1,Z^B=0\right)\,\int_{-\infty}^x D^{1}_{\underline{G}_A}(u)\,du-\int_{-\infty}^x\int_{-\infty}^{u} \underline{G}_A(u^\prime)\,du^\prime dF_A\left(u|Z^A=0,Z^B=1\right)\right\}\frac{\delta_{10}}{1-\delta_{00}},
 \end{align*}
 and
 \begin{align*}
 \left\{F_A\left(x|Z^A=Z^B=1\right)\,\int_{-\infty}^x D^{1}_{\overline{G}_A}(u)\,du-\int_{-\infty}^x\int_{-\infty}^{u} \overline{G}_A(u^\prime)\,du^\prime dF_A\left(u|Z^A=Z^B=1\right)\right\}\frac{\delta_{11}}{1-\delta_{00}}\\
 +\left\{F_A\left(x|Z^A=1,Z^B=0\right)\,\int_{-\infty}^x D^{1}_{\overline{G}_A}(u)\,du-\int_{-\infty}^x\int_{-\infty}^{u} \overline{G}_A(u^\prime)\,du^\prime dF_A\left(u|Z^A=0,Z^B=1\right)\right\}\frac{\delta_{10}}{1-\delta_{00}},
 \end{align*}
 respectively. Now recognizing that $D^{2}_{\underline{G}_A}(x)=\int_{-\infty}^x D^{1}_{\underline{G}_A}(u)\,du$, $D^{2}_{\overline{G}_A}(x)=\int_{-\infty}^x D^{1}_{\overline{G}_A}(u)\,du$,
 and
 \begin{align*}
 \int_{-\infty}^x\int_{-\infty}^{u} \underline{G}_A(u^\prime)\,du^\prime dF_A\left(u| Z^A=z_A,Z^B=z_B\right) & =\int_{-\infty}^x D^{1}_{\underline{G}_A}(u)dF_A\left(u| Z^A=z_A,Z^B=z_B\right)\\
 & = \mathbbm{E}_{F_A\left(\cdot| Z^A=z_A,Z^B=z_B\right)}\left[D^{2}_{\underline{G}_A}(Y^A)\,R_0(Y^A,x)\right],
 \end{align*}
 we are done with $s=2$, as $x\in\mathcal{X}_A$ was arbitrary.

\par Now we focus on $s=3$. Let $x\in\mathcal{X}_A$. Since $\underline{D}_{N_A}^{2}(x)\leq D_{N_A}^{2}(x)\leq \overline{D}_{N_A}^{2}(x)$ and $D_{N_A}^{3}(x)=\int_{-\infty}^{x}D_{N_A}^{2}(u)\,du$, we can integrate all sides of these inequalities to obtain
\begin{align*}
\underline{D}_{N_A}^{3}(x)=\int_{-\infty}^{x}\underline{D}_{N_A}^{2}(u)\,du\leq D_{N_A}^{3}(x)\leq\int_{-\infty}^{x} \overline{D}_{N_A}^{2}(u)\,du=\overline{D}_{N_A}^{3}(x).
\end{align*}
 Noting that
 \begin{align*}
\underline{D}_{N_A}^{3}(x) & = \int_{-\infty}^{x}\left(\frac{\delta_{11}}{1-\delta_{00}}\,F_A\left(u| Z^A=Z^B=1\right)+\frac{\delta_{10}}{1-\delta_{00}}\,F_A\left(u| Z^A=1,Z^B=0\right)\right)D^{2}_{\underline{G}_A}(u)\,du\\
 & \quad-\frac{\delta_{11}}{1-\delta_{00}}\, \int_{-\infty}^{x}\mathbbm{E}_{F_A\left(\cdot| Z^A=Z^B=1\right)}\left[D^{2}_{\underline{G}_A}(Y^A)\,R_0(Y^A,u)\right]\,du\\
 & \quad -\frac{\delta_{10}}{1-\delta_{00}}\, \int_{-\infty}^{x}\mathbbm{E}_{F_A\left(\cdot| Z^A=1,Z^B=0\right)}\left[D^{2}_{\underline{G}_A}(Y^A)\,R_0(Y^A,u)\right]\,du,\quad\text{and}\\
 \overline{D}_{N_A}^{3}(x) & =  \int_{-\infty}^{x}\left(\frac{\delta_{11}}{1-\delta_{00}}\,F_A\left(u| Z^A=Z^B=1\right)+\frac{\delta_{10}}{1-\delta_{00}}\,F_A\left(u| Z^A=1,Z^B=0\right)\right)D^{2}_{\overline{G}_A}(u)\,du\\
 & \quad-\frac{\delta_{11}}{1-\delta_{00}}\, \int_{-\infty}^{x}\mathbbm{E}_{F_A\left(\cdot| Z^A=Z^B=1\right)}\left[D^{2}_{\overline{G}_A}(Y^A)\,R_0(Y^A,u)\right]\,du\\
 & \quad -\frac{\delta_{10}}{1-\delta_{00}}\, \int_{-\infty}^{x}\mathbbm{E}_{F_A\left(\cdot| Z^A=1,Z^B=0\right)}\left[D^{2}_{\overline{G}_A}(Y^A)\,R_0(Y^A,u)\right]\,du,
 \end{align*}
 we can apply the same integration by parts technique, as in the case $s=2$, to the first terms of these expressions to simplify them, and for the remaining terms we interchange integration and expectation to simply them. Implementing these operations results in
 \begin{align*}
\underline{D}_{N_A}^{3}(x) & = \left(\frac{\delta_{11}}{1-\delta_{00}}\,F_A\left(x| Z^A=Z^B=1\right)+\frac{\delta_{10}}{1-\delta_{00}}\,F_A\left(x| Z^A=1,Z^B=0\right)\right)D^{3}_{\underline{G}_A}(x)\\
& -\frac{\delta_{11}}{1-\delta_{00}}\left(\mathbbm{E}_{F_A\left(\cdot| Z^A=Z^B=1\right)}\left[D^{3}_{\underline{G}_A}(Y^A)\,R_0(Y^A,x)\right]+\mathbbm{E}_{F_A\left(\cdot| Z^A=Z^B =1\right)}\left[D^{2}_{\underline{G}_A}(Y^A)\,R_1(Y^A,x)\right]\right)\\
 & -\frac{\delta_{10}}{1-\delta_{00}}\left(\mathbbm{E}_{F_A\left(\cdot| Z^A=1,Z^B=0\right)}\left[D^{3}_{\underline{G}_A}(Y^A)\,R_0(Y^A,x)\right]+\mathbbm{E}_{F_A\left(\cdot| Z^A=1,Z^B =0\right)}\left[D^{2}_{\underline{G}_A}(Y^A)\,R_1(Y^A,x)\right]\right),\quad\text{and}\\
 \overline{D}_{N_A}^{3}(x) & = \left(\frac{\delta_{11}}{1-\delta_{00}}\,F_A\left(u| Z^A=Z^B=1\right)+\frac{\delta_{10}}{1-\delta_{00}}\,F_A\left(u| Z^A=1,Z^B=0\right)\right)D^{3}_{\overline{G}_A}(x)\\
 &-\frac{\delta_{11}}{1-\delta_{00}}\left(\mathbbm{E}_{F_A\left(\cdot| Z^A=Z^B=1\right)}\left[D^{3}_{\overline{G}_A}(Y^A)\,R_0(Y^A,x)\right]\mathbbm{E}_{F_A\left(\cdot| Z^A=Z^B =1\right)}\left[D^{2}_{\overline{G}_A}(Y^A)\,R_1(Y^A,x)\right]\right)\\
 & -\frac{\delta_{10}}{1-\delta_{00}}\left(\mathbbm{E}_{F_A\left(\cdot| Z^A=1,Z^B=0\right)}\left[D^{2}_{\overline{G}_A}(Y^A)\,R_0(Y^A,x)\right]+\mathbbm{E}_{F_A\left(\cdot| Z^A=1,Z^B =0\right)}\left[D^{2}_{\overline{G}_A}(Y^A)\,R_1(Y^A,x)\right]\right)
 \end{align*}
 where we have used $D^{3}_{\underline{G}_A}(x)=\int_{-\infty}^x D^{2}_{\underline{G}_A}(u)\,du$, $D^{3}_{\overline{G}_A}(x)=\int_{-\infty}^x D^{2}_{\overline{G}_A}(u)\,du$,
 \begin{align*}
 \int_{-\infty}^{x}\mathbbm{E}_{F_A\left(\cdot| Z^A=z_A,Z^B=z_B\right)}\left[D^{2}_{\underline{G}_A}(Y^A)\,R_0(Y^A,u)\right]\,du & =\int_{-\infty}^x\int_{-\infty}^{u} D^{2}_{\underline{G}_A}(u^\prime)\,dF_A\left(u^\prime| Z^A=z_A,Z^B=z_B\right) du\\
 & = \mathbbm{E}_{F_A\left(\cdot| Z^A=z_A,Z^B =z_B\right)}\left[D^{2}_{\underline{G}_A}(Y^A)\,R_1(Y^A,x)\right],
 \end{align*}
and an identical expression as the preceding one but with  $\underline{G}_A$ replaced with $\overline{G}_A$.

\par Proceeding by induction on $s\in\mathbb{Z}_+$, once we have expressions for $\underline{D}_{N_A}^{s-1}$ and $\overline{D}_{N_A}^{s-1}$, we can repeat similar steps to those for the case of $s=3$, but with appropriate adjustments, to obtain the expressions for $\underline{D}_{N_A}^{s}(x)$ and $\overline{D}_{N_A}^{s}(x)$ for each $x\in\mathcal{X}_A$. We omit the details for brevity.

\par Turning to population $B$, we only provide a brief description for deriving $\underline{D}_{N_B}^{s}$ and $\overline{D}_{N_B}^{s}$, as it follows steps identical to those for population $A$, but with a minor adjustment. The adjustment is that unlike the derivations for population $A$, we now must account for the two set of inequalities~(\ref{eq - Appendix NBD MCAR Example 0}) and~(\ref{eq - Appendix NBD MCAR Example 1}), representing the restrictions on the parameter space, which now include restrictions on wave nonresponse. For $s=1$ and $x\in\mathcal{X}_B$, we apply the Law of Total Probability to $F_{B}(x)$ and use Bayes' Theorem to obtain
\begin{align*}
F_{B}(x)=\frac{F_B\left(x|Z^A=Z^B=1\right)}{1-\text{Prob}\left(Z^A=Z^B=0| Y^B\leq x\right)-\text{Prob}\left(Z^A=1,Z^B=0| Y^B\leq x\right)}\,\delta_{11}.
\end{align*}
Now we can use the inequalities~(\ref{eq - Appendix NBD MCAR Example 0}) and~(\ref{eq - Appendix NBD MCAR Example 1}) for $K=B$ to substitute out the propensities and derive the bounds $\underline{D}_{N_B}^{1}(x)$ and
$\overline{D}_{N_B}^{1}(x)$, so that $\underline{D}_{N_B}^{1}(x)\leq D_{N_B}^{1}(x)\leq \overline{D}_{N_B}^{1}(x)$. Since $x\in\mathcal{X}_B$ was arbitrary, these inequalities hold for each $x\in\mathcal{X}_B$. Now we can integrate all sides of this inequality, to obtain the bounds $\underline{D}_{N_B}^{2}(x)$ and $\overline{D}_{N_B}^{2}(x)$ for each $x\in\mathcal{X}_B$, which like with population $A$, the derivations invovle the use of Integration by Parts and interchanging expectation and integration. Proceeding by induction on $s\in\mathbb{Z}_+$, once we have expressions for $\underline{D}_{N_B}^{s-1}$ and $\overline{D}_{N_B}^{s-1}$, we can repeat similar derivations, but with appropriate adjustments, to obtain the expressions for $\underline{D}_{N_B}^{s}(x)$ and $\overline{D}_{N_B}^{s}(x)$ for each $x\in\mathcal{X}_A$. We omit the details for brevity.
\end{proof}

\subsection{Example~\ref{Example - Kline and Santos}: KS Neighborhood of MCAR}
\par Recall that Example~\ref{Example - Kline and Santos} is based on~\cite{klinesantos2013}, who put forward a construction using the maximal Kolmogorov-Smirnov distance between the distributions of missing and observed outcomes. Their approach builds a neighborhood for the CDFs $$F_A\left(\cdot| Z^A=Z^B=0\right), F_B\left(\cdot| Z^A=Z^B=0\right),\quad\text{and}\quad F_B\left(\cdot| Z^A=1,Z^B=0\right),$$ according to the maximal Kolmogorov-Smirnov distance using the conditions in~(\ref{eq - Example MCAR-KS}). Let $\gamma_A,\gamma^{00}_B,\gamma^{10}_B\in[0,1]$ be such that
\begin{align}
\sup_{x\in\mathcal{X}_A} & \left|F_A\left(x| Z^A=Z^B=0\right)-F_A\left(x| Z^A=Z^B=1\right)\right|\leq \gamma_A,\label{eq - Appendix KS Example 0}\\
\sup_{x\in\mathcal{X}_B} & \left|F_B\left(x| Z^A=Z^B=0\right)-F_B\left(x| Z^A=Z^B=1\right)\right|\leq \gamma^{00}_B,\quad\text{and}\label{eq - Appendix KS Example 1}\\
\sup_{x\in\mathcal{X}_B} & \left|F_B\left(x| Z^A=1,Z^B=0\right)-F_B\left(x| Z^A=1,Z^B=1\right)\right|\leq \gamma^{10}_B,\label{eq - Appendix KS Example 2}
\end{align}
which are predesignated by the practitioner, where MCAR means
\begin{equation}
\label{eq - Example MCAR-KS}
\begin{aligned}
F_K\left(x| Z^A=Z^B=1\right) & =F_K\left(x| Z^A=Z^B=0\right)\quad\forall x\in\mathcal{X}_K,K=A,B,\quad\text{and}\\
F_B\left(x| Z^A=Z^B=1\right) & =F_B\left(x| Z^A=1,Z^B=0\right)\quad\forall x\in\mathcal{X}_B.
\end{aligned}
\end{equation}

\begin{proposition}\label{prop - Kline Santos ID Set}
Let $\gamma_A,\gamma^{00}_B,\gamma^{10}_B\in[0,1]$, and suppose that~(\ref{eq - Appendix KS Example 0}) - (\ref{eq - Appendix KS Example 2}) hold. For each $x\in[\underline{t},\overline{t}]$ and $s\in\mathbb{N}$, the corresponding bounds on the dominance functions are
\begin{align*}
\underline{D}_{N_A}^{s}(x) & = \mathbbm{E}_{F_A\left(x| Z^A=Z^B=1\right)}\left[\frac{(x-Y^A)^{s-1}}{(s-1)!}\,\mathbbm{1}\left[Y^A\leq x\right]\right]\,(\delta_{11}+(1-\gamma_A)\delta_{00}), \\
& \qquad +\mathbbm{E}_{F_A\left(x| Z^A=1,Z^B=0\right)}\left[\frac{(x-Y^A)^{s-1}}{(s-1)!}\,\mathbbm{1}\left[Y^A\leq x\right]\right]\,\delta_{10} \\
\overline{D}_{N_A}^{s}(x) & = \underline{D}_{N_A}^{s}(x) +\frac{(x-\underline{Y}^A)^{s-1}}{(s-1)!}\,\gamma_A \delta_{00},\\
\underline{D}_{N_B}^{s}(x) & =  \mathbbm{E}_{F_B\left(x| Z^A=Z^B=1\right)}\left[\frac{(x-Y^B)^{s-1}}{(s-1)!}\,\mathbbm{1}\left[Y^B\leq x\right]\right]\,(\delta_{11}+(1-\gamma^{00}_B)\delta_{00}+(1-\gamma^{10}_B)\delta_{10}),\quad\text{and}\\
\overline{D}_{N_B}^{s}(x) & =  \underline{D}_{N_B}^{s}(x)+ \frac{(x-\underline{Y}^B)^{s-1}}{(s-1)!}\,(\gamma^{00}_B \delta_{00}+\gamma^{10}_B\delta_{10}),
\end{align*}
where $\underline{Y}^K=\inf\mathcal{X}_K$, for $K=A,B$.
\end{proposition}
\begin{proof}
The proof proceeds by the direct method. Throughout, we adopt the convention $0\times\pm\infty=0$. For $s=1$, the restrictions~(\ref{eq - Appendix KS Example 0}) - (\ref{eq - Appendix KS Example 2}) yields bounds via the decomposition having the form
\begin{align*}
F_A\left(x| Z^A=Z^B=0\right) & =(1-\gamma_A)F_A\left(x| Z^A=Z^B=1\right)+\gamma_A G_A(x)\\
F_B\left(x| Z^A=Z^B=0\right) & =(1-\gamma^{00}_B)F_B\left(x| Z^A=Z^B=1\right)+\gamma^{00}_B G^{00}_B(x) \\
F_B\left(x| Z^A=1,Z^B=0\right) & =(1-\gamma^{10}_B)F_B\left(x| Z^A=Z^B=1\right)+\gamma^{10}_B G^{10}_B(x),
\end{align*}
for each $x$, where $G_A, G^{00}_B,G^{10}_B$ are unknown CDFs. Focusing on population $A$, let $x\in[\underline{t},\overline{t}]$, the Law of Total Probability applied to $F_A(x)$ yields
\begin{align*}
F_A(x) & =F_A\left(x|Z^A=Z^B=1\right)\,\delta_{11}+ F_A\left(x|Z^A=1,Z^B=0\right)\,\delta_{10}+F_A\left(x|Z^A=Z^B=0\right)\,\delta_{00}\\
 & = F_A\left(x|Z^A=Z^B=1\right)\,\delta_{11}+ F_A\left(x|Z^A=1,Z^B=0\right)\,\delta_{10}\\
 & \qquad +\delta_{00}\,(1-\gamma_A)F_A\left(x| Z^A=Z^B=1\right)+\delta_{00}\,\gamma_A G_A(x)\\
 & = \left(\delta_{00}\,(1-\gamma_A)+\delta_{11}\right)F_A\left(x| Z^A=Z^B=1\right)+F_A\left(x|Z^A=1,Z^B=0\right)\,\delta_{10}+\delta_{00}\,\gamma_A G_A(x).
\end{align*}
Now since $G_A$ is a CDF, $G_A(x)\in[0,1]$, holds, implying the worst-case CDF bounds $\underline{F}_A(x)\leq F_A(x)\leq\overline{F}_A(x)$, where
\begin{align*}
\overline{F}_A(x) & =\begin{cases}
  \delta_{00}\,\gamma_A  &  x= \underline{Y}^A \\
 \left(\delta_{00}\,(1-\gamma_A)+\delta_{11}\right)F_A\left(x| Z^A=Z^B=1\right)+F_A\left(x|Z^A=1,Z^B=0\right)\,\delta_{10}+\delta_{00}\,\gamma_A & x> \underline{Y}^A\\
  0 & \text{otherwise}
\end{cases}  \\
\underline{F}_A(x) & =\begin{cases}
  \left(\delta_{00}\,(1-\gamma_A)+\delta_{11}\right)F_A\left(x| Z^A=Z^B=1\right)+F_A\left(x|Z^A=1,Z^B=0\right)\,\delta_{10}&  \underline{Y}^A\leq x<\overline{Y}^A \\
 1 &  x= \overline{Y}^A \\
 0 & \text{otherwise}.
\end{cases}
\end{align*}
These bounds are sharp. Here, we have $\overline{D}_{N_A}^{1}(x)=\overline{F}_A(x)$ and $\underline{D}_{N_A}^{1}(x)=\underline{F}_A(x)$. As $x\in\mathcal{X}_A$ was arbitrary, we done with the case $s=1$. For $s=2$, we integrate all sides of $\underline{D}_{N_A}^{1}(x)\leq F_A(x)\leq\overline{D}_{N_A}^{1}(x)$ as such
\begin{align*}
\int_{-\infty}^x\underline{D}_{N_A}^{1}(u)\,du\leq \int_{-\infty}^x F_A(u)\,du\leq\int_{-\infty}^x\overline{D}_{N_A}^{1}(u)\,du,
\end{align*}
and recognizing that $\underline{D}_{N_A}^{2}(x)=\int_{-\infty}^x\underline{D}_{N_A}^{1}(u)\,du$, $D_{N_A}^{2}(x)=\int_{-\infty}^x F_A(u)\,du$, and $\overline{D}_{N_A}^{2}(x)=\int_{-\infty}^x\overline{D}_{N_A}^{1}(u)\,du$, we must obtain the forms of $\int_{-\infty}^x\underline{D}_{N_A}^{1}(u)\,du$ and $\int_{-\infty}^x\overline{D}_{N_A}^{1}(u)\,du$  to complete the proof in this case. Towards that end, using integration by parts
\begin{align*}
\int_{-\infty}^x\underline{D}_{N_A}^{1}(u)\,du & = \underline{F}_A(x)x-(0\times-\infty)-\int_{-\infty}^x u\,d\underline{F}_A(u)\\
& = \underline{F}_A(x)x-\int_{-\infty}^x u\,d\underline{F}_A(u)\\
 & =\int_{\mathbb{R}} (x-u)\mathbbm{1}\left[u\leq x\right]\,d\underline{F}_A(u) \\
& =\left(\delta_{00}\,(1-\gamma_A)+\delta_{11}\right)\mathbbm{E}_{F_A\left(x| Z^A=Z^B=1\right)}\left[(x-Y^A)\,\mathbbm{1}\left[Y^A\leq x\right]\right]\\
& \quad +\delta_{10}\mathbbm{E}_{F_A\left(x| Z^A=1,Z^B=0\right)}\left[(x-Y^A)\,\mathbbm{1}\left[Y^A\leq x\right]\right],
\end{align*}
and
\begin{align*}
\int_{-\infty}^x\overline{D}_{N_A}^{1}(u)\,du & =  \overline{F}_A(x)x-(0\times-\infty)-\int_{-\infty}^x u\,d\overline{F}_A(u)\\
& =  \overline{F}_A(x)x-\int_{-\infty}^x u\,d\overline{F}_A(u)\\
& = \int_{\mathbb{R}} (x-u)\mathbbm{1}\left[u\leq x\right]\,d\overline{F}_A(u) \\
& =\left(\delta_{00}\,(1-\gamma_A)+\delta_{11}\right)\int_{-\infty}^x F_A\left(u| Z^A=Z^B=1\right)\,du\\
& =\left(\delta_{00}\,(1-\gamma_A)+\delta_{11}\right)\mathbbm{E}_{F_A\left(x| Z^A=Z^B=1\right)}\left[(x-Y^A)\,\mathbbm{1}\left[Y^A\leq x\right]\right]\\
& \quad +\delta_{10}\mathbbm{E}_{F_A\left(x| Z^A=1,Z^B=0\right)}\left[(x-Y^A)\,\mathbbm{1}\left[Y^A\leq x\right]\right]+(x-\underline{Y}^A)\gamma_A \delta_{00}
\end{align*}
Now for $s=3$, integrate all sides of $\underline{D}_{N_A}^{2}(x)\leq D_{N_A}^{2}(x)\leq\overline{D}_{N_A}^{2}(x)$ as such
\begin{align*}
\int_{-\infty}^x\underline{D}_{N_A}^{2}(u)\,du\leq \int_{-\infty}^x D_{N_A}^{2}(u)\,du\leq\int_{-\infty}^x\overline{D}_{N_A}^{2}(u)\,du,
\end{align*}
and recognizing that $\underline{D}_{N_A}^{3}(x)=\int_{-\infty}^x\underline{D}_{N_A}^{2}(u)\,du$, $D_{N_A}^{3}(x)=\int_{-\infty}^x D_{N_A}^{2}(u)\,du$, and $\overline{D}_{N_A}^{3}(x)=\int_{-\infty}^x\overline{D}_{N_A}^{2}(u)\,du$, we must obtain the forms of $\int_{-\infty}^x\underline{D}_{N_A}^{2}(u)\,du$ and $\int_{-\infty}^x\overline{D}_{N_A}^{2}(u)\,du$  to complete the proof in this case. We have
\begin{align*}
\int_{-\infty}^x\underline{D}_{N_A}^{2}(u)\,du & =\int_{-\infty}^x \int_{\mathbb{R}} (u^\prime-u)\mathbbm{1}\left[u\leq u^\prime\right]\,d\underline{F}_A(u)\,du^\prime \\
& = \int_{\mathbb{R}} \int_{-\infty}^x(u^\prime-u)\mathbbm{1}\left[u\leq u^\prime\right]\,du^\prime\, d\underline{F}_A(u)\\
& = \int_{\mathbb{R}} \frac{(x-u)^{2}}{2!}\,d\underline{F}_A(u)\\
& = \mathbbm{E}_{F_A\left(x| Z^A=Z^B=1\right)}\left[\frac{(x-Y^A)^{2}}{2}\,\mathbbm{1}\left[Y^A\leq x\right]\right]\,(\delta_{11}+(1-\gamma_A)\delta_{00}), \\
& \qquad +\mathbbm{E}_{F_A\left(x| Z^A=1,Z^B=0\right)}\left[\frac{(x-Y^A)^{2}}{2}\,\mathbbm{1}\left[Y^A\leq x\right]\right]\,\delta_{10}
\end{align*}
and
\begin{align*}
\int_{-\infty}^x\overline{D}_{N_A}^{2}(u)\,du & =\int_{-\infty}^x \int_{\mathbb{R}} (u^\prime-u)\mathbbm{1}\left[u\leq u^\prime\right]\,d\overline{F}_A(u)\,du^\prime \\
& = \int_{\mathbb{R}} \int_{-\infty}^x(u^\prime-u)\mathbbm{1}\left[u\leq u^\prime\right]\,du^\prime\, d\overline{F}_A(u)\\
& = \int_{\mathbb{R}} \frac{(x-u)^{2}}{2!}\,d\overline{F}_A(u)\\
& = \int_{-\infty}^x\underline{D}_{N_A}^{2}(u)\,du+ \frac{(x-\underline{Y}^A)^{2}}{2}\,\gamma_A \delta_{00}.
\end{align*}
Proceeding by induction on $s\in\mathbb{Z}_+$, once we have expressions for $\underline{D}_{N_A}^{s-1}$ and $\overline{D}_{N_A}^{s-1}$, we can repeat similar steps to those for the case of $s=3$, but with appropriate adjustments, to obtain the expressions for $\underline{D}_{N_A}^{s}(x)$ and $\overline{D}_{N_A}^{s}(x)$ for each $x\in\mathcal{X}_A$. We omit the details for brevity.

\par Turning to population $B$, let $x\in[\underline{t},\overline{t}]$, the Law of Total Probability applied to $F_B(x)$ yields
\begin{align*}
F_B(x) & =F_B\left(x|Z^A=Z^B=1\right)\,\delta_{11}+ F_B\left(x|Z^A=1,Z^B=0\right)\,\delta_{10}+F_B\left(x|Z^A=Z^B=0\right)\,\delta_{00}\\
 & = F_B\left(x|Z^A=Z^B=1\right)\,\delta_{11}+ \left((1-\gamma^{10}_B)F_B\left(x| Z^A=Z^B=1\right)+\gamma^{10}_B G^{10}_B(x)\right)\,\delta_{10}\\
 & \qquad +\delta_{00}\,(1-\gamma^{00}_B)F_B\left(x| Z^A=Z^B=1\right)+\delta_{00}\,\gamma^{00}_B G^{00}_B(x) \\
 & = \left(\delta_{00}\,(1-\gamma^{00}_B)+\delta_{10}\,(1-\gamma^{10}_B)+\delta_{11}\right)F_B\left(x| Z^A=Z^B=1\right) +\delta_{00}\,\gamma^{00}_B G^{00}_B(x)\\
 & \qquad+\delta_{10}\,\gamma^{10}_B G^{10}_B(x).
\end{align*}
Now since $G^{00}_B$ and $G^{10}_B$ are CDFs, $G^{00}_B(x),G^{00}_B(x)\in[0,1]$, holds, implying the worst-case CDF bounds $\underline{F}_B(x)\leq F_B(x)\leq\overline{F}_B(x)$, where
\begin{align*}
\overline{F}_B(x) & =\begin{cases}
  \delta_{10}\,\gamma^{10}_B+\delta_{00}\,\gamma^{00}_B   &  x= \underline{Y}^B \\
 \left(\delta_{00}\,(1-\gamma^{00}_B)+\delta_{10}\,(1-\gamma^{10}_B)+\delta_{11}\right)F_B\left(x| Z^A=Z^B=1\right)+\delta_{10}\,\gamma^{10}_B+\delta_{00}\,\gamma^{00}_B  & x> \underline{Y}^B\\
  0 & \text{otherwise}
\end{cases}  \\
\underline{F}_B(x) & =\begin{cases}
  \left(\delta_{00}\,(1-\gamma^{00}_B)+\delta_{10}\,(1-\gamma^{10}_B)+\delta_{11}\right)F_B\left(x| Z^A=Z^B=1\right)&  \underline{Y}^B\leq x<\overline{Y}^B \\
 1 &  x= \overline{Y}^B \\
 0 & \text{otherwise}.
\end{cases}
\end{align*}
To derive the bounds, we follows steps identical to those for $K=A$ above, except that we now replace $\underline{F}_A$ and $\overline{F}_A$ with $\underline{F}_B$ and $\overline{F}_B$, respectively, and calculate the integrals accordingly. We omit the details for brevity.
\end{proof}

\section{Technical Details for Section~\ref{Section - Discussion}}\label{Appendix E}
\subsection{Examples of Bounds for The Testing Problem~(\ref{eq - Testing Problem B Dom A})}\label{Subsection - Appendix D - B Dom A}
\begin{example}\label{Example - WC B dom A}
{\bf WC Bounds}. The WC bounds can be recovered using the following specification: for each $x\in[\underline{t},\overline{t}]$,
\begin{align*}
\varphi_{1,N_A,N_B}(x) & =\varphi_{4,N_A,N_B}(x)=\varphi_{2,N_A,N_B}(x)=-(\delta_{11}+\delta_{10})\quad\text{and}\\
\varphi_{3,N_A,N_B}(x) & =\frac{(x-\underline{Y}^B)^{s-1}}{(s-1)!}\,\delta_{00}+\delta_{10},
\end{align*}
%
where $\underline{Y}^B=\inf\mathcal{X}_B$. Hence, by definition, $\underline{Y}^B\leq x$ for each $x\in[\underline{t},\overline{t}]$. It defines $\theta_{N_A,N_B}(\cdot;\varphi_{N_A,N_B})=\overline{D}_{N_B}^{s}(\cdot)-\underline{D}_{N_A}^{s}(\cdot)$, where for each $x\in[\underline{t},\overline{t}]$,
\begin{align*}
\overline{D}_{N_B}^{s}(x) & = \mathbbm{E}_{F_B\left(x| Z^A=Z^B=1\right)}\left[\frac{(x-Y^B)^{s-1}}{(s-1)!}\,\mathbbm{1}\left[Y^B\leq x\right]\right]\,\delta_{11}+\frac{(x-\underline{Y}^B)^{s-1}}{(s-1)!}\,(\delta_{00}+\delta_{10})\quad\text{and}\\
\underline{D}_{N_A}^{s}(x) & =  \mathbbm{E}_{F_A\left(x| Z^A=Z^B=1\right)}\left[\frac{(x-Y^A)^{s-1}}{(s-1)!}\,\mathbbm{1}\left[Y^A\leq x\right]\right]\,\delta_{11}\\
& \qquad +\mathbbm{E}_{F_A\left(x| Z^A=1,Z^B=0\right)}\left[\frac{(x-Y^A)^{s-1}}{(s-1)!}\,\mathbbm{1}\left[Y^A\leq x\right]\right]\,\delta_{10}.
\end{align*}
\end{example}

\begin{example}
{\bf MCAR for Unit Nonresponse}. In conjunction with the WC upper bound for $F_B(\cdot|Z^A=1,Z^B=0)$, the following specification encodes this MCAR assumption: for each $x\in [\underline{t},\overline{t}]$,
\begin{align*}
\varphi_{1,N_A,N_B}(x) & =\varphi_{4,N_A,N_B}(x)=-\frac{(\delta_{11}+\delta_{00})(\delta_{11}+\delta_{10})}{\delta_{11}},\,\varphi_{2,N_A,N_B}(x) =-(\delta_{11}+\delta_{10})\quad\text{and}\\
\varphi_{3,N_A,N_B}(x)& = \frac{(x-\underline{Y}^B)^{s-1}}{(s-1)!}\,\delta_{10}.
\end{align*}

It defines $\theta_{N_A,N_B}(\cdot;\varphi_{N_A,N_B})=\overline{D}_{N_B}^{s}(\cdot)-\underline{D}_{N_A}^{s}(\cdot)$, where for each $x\in[\underline{t},\overline{t}]$,
\begin{align*}
\overline{D}_{N_B}^{s}(x) & = \mathbbm{E}_{F_B\left(x| Z^A=Z^B=1\right)}\left[\frac{(x-Y^B)^{s-1}}{(s-1)!}\,\mathbbm{1}\left[Y^B\leq x\right]\right]\,(\delta_{11}+\delta_{00})+\frac{(x-\underline{Y}^B)^{s-1}}{(s-1)!}\,\delta_{10}\quad\text{and}\\
\underline{D}_{N_A}^{s}(x) & =  \mathbbm{E}_{F_A\left(x| Z^A=Z^B=1\right)}\left[\frac{(x-Y^A)^{s-1}}{(s-1)!}\,\mathbbm{1}\left[Y^A\leq x\right]\right]\,(\delta_{11}+\delta_{00})\\
& \qquad +\mathbbm{E}_{F_A\left(x| Z^A=1,Z^B=0\right)}\left[\frac{(x-Y^A)^{s-1}}{(s-1)!}\,\mathbbm{1}\left[Y^A\leq x\right]\right]\,\delta_{10}.
\end{align*}
\end{example}

\begin{example}
{\bf Neighborhood Assumption for Unit and Wave Nonresponse}. The following specification of $\varphi_{N_A,N_B}$ delivers the desired bounds: for each $x\in[\underline{t},\overline{t}]$

\begin{align*}
\varphi_1(x) & =\varphi_2(x)=-D^{s}_{\underline{G}_A}(x),\,\varphi_4(x) =- \frac{1-\delta_{00}}{\delta_{11}}\,D^{s}_{\overline{G}_B}(x)\quad\text{and}\\
\varphi_3(x) & =-\mathbbm{1}[s>1]\sum_{j=0}^{s-2}\mathbbm{E}_{F_B\left(\cdot| Z^A=Z^B=1\right)}\left[D^{s}_{\overline{G}_B}(Y^B)\,R_j(Y^B,x)\right]\\
 & \qquad+\mathbbm{1}[s>1]\frac{\delta_{11}}{1-\delta_{00}}\sum_{j=0}^{s-2}\mathbbm{E}_{F_A\left(\cdot| Z^A=Z^B=1\right)}\left[D^{s}_{\underline{G}_A}(Y^A)\,R_j(Y^A,x)\right]\\
& \qquad+\mathbbm{1}[s>1]\frac{\delta_{10}}{1-\delta_{00}}\sum_{j=0}^{s-2}\mathbbm{E}_{F_A\left(\cdot| Z^A=1,Z^B=0\right)}\left[D^{s}_{\underline{G}_A}(Y^A)\,R_j(Y^A,x)\right].
\end{align*}
\noindent This specification defines $\theta_{N_A,N_B}(\cdot;\varphi_{N_A,N_B})=\overline{D}_{N_B}^{s}(\cdot)-\underline{D}_{N_A}^{s}(\cdot)$, where for each $x\in[\underline{t},\overline{t}]$,
\begin{align*}
\underline{D}_{N_A}^{s}(x) & =\left(\frac{\delta_{11}}{1-\delta_{00}}\,F_A\left(x| Z^A=Z^B=1\right)+\frac{\delta_{10}}{1-\delta_{00}}\,F_A\left(x| Z^A=1,Z^B=0\right)\right)D^{s}_{\underline{G}_A}(x)\\
& \qquad-\mathbbm{1}[s>1]\frac{\delta_{11}}{1-\delta_{00}}\sum_{j=0}^{s-2}\mathbbm{E}_{F_A\left(\cdot| Z^A=Z^B=1\right)}\left[D^{s}_{\underline{G}_A}(Y^A)\,R_j(Y^A,x)\right]\\
& \qquad-\mathbbm{1}[s>1]\frac{\delta_{10}}{1-\delta_{00}}\sum_{j=0}^{s-2}\mathbbm{E}_{F_A\left(\cdot| Z^A=1,Z^B=0\right)}\left[D^{s}_{\underline{G}_A}(Y^A)\,R_j(Y^A,x)\right]\\
\overline{D}_{N_B}^{s}(x) & = F_B\left(x| Z^A=Z^B=1\right)\,D^{s}_{\overline{G}_B}(x)-\mathbbm{1}[s>1]\sum_{j=0}^{s-2}\mathbbm{E}_{F_B\left(\cdot| Z^A=Z^B=1\right)}\left[D^{s}_{\overline{G}_B}(Y^B)\,R_j(Y^B,x)\right].
\end{align*}
\end{example}

\begin{example}
{\bf~\cite{klinesantos2013}}. In conjunction with the WC bounds on wave nonresponse, the following specification of $\varphi_{N_A,N_B}$ encodes this assumption on unit and wave nonrepsonse: for each $x\in[\underline{t},\overline{t}]$,
\begin{align*}
\varphi_1(x) & =-\frac{(\delta_{11}+(1-\gamma_A)\delta_{00})(\delta_{11}+\delta_{10})}{\delta_{11}},\;\varphi_2(x)=-(\delta_{11}+\delta_{10}),\,\varphi_3(x)=\frac{(x-\underline{Y}^B)^{s-1}}{(s-1)!}\,(\gamma^{00}_B \delta_{00}+\gamma^{10}_B\delta_{10}) \\
& \quad\text{and}\quad\varphi_4(x)=-\frac{(\delta_{11}+(1-\gamma^{00}_B)\delta_{00}+(1-\gamma^{10}_B)\delta_{10})(\delta_{11}+\delta_{10})}{\delta_{11}}.
\end{align*}
It defines $\theta_{N_A,N_B}(\cdot;\varphi_{N_A,N_B})=\overline{D}_{N_B}^{s}(\cdot)-\underline{D}_{N_A}^{s}(\cdot)$, where for each $x\in[\underline{t},\overline{t}]$,
\begin{align*}
\underline{D}_{N_A}^{s}(x)  & = \mathbbm{E}_{F_A\left(x| Z^A=Z^B=1\right)}\left[\frac{(x-Y^A)^{s-1}}{(s-1)!}\,\mathbbm{1}\left[Y^A\leq x\right]\right]\,(\delta_{11}+(1-\gamma_A)\delta_{00}), \\
& \qquad +\mathbbm{E}_{F_A\left(x| Z^A=1,Z^B=0\right)}\left[\frac{(x-Y^A)^{s-1}}{(s-1)!}\,\mathbbm{1}\left[Y^A\leq x\right]\right]\,\delta_{10} \quad\text{and}\\
\overline{D}_{N_B}^{s}(x) & =   \underline{D}_{N_B}^{s}(x)+ \frac{(x-\underline{Y}^B)^{s-1}}{(s-1)!}\,(\gamma^{00}_B \delta_{00}+\gamma^{10}_B\delta_{10})
\end{align*}
where $\underline{D}_{N_B}^{s}(x) =  \mathbbm{E}_{F_B\left(x| Z^A=Z^B=1\right)}\left[\frac{(x-Y^B)^{s-1}}{(s-1)!}\,\mathbbm{1}\left[Y^B\leq x\right]\right]\,(\delta_{11}+(1-\gamma^{00}_B)\delta_{00}+(1-\gamma^{10}_B)\delta_{10})$.
\end{example}

\subsection{Section~\ref{Subsection - Discussion - Longitudinal weights}}\label{Appendix E Longitudinal weights}

\par This section provides the details on the statistical procedure for the testing problem~(\ref{eq - test problem Longitudinal}). It is based on calibration via pseudo-empirical likelihood as in~\cite{wurao2006}. When $A$ represents a wave beyond the first one, the event $\{Z_i^{A}=0, Z_i^{B}=1\}$ can occur, which was not possible when $A$ represented the population of the first wave. In words, this event corresponds to the sampled unit $i$ not responding in wave $A$ but responding in wave $B$. The event $\{Z_i^{A}=0, Z_i^{B}=0\}$ in the setup of this section now corresponds to unit $i$ either being a unit nonresponder, or a wave nonresponder that did not respond in waves $A$ and $B$ but responded in other waves. This distinction between different types of nonresponse within this event can be meaningful when their characteristics differ (e.g., poor and wealthy individuals). The foregoing general estimating function approach that targets the estimand $\overline{D}_{N_A}^{s}(\cdot)-\underline{D}_{N_B}^{s}(\cdot)$ must be adjusted now to incorporate events of the form $\{Z_i^{A}=0, Z_i^{B}=1\}$. For each $x\in[\underline{t},\overline{t}]$, consider the following estimating function $\psi_s\left(Y^A, Y^B,Z^A,Z^B,\theta(x),\varphi(x)\right)$, given by
\begin{equation}
\begin{aligned}
&\frac{(x-Y^A)^{s-1}}{(s-1)!} \left[\mathbbm{1}\left[Y^A\leq x,Z^A=Z^B=1\right]\varphi_1(x)+\mathbbm{1}\left[Y^A\leq x,Z^A=1,Z^B=0\right]\varphi_2(x)\right]\\
 &-\frac{(x-Y^B)^{s-1}}{(s-1)!}\left[\mathbbm{1}\left[Y^B\leq x,Z^A=Z^B=1\right]\varphi_4(x)+\mathbbm{1}\left[Y^B\leq x,Z^A=0,Z^B=1\right]\varphi_5(x)\right]\\
 &\quad+\varphi_3(x)-\theta(x)
\end{aligned}
 \label{eq - Estimating Function beyond wave 1}
\end{equation}
where $\varphi=[\varphi_1,\varphi_2,\varphi_3,\varphi_4,\varphi_5]$ is a vector of nuisance functionals. The nuisance functions, as in the previous sections of this paper, serves as a vehicle to encode side information of maintained assumptions on nonresponse.

\par For a given value $\varphi=\varphi_{N_A,N_B}$, the solution of the census estimating equations
\begin{align*}
N_A^{-1}N_B^{-1}\sum_{i=1}^{N_A}\sum_{j=1}^{N_B}\psi_s\left(Y_i^A, Y_j^B,Z_i^A,Z_j^B,\theta(x),\varphi(x)\right)=0\quad\forall x\in [\underline{t},\overline{t}]
\end{align*}
is
\begin{align*}
\theta(x) & =\mathbbm{E}_{F_A\left(\cdot| Z^A=Z^B=1\right)}\left[\frac{(x-Y^A)^{s-1}}{(s-1)!}\,\mathbbm{1}\left[Y^A\leq x\right]\right]\delta_{11}\,\varphi_{1,N_A,N_B}(x)\\
 & +\mathbbm{E}_{F_A\left(\cdot| Z^A=1,Z^B=0\right)}\left[\frac{(x-Y^A)^{s-1}}{(s-1)!}\,\mathbbm{1}\left[Y^A\leq x\right]\right]\delta_{10}\,\varphi_{2,N_A,N_B}(x)\\
 & - \mathbbm{E}_{F_B\left(\cdot| Z^A=Z^B=1\right)}\left[\frac{(x-Y^B)^{s-1}}{(s-1)!}\,\mathbbm{1}\left[Y^B\leq x\right]\right]\delta_{11}\,\varphi_{4,N_A,N_B}(x)\\
& - \mathbbm{E}_{F_B\left(\cdot| Z^A=0,Z^B=1\right)}\left[\frac{(x-Y^B)^{s-1}}{(s-1)!}\,\mathbbm{1}\left[Y^B\leq x\right]\right]\delta_{01}\,\varphi_{5,N_A,N_B}(x)\\
& + \varphi_3(x),
\end{align*}
for each $x\in [\underline{t},\overline{t}]$.

\par Now we introduce notation regarding the sampled units. Let $V\subset \{1,2,\ldots,N_{W_1}\}$ denote the survey sample of the target population in the first wave of the panel, and let $V^\prime=\left\{i\in V: Z^{W_1}_i=1\right\}$, which consists of units that have responded in the first wave. As those units are followed over time, define
\begin{align}
U=\left\{i\in V^{\prime}: \text{either}\,\left\{Z_i^A=Z_i^B=1\right\},\,\text{or}\,\left\{Z_i^A=1,Z_i^B=0\right\},\,\text{or}\,\left\{Z_i^A=0,Z_i^B=1\right\}\right\},
\end{align}
to be the set of units that are followed form the first wave that have responded in either of the waves corresponding to $A$ and $B$. Also, let $\Psi_i (x;\hat{\varphi}(x))=\psi_s\left(Y_i^A, Y_i^B,Z_i^A,Z_i^B,0,\hat{\varphi}(x)\right)$ for each $i\in U$, where $\hat{\varphi}$ is a plug-in estimator of $\varphi_{N_A,N_B}$ than can depend on $V$. Furthermore, define
the moment functions
\begin{align}
\Gamma_{1,i}=G_0-G_i\mathbbm{1}\left[Z^{A}_i=1\right]-\overline{G}\mathbbm{1}\left[Z^{A}_i=0\right]\quad\text{and}\quad \Gamma_{2,i}=G_i\mathbbm{1}\left[Z^{A}_i=1\right]+\underline{G}\mathbbm{1}\left[Z^{A}_i=0\right]-G_0
\end{align}
for each $i\in \{1,2,\ldots,N_A\}$. Finally, let $\Gamma_i=[\Gamma_{1,i},\Gamma_{2,i}]^\intercal$, and let $\Gamma_{i,b}$ denote the subvector of $\Gamma_i$ corresponding to the active/binding inequalities in~(\ref{eq - WC bounds Aux})

\par Accordingly, the following pseudo-empirical likelihood statistic defines a \emph{restricted} testing procedure for the test problem~(\ref{eq - test problem Longitudinal}):
\begin{align}
LR_{\text{aux}}=\begin{cases}
    2\min\limits_{x \in [\underline{t},\overline{t}]}\left(L_{UR,\text{aux}}(x)-L_{R,\text{aux}}(x)\right)/\widehat{\text{Deff}}_{\text{GR}}(x) & \text{if} \ \tilde{\theta}(x;\hat{\varphi}(x))<0 \quad \forall x \in [\underline{t},\overline{t}] \\
    0, & \text{otherwise}
    \end{cases}
\end{align}\
where
\begin{equation}\label{eq: PELF- aux H0}
\begin{aligned}
L_{R,\text{aux}}(x) =\max_{\Vec{p} \in (0,1]^{n}}  \sum_{i\in U} W^\prime_i \log(p_i) &\quad \text{s.t.}\quad p_i>0 \quad \forall i,\;\sum_{i\in U} W^\prime_i p_i=1,\,\sum_{i\in U} W^\prime_i p_i \Psi_i (x;\hat{\varphi}(x)) = 0\\
&\quad \text{and}\, \sum_{i\in U} W^\prime_i p_i \Gamma_{\ell,i} \leq 0\quad \ell=1,2;
\end{aligned}
\end{equation}

\begin{equation}
\begin{aligned}
L_{UR,\text{aux}}(x)=\max_{\Vec{p} \in (0,1]^{n}} \sum_{i\in U} W^\prime_i \log(p_i) & \quad \text{s.t.}\quad p_i>0 \quad \forall i,\;\sum_{i\in U} W^\prime_i p_i=1,\\
&\quad \text{and}\, \sum_{i\in U} W^\prime_i p_i \Gamma_{\ell,i} \leq 0\quad \ell=1,2,
\end{aligned}
\label{eq: PELF- aux H1}
\end{equation}
 and $\tilde{\theta}(x;\hat{\varphi})=\sum_{i\in U} W^\prime_i \tilde{p}_i \Psi_i (x;\hat{\varphi}(x))$ with $\{\tilde{p}_i,i\in U\}$ the solution of~(\ref{eq: PELF- aux H1}). The design-effect in this scenario is associated with the estimator $\acute{\theta}(x;\hat{\varphi}(x))=\sum_{i\in U} (W^\prime_i/n) r_i(x)$,
\begin{align}
\widehat{\text{Deff}}_{\text{GR}}(x)= \left[n^{-1}\sum_{i\in U}(W^\prime_i/n)\Psi^2_i (x;\hat{\varphi}(x))\right]^{-1}\,\widehat{Var}\left(\acute{\theta}(x;\hat{\varphi}(x)) \right),
\end{align}
 where for each $i\in U$
 \begin{align}
r_i(x)& =\Psi_i (x;\hat{\varphi}(x))-{\bf B}^\intercal(x)\Gamma_{i,b}\quad\text{and}\\
 {\bf B}(x) & =\left[\frac{1}{N_A}\sum_{i=1}^{N_A}\Gamma_{i,b}\,\Gamma_{i,b}^\intercal \right]^{-1}\left[\frac{1}{N_A}\frac{1}{N_B}\sum_{i=1}^{N_A}\sum_{j=1}^{N_B}\Gamma_{i,b}\,\psi_s\left(Y_i^A, Y_j^B,Z_i^A,Z_j^B,0,\varphi_{N_A,N_B}(x)\right)\right],
 \end{align}
and $\widehat{Var}\left(\acute{\theta}(x;\hat{\varphi}(x)) \right)$ is an estimator of the design variance $Var\left(\acute{\theta}(x;\hat{\varphi}(x)) \right)$. In estimation of this design-variance, one must estimate ${\bf B}$ using its sample version:
\begin{align*}
{\bf \hat{B}}(x) & =\left[\sum_{i\in U} (W^\prime_i/n)\hat{\Gamma}_{i,b}\,\hat{\Gamma}_{i,b}^\intercal \right]^{-1}\left[\sum_{i\in U} (W^\prime_i/n)\hat{\Gamma}_{i,b}\,\Psi_i (x;\hat{\varphi}(x))\right],
\end{align*}
where $\hat{\Gamma}_{i,b}$ is exactly $\Gamma_{i,b}$ but based on an estimator of the set of active/binding inequalities.

\par Let $\nu=[\nu_{eq},\nu_b]\in \mathbb{R}\times\mathbb{R}_{-}^{\text{dim}(\Gamma_{i,b})}$. For a fixed nominal level $\alpha \in (0,1)$, the decision rule of the test is to
\begin{align}\label{eq - decision rule Restricted test}
    \text{reject} \ H_0^3 \iff LR_{\text{aux}} > c(\alpha),
\end{align}
where $c(\alpha)$ is be the $1-\alpha$ quantile of the distribution for the random variable
\begin{align}\label{eq - Appendix Longitud Asy null Dist}
\inf_{\nu\in\mathbb{R}\times\mathbb{R}_{-}^{\text{dim}(\Gamma_{i,b})}:\nu_{eq}=0}\left(\nu-T\right)^\intercal\Omega^{-1}\left(\nu-T\right)-\inf_{\nu\in\mathbb{R}\times\mathbb{R}_{-}^{\text{dim}(\Gamma_{i,b})}}\left(\nu-T\right)^\intercal\Omega^{-1}\left(\nu-T\right),
\end{align}
where $T\sim \text{MVN}\left({\bf 0},\Omega\right)$, and $\Omega$ is the limit of the covariance matrix
\begin{align*}
\frac{1}{N_A}\frac{1}{N_B}\sum_{i=1}^{N_A}\sum_{j=1}^{N_B}\begin{bmatrix} \psi_s^2\left(Y_i^A, Y_j^B,Z_i^A,Z_j^B,0,\varphi_{N_A,N_B}(x_0)\right) & \Gamma_{i,b}^\intercal\,\psi_s\left(Y_i^A, Y_j^B,Z_i^A,Z_j^B,0,\varphi_{N_A,N_B}(x_0)\right) \\
 \Gamma_{i,b}\psi_s\left(Y_i^A, Y_j^B,Z_i^A,Z_j^B,0,\varphi_{N_A,N_B}(x_0)\right) & \Gamma_{i,b}\,\Gamma_{i,b}^\intercal
\end{bmatrix}
\end{align*}
as $N_A,N_B\rightarrow+\infty$, $x_0=\min\{x\in[\underline{t},\overline{t}]:C(x)=0\}$ with $C(\cdot)$ defined in~(\ref{C function}).

\par Some remarks are in order.

\begin{remark}
The reason the set of active/binding inequalities in~(\ref{eq - WC bounds Aux}) enters the design effect is because they would enter the Taylor expansion of $L_{UR,\text{aux}}(x)-L_{R,\text{aux}}(x)$, and the inactive/slack inequalities do not because their Lagrange multipliers would be zero on account of complementary slackness.
\end{remark}

\begin{remark}
The above procedure relies on consistent estimation of the set of active/binding inequalities in~(\ref{eq - WC bounds Aux}) through $\{\hat{\Gamma}_{i,b},i\in U\}$. Estimation of this set can be implemented, for example, using a generalized moment selection procedure as in~\cite{Andrews-Soares2010}.
\end{remark}

\begin{remark}
This procedure uses an asymptotic critical value, where justification of the asymptotic form of $LR_{\text{aux}}$, given by~(\ref{eq - Appendix Longitud Asy null Dist}), is through an application of a central limit theorem. This asymptotic form of the test statistic is for a cone-based testing problem on a multivariate normal mean vector $\nu$:
\begin{align}
\mathcal{H}_0:\nu\in\mathbb{R}\times\mathbb{R}_{-}^{\text{dim}(\Gamma_{i,b})}:\nu_{eq}=0\quad\text{versus}\quad\mathcal{H}_1:\nu\in\mathbb{R}\times\mathbb{R}_{-}^{\text{dim}(\Gamma_{i,b})}.
\end{align}
\end{remark}

\end{document}